\documentclass{fundam}
 
\publyear{26}
\papernumber{4}
\volume{196}
\issue{1}
\theDOI{10.46298/fi.13767}

\versionForARXIV

\usepackage[ruled,lined]{algorithm2e}
\usepackage{graphicx}

\usepackage{pdfsync} 
\usepackage{amsfonts}
\usepackage[english]{babel} 
\usepackage{mymacros}
\usepackage{color}
\usepackage{subcaption}

\usepackage{hyperref}

\begin{document}

\title{On Reduction and Synthesis of  Petri's Cycloids}

\address{ruediger.valk@uni-hamburg.de}

\author{R\"udiger Valk \\
Department of Informatics, University of Hamburg\\ 
Vogt-K\"olln-Str. 30,  D-22527 Hamburg, Germany\\
ruediger.valk@uni-hamburg.de
\and Daniel Moldt\\
Department of Informatics, University of Hamburg\\ 
Vogt-K\"olln-Str. 30,  D-22527 Hamburg, Germany\\
daniel.moldt@uni-hamburg.de
 } 
 \maketitle

\runninghead{R. Valk, D. Moldt}{On Reduction and Synthesis of  Petri's Cycloids}

\vspace*{-9ex}

\begin{abstract}
Cycloids are particular Petri nets for modelling processes of actions and events, belonging to the fundaments of Petri's general systems theory.
Defined by four parameters, they provide an algebraic formalism to describe strongly synchronised sequential processes.
To further investigate their structure, reduction systems of cycloids are defined in the style of rewriting systems
and properties of irreducible cycloids are proved. In particular, the synthesis of cycloid parameters from their Petri net structure is derived, leading to
an efficient method for a decision procedure for cycloid isomorphism.
\end{abstract}


\maketitle

\section{Introduction}\label{sec-intro}

\begin{figure}[bt]
\hspace{-1.5cm}
\begin{subfigure}[c]{0.6\textwidth}
\hspace{1.0cm}
\includegraphics[width=1\textwidth]{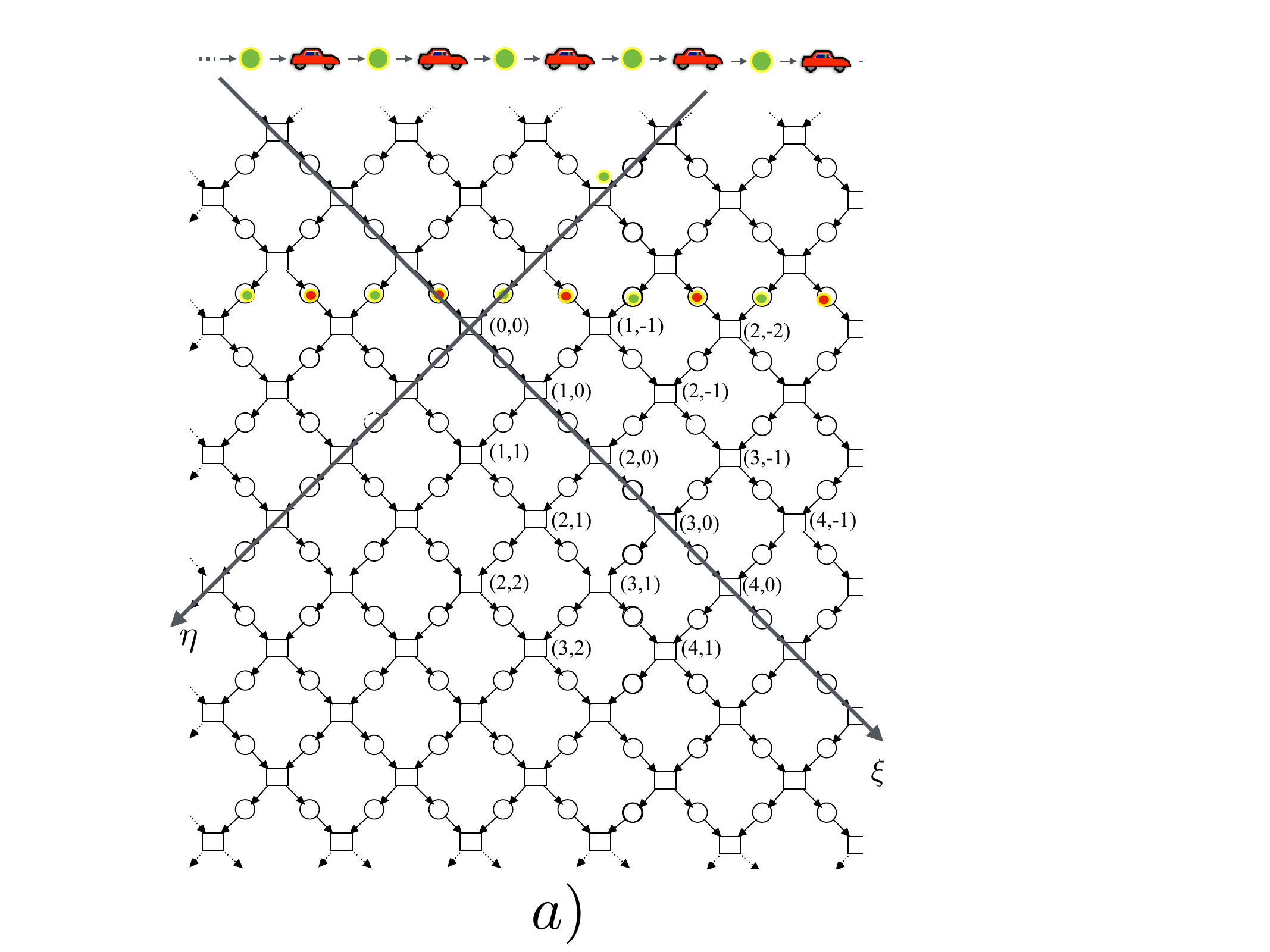}
\end{subfigure}
\begin{subfigure}[c]{0.6\textwidth}
\hspace{-1.9cm}
\includegraphics[width=1.3\textwidth]{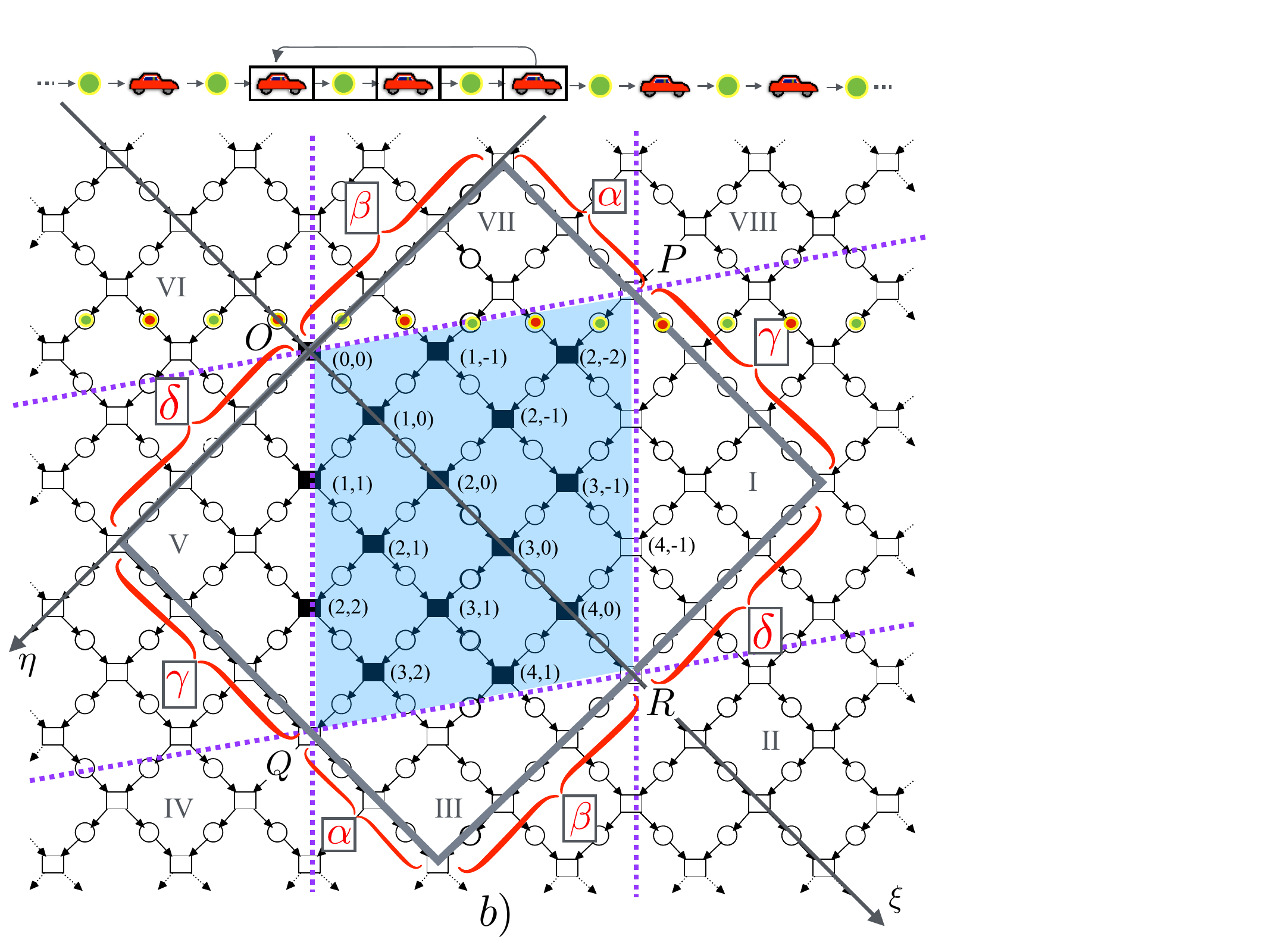}
\label{}
\end{subfigure}
\vspace{-2ex}
\caption{a) Petri space with infinite queue of cars and b) its folding to a fundamental parallelogram. }
\label{pspace}
\end{figure}


Similar to the Minkowski space \cite{Min-1923,Cat-2008}, but with discrete space and time steps, with the
\emph{Petri space} Carl Adam Petri created a basis for modelling and visualising movement in space and time
\cite{Petri-NTS}. This is, as shown in Figure~\ref{pspace} a), a net (in particular an infinite marked graph) labelled
with coordinates, where only some of the coordinates are given here.
The left (resp. right) input place of a transition is denoted forward (backward) input place of the transition.
An (infinite) initial marking is drawn to enable the net to be active. This marking corresponds to the (infinite) queue
of cars drawn at the top of the figure.

The net models an infinite number of cars and spaces.
Only six cars of the infinite number of cars are shown in the figure.
For instance, the token in the forward input place of the transition with coordinate $(0,0)$ corresponds to the second
of the cars drawn, as well as the token in the backward
input place of the transition corresponds to the  gap after the aforementioned car. Both, the car and the transition can
make a step, just like the other infinite number of cars and transitions on the right and left. All in all, the net
behaves like the endless queue of cars with all the phenomena of concurrency.

For Petri, it was always of the utmost importance that all modelling in the context of computer applications be carried
out in finite spaces and discrete steps. He therefore limited the queue of cars to a finite length and modelled this in
such a way that he defined a first and last car. As a result of this restriction of the spatial extension, the lower
bound is created by the lines across the transitions $O =(0,0)$ and $Q=(3,3)$ and the upper bound by $P=(2,-3)$ and
$R=(5,0)$ (as vertical dotted lines in
Figure  \ref{pspace} b).

\begin{figure}[htbp]
 \begin{center}
\hspace{0 cm}
        \includegraphics [scale = 0.25]{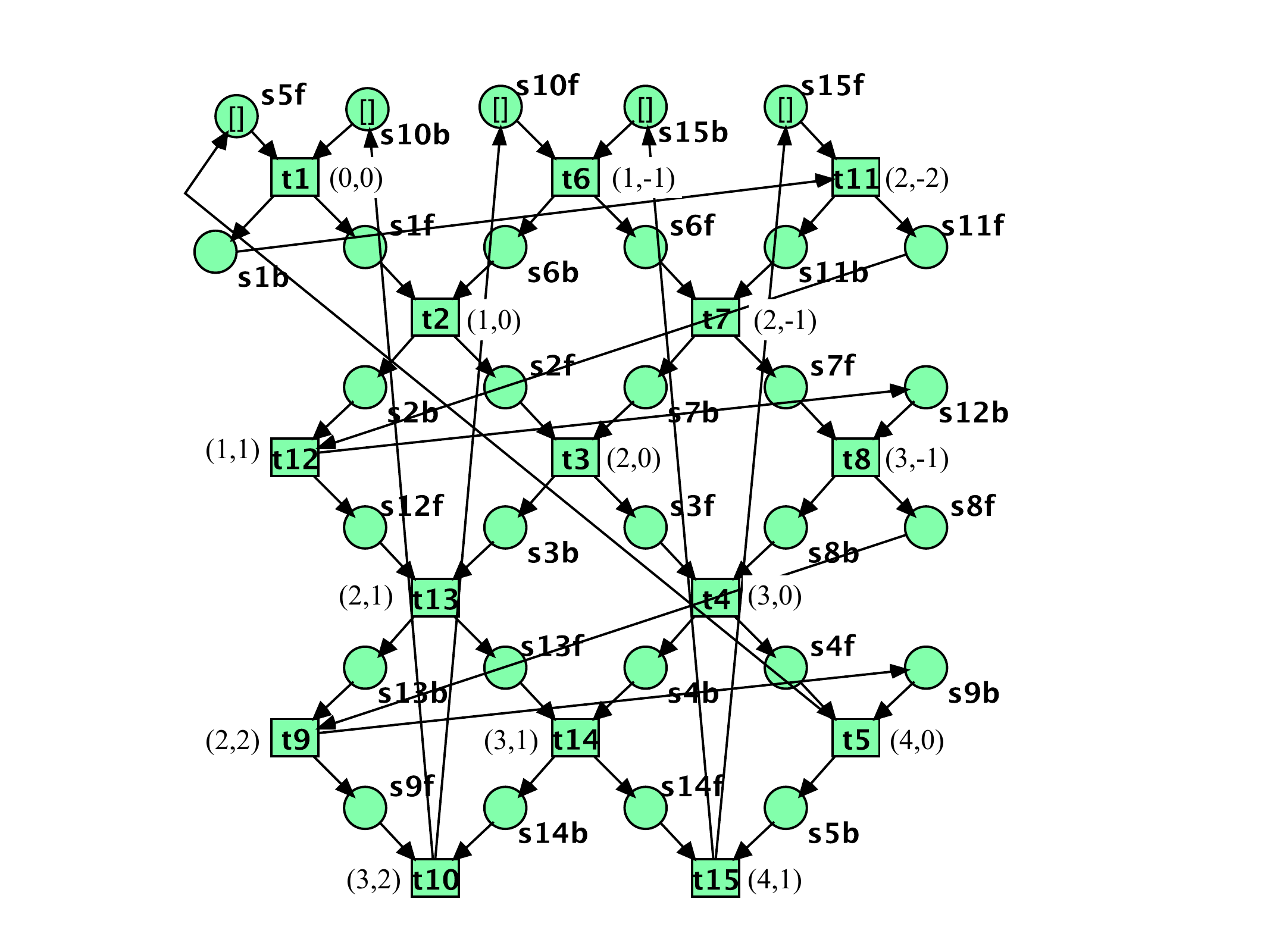}
   \caption{The cycloid $\mathcal{C}( 2,3,3,3 )$. }
        \label{2333}
      \end{center}
\end{figure}

Similarly, the temporal restriction is represented by the lines containing the transitions $O$ and $P$ and the upper
bound by $Q$ and $R$ (as approximately horizontal lines in Figure  \ref{pspace} b). The resulting form is known as a
\emph{fundamental parallelogram} with corners $O, P, Q$ and $R$ of a cycloid.

\eject

\vspace*{-7ex}

Petri has represented such a parallelogram by the sections of lengths  $ \alpha, \beta, \gamma$ and $\delta  $ on the
grid lines and it is notated as  $ \mathcal{C}( \alpha, \beta, \gamma, \delta ) $, i.e.  $ \mathcal{C}( 2, 3, 3, 3 ) $
for the example of Figure  \ref{pspace} b). All transitions within this parallelogram and those on the lines containing
the origin $O=(0,0)$ are declared as transitions of the cycloid.

In contrast, transitions on the lines through the diagonally opposite point $R=(5,0)$ are not included, as they are
assigned to the neighbouring parallelograms of the same shape.
The transitions included are shown in black. Their number is
$A = \alpha \cdot \delta + \beta \cdot \gamma = 2 \cdot 3 + 3 \cdot 3 = 15$, as Petri has shown in  \cite{Petri-NTS} .

A cyclical structure is then formed by a folding of all transitions of these neighbouring parallelograms with the same
shape onto the fundamental parallelogram. 
This corresponds to the fact that the fourth car, as represented in the infinite queue,
at the top of the figure takes the place of the second car when  this has become free.

If, for example, the transition $(3,-1)$ occurs and enables the transition $(4,-1)$ in the Petri space, then the same
transition $(3,-1)$ enables $(2,2)$ in the  folding. The resulting net of the cycloid is shown in Figure \ref{2333}. To
show the relationship to the Petri space coordinates, these are inserted. Note, that in fact,   the transition 
$\textbf{t8}=(3,-1)$ enables $\textbf{t9}=(2,2)$, if the second input place is marked.

The initial marking of cycloids has been defined by Petri in an informal way in \cite{Petri-NTS}.
In presentations and discussions, however, he used many examples to explain how this should be specified. On this basis,
a formal definition (Definition \ref{standard-M0}) was given in \cite{ Kummer-Stehr-1997} and \cite{fenske-da}, which
was used in \cite{fenske-da,Valk-2019} to prove that all cycloids are safe (1-bounded) and live.
The cycloid  $ \mathcal{C}( 2,3,3,3) $ models a circular queue of length $\alpha+\beta = 5$ containing $\beta = 3$ cars
with $\alpha = 2$  gaps in between.

\begin{figure}[h]
 \begin{center}
\hspace{0 cm}
        \includegraphics [scale = 0.4]{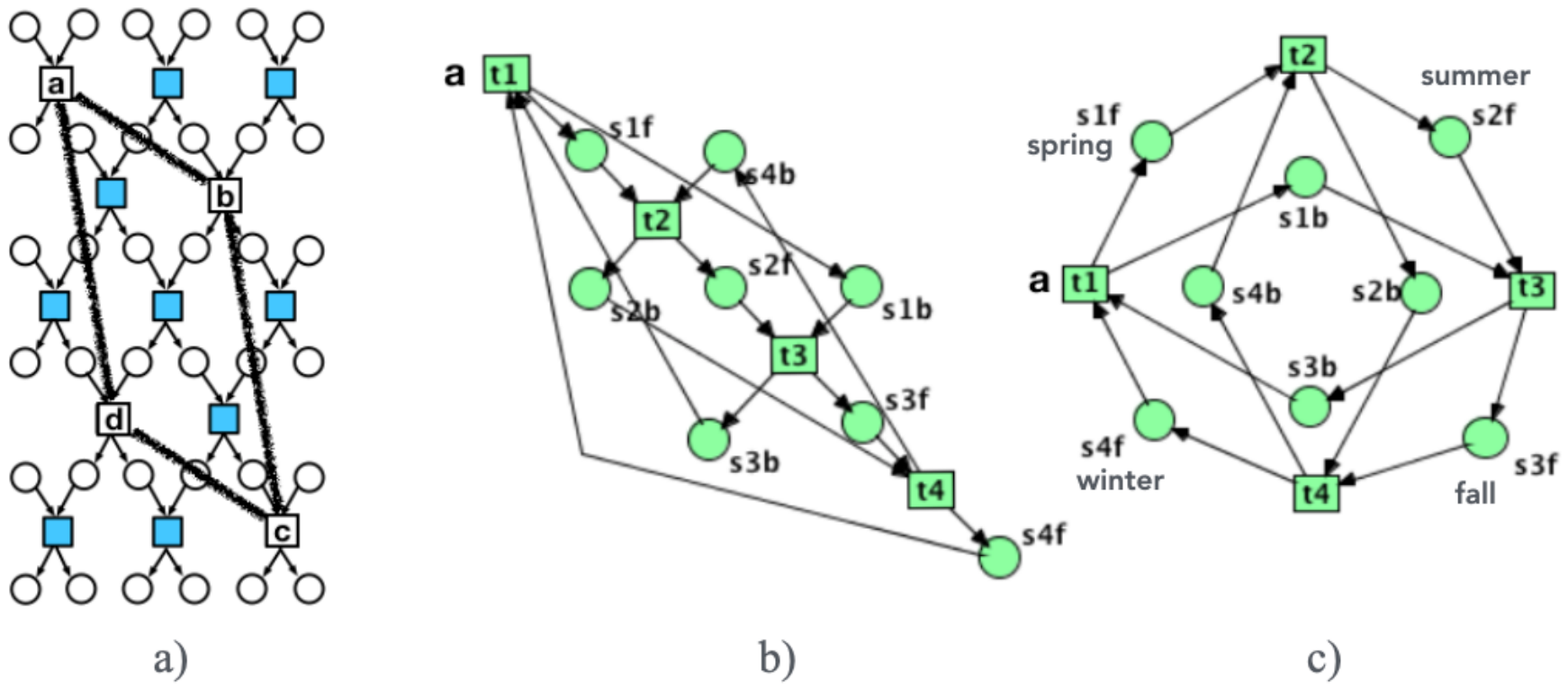}
\vspace*{-6ex}
   \caption{Folding for the cycloid $\mathcal{C}( 2,1,2,1)$ called \emph{4-seasons} by Petri.  }
        \label{4seasons}
      \end{center}
\end{figure}

\eject

\vspace*{-7ex}

By the condition that $\beta$ divides $\delta$ \emph{regular cycloids} \cite{Valk-2020} are defined.  This subclass of
cycloids  models synchronised sequential processes. The net representation of  such a cycloid
consists  of $\beta$ strongly synchronised sequential processes of length $p = \frac{A}{\beta}$. An additional
restriction $\beta=\gamma=\delta$  defines the \emph{canonical regular cycloids} whose behaviour corresponds to that of
a cyclic queue of length $p = \alpha + \beta$. As proved in \cite{Valk-2020} they model circular traffic queues of
$\beta$ traffic items (e.g. cars) and $\alpha$ gaps in between. The cycloid  $ \mathcal{C}( 2,3,1,6 ) $ is regular, but
not canonical regular. However, by the results of this article it is isomorphic to the canonical regular cycloid  $
\mathcal{C}( 2,3, 3,3 ) $.

The cycloid  \emph{4-seasons} is one of Petri's favourite examples.
Figure~\ref{4seasons} from \cite{Valk-2019} shows its development from the Petri space in part a) via the folding to a
net in b) to an interpretation originating from Petri in c). Transitions \textbf{b},  \textbf{c} and  \textbf{d} are
folded with  transition \textbf{a}.
Places like \textbf{s2b} model fuzzy specifications, here the statement "\emph{it is summer or fall}".
The same net was also used by Petri to represent an oscillator, coupling the momentum and a position.


In this article, after the introduction to the field of cycloids 
 (\emph{Section 1}), 
 the formal prerequisites for the following results are presented.
 These contain new results, some of which involve complex proofs, which are used in this article on the one hand, but are also of general use for research into the structure of cycloids on the other. Petri formulated nets almost exclusively with places that could contain at most one token. He only advocated nets with multiple tokens if they were used to model resources from many objects. With this in mind, he also introduced cycloids. We follow this approach, also because it is easier to handle and  extensions to multiple tokens are not used in this context. Nets (Definition \ref{def-net}) are not restricted to finite sets of places and transitions, since the Petri space (Definition \ref{petrispace}) is an infinite marked graph. Afterwards, cycloids are defined by finite sets of places and transitions, which are infinite classes of places and transitions of the Petri space. To facilitate understanding of this definition, Table \ref{values} illustrates some elements (six, to be precise) of five  infinite classes using the  example $ \mathcal{C}( 2,3,3,3 ) $. These classes are defined by an equivalence relation using the four parameters  $  \alpha, \beta, \gamma$ and $ \delta  $. 
%
%
%
 This relation must be calculated relatively frequently in proofs. For this purpose, Theorem 
\ref{parameter} uses the matrix of a cycloid. Since this approach can often be used to advantage, it is referred to as (linear) \emph{cycloid algebra}. 
A much more difficult problem is to find  for a given transition   in the  Petri space
the (uniquely determined) equivalent transition of the fundamental parallelogram of the cycloid. This problem can be solved by enumerating all transitions of the fundamental parallelogram and each time checking for equivalence. 
Using  cycloid algebra, however,  Theorem \ref{xy-to-FP} gives a solution by a closed expression that is used in different proofs.
 
 Four instances of \emph{cycloid shearing} have been presented in \cite{Valk-2019} and proved to be net isomorphisms. However, since cycloid algebra was not yet known at that time, we present a new proof by
  Theorem \ref{shear} with the help of this algebra. When analysing cycloids based on their net representation, it is important to know the minimum length of cyclic paths. Again, by Theorem \ref{minimal cycles} a) cycloid algebra is used to solve this problem. Since this solution involves high runtime complexity due to a minimum calculation, parts c) to d) of this theorem deal with special cases without such a minimum operator.  Since case c) covers a particularly large class, cycloids of this class are referred to as \emph{lbc-cycloids}.

 
 In applications a system is represented or  modelled by a description of its observed behaviour, for instance by a transition system or a reachability graph.
\emph{Synthesis}, on the other hand, refers to the implementation of the system in a formalism such as an automaton, Petri net or programming language \cite{reisig-under}. 
For Petri nets, this has been introduced by the \emph{region theory} \cite{Ehr-Roz-1990,Bernard-1998}, whereby a Petri net is constructed from a given transition system whose reachability graph is isomorphic to the transition system.
By the term \emph{synthesis}, mentioned in the title of this article, we mean a similar and additional approach:
from a Petri net representation of a cycloid the four parameters of a cycloid are computed. While the Petri net may have thousands of transitions and a complex structure, this is then described by only four integers of the cycloid.

 To achieve this goal, reduction rules for cycloids are formulated in Section \ref{sec-reduction} on the basis of the shear isomorphisms from Section \ref{sec-cycloids}, and concepts from rewriting systems are adopted for reduction. As a first result Theorem \ref{Th-normal-a-neq-b} shows that the lbc synthesis theorem (Theorem \ref{synthesis}) allows for an easier computation of the parameters $\gamma$ and $\delta$ for some irreducible cycloids.

 The following Section  \ref{sec-morphism} presents a considerably stronger form of cycloid synthesis. This result is connected with a decision procedure for cycloid isomorphism. 
 The algorithm is much more efficient than a general graph isomorphism decision procedure, the complexity of which is suspected to be  NP-intermediate (not known to be solvable in polynomial time or NP-complete) \cite{Babei-1983}.
 Therefore, it will first be shown that the shear mappings (Theorem \ref{shear}) are not only net isomorphisms (Definition \ref{def-net}) but even cycloid isomorphisms (Definition \ref{def-cyc-iso} and Theorem \ref{shearing}). An important part of the proof of Theorem  \ref{shearing} was anticipated in the proof of Lemma \ref{pi(post)}. Using an algorithm similar to Euclid's algorithm for calculating the greatest common divisor, important properties of irreducible cycloids are derived in Theorem \ref{bd-irreducible}.
 From this, a method is derived that determines the parameters of the irreducible cycloid using path properties of the cycloid net. Since nets of cycloid-isomorphic cycloids must have the same path structure, this method can be used to determine the same parameters of their irreducible cycloid. This makes it possible to determine whether two cycloids are cycloid-isomorphic with lower time complexity  
 of $T(n) = \mathcal{O}(log_2 \, n)$, where $n= max\{\beta_1,\delta_1,\beta_2,\delta_2\}$ for the cycloids  
 $ \mathcal{C}_1( \alpha_1, \beta_1, \gamma_1, \delta_1 ) $ and  
 $ \mathcal{C}_2( \alpha_2, \beta_2, \gamma_2, \delta_2) $ in question.

  It extends Theorem \ref{th-beta-delta-red} by determining not only the parameters of the irreducible cycloid from the cycloid net, but all parameters of the entire reduction chain.  By reading an $\beta\delta$-reduction backwards, one obtains an $\alpha\gamma$-reduction. These reductions are connected by the symmetric transformation, whose isomorphism property is also proven in Theorem \ref{shear} e). 
 Obviously, this transformation is not a cycloid-isomorphism, since forward and backward places are swapped.
 
 
 Substantial progress has been made compared to the previously published workshop version  \cite{ValkMoldt22} of this article. This applies in particular to the representation of reductions as rewriting systems and the properties of $\beta\delta$-irreducible cycloids obtained in Theorem \ref{bd-irreducible}. 
  As a result, new methods of reconstructing the cycloid parameters from the Petri net representation of such cycloids were found.
These results apply to the entire class of cycloids and are not restricted to subclasses as regular cycloids. Circular queues of cars were only used in this introduction to explain foldings of the Petri space and play no role in the rest of this article. A formal definition would go too far and can be found in \cite{Valk-2020}.


 Cycloids have been introduced  by C.A. Petri in \cite{Petri-NTS} in the section on physical spaces, using as examples firemen carrying the buckets with water to extinguish a fire, the shift from Galilei to Lorentz transformation and the representation of elementary logical gates like Quine-transfers.
Besides the far-sighted work of Petri we got insight in his concepts of cycloids by numerous seminars he held at the University of Hamburg \cite{valk-2019-2worlds}.
Based on  formal descriptions of cycloids in \cite{Kummer-Stehr-1997} and \cite{fenske-da} a more elaborate formalisation is given in \cite{Valk-2019}. Semantical extensions to include more elaborate features of traffic systems have been presented in \cite{jessen-moldt}.
The nets of this article are generated by the RENEW-tool\footnote{http://www.renew.de}\cite{Moldt+23b}.

We recall some standard notations for set theoretical relations.
If $R\subseteq A \!\times\! B$ is a relation and $U \subseteq A$ then 
$R[U]:= \{b\,|\,\exists u \in U: (u,b)\in R\}$ is the \emph{image} of $U$ and $R[a]$ stands for $R[\{a\}]$. 
$R^{-1}$ is the \emph{inverse relation} and $R^+$ is the \emph{transitive closure} of $R$ if $A=B$.
Also, if $R\subseteq A \!\times\! A$ is an equivalence relation then $\eqcl[R]{a}$
 is the \emph{equivalence class} of the quotient $A/R$ containing $a$
 and $\eqcl[R]{U} = \{ \eqcl[R]{u} | u \in U \}$ is the set of equivalence classes with representatives in $U$.
 Furthermore  $\Nat$, $\Natp$, $\Int$ and $\Real$ denote the sets of nonnegative integers,  positive integers, integers and real numbers, respectively.
For integers: $a|b$ if $a$ is a factor of $b$.
The $modulo$-function is used in the form 
$a \,mod\,b = a - b \cdot \lfloor \frac{a}{b} \rfloor$, which also holds for negative integers $a \in \Int$.
In particular, $-a\,mod\,b = b-a$ for $0<a\leq b$.  \linebreak $gcd(a,b)$ denotes the greatest common divisor of $a$ and $b$.



\section{Petri Space and Cycloids}  \label{sec-cycloids}

In general for cycloids all places contain at most one token, as they are to be understood as conditions, like the \emph{condition-event nets} favoured by Petri. 
As this allows for a simpler description of the behaviour of Petri nets, they will  be introduced in this way in the following.

 \begin{definition}[\cite{Valk-2019}] \label{def-net} 
 As usual, a net $ \N{} = (S, T, F)$ is defined by non-empty, disjoint sets 
 $S$ of places and $T$ of transitions, connected by a flow relation 
 $F \subseteq (S \cp T) \cup (T \cp S)$ and $X := S \cup T$.
A transition $t \in T$ is \emph{active} or \emph{enabled} in a marking $M \subseteq S$ if $\; ^{\ndot} t \subseteq M$,
where 
$^{\ndot} x := F^{-1}[x], \; x^{\ndot} := F[x] $
denotes the input and output elements of an element $x \in X$, respectively. 
In this case we obtain $M \stackrel{t}{\rightarrow}M'$ with the follower marking $M' = (M \backslash^{\ndot} t )\cup t^{\ndot}$. 
The relation $\stackrel{*}{\rightarrow}$ is the reflexive and transitive closure of $\rightarrow$. \linebreak
 A net together with an initial marking $M_0 \subseteq S$ is called a \emph{net system} $(\N{},M_0)$.
 Given two net systems $ \N{1} = (S_1, T_1, F_1,M_0^1)$ and 
 $ \N{2} = (S_2, T_2, F_2,M_0^2)$ a mapping $f: X_1 \to X_2$ ($X_i = S_i \cup T_i$) is a \emph{net morphism}
 (\cite{Reisig-Smith-87}) if $f(F_1 \cap (S_1 \cp T_1) )\subseteq (F_2 \cap (S_2 \cp T_2)) \cup id$ and 
 $f(F_1 \cap (T_1 \cp S_1) )\subseteq (F_2 \cap (T_2 \cp S_2)) \cup id$ and $f(M^1_0) = M^2_0$, 
 where
 $id$ is the \emph{identity relation} $\{ (x,x) | x \in X_2 \}$.
 It is an \emph{isomorphism} if it is bijective and the inverse mapping $f^{-1}$ is also a net morphism.
 $\N{1}\simeq\N{2}$ denotes isomorphic net systems.
 Omitting the initial states the definitions apply also to nets.
   $ \N{} = (S, T, F)$ is a  \emph{marked graph} (also called  $T$-net or \emph{synchronisation graph}) if $|^{\ndot} s| = |s ^{\ndot}| = 1$ for all places $s\in S$.
\end{definition}

As motivated in the introduction, the Petri space is an infinite marked graph whose transitions are labelled by Cartesian coordinates.  As shown in Figure \ref{P-space+FD} a) the  axes are denoted by the characters $\xi $ and 
$\eta$ and the two output places of a transition $t_{\xi,\eta}$ are called \emph{forward output place} 
$\gsvw{\xi,\eta}$ and \emph{backward output place} $\gsrw{\xi,\eta}$ (see Figure \ref{P-space+FD} b).
$\GSvw[1]$ and $\GSrw[1]$ are the sets of forward and backward  places, respectively, and their union is the set of places of the Petri space, as defined below.

\begin{figure}[h]
	\begin{center}
	\includegraphics[width=0.8\textwidth]{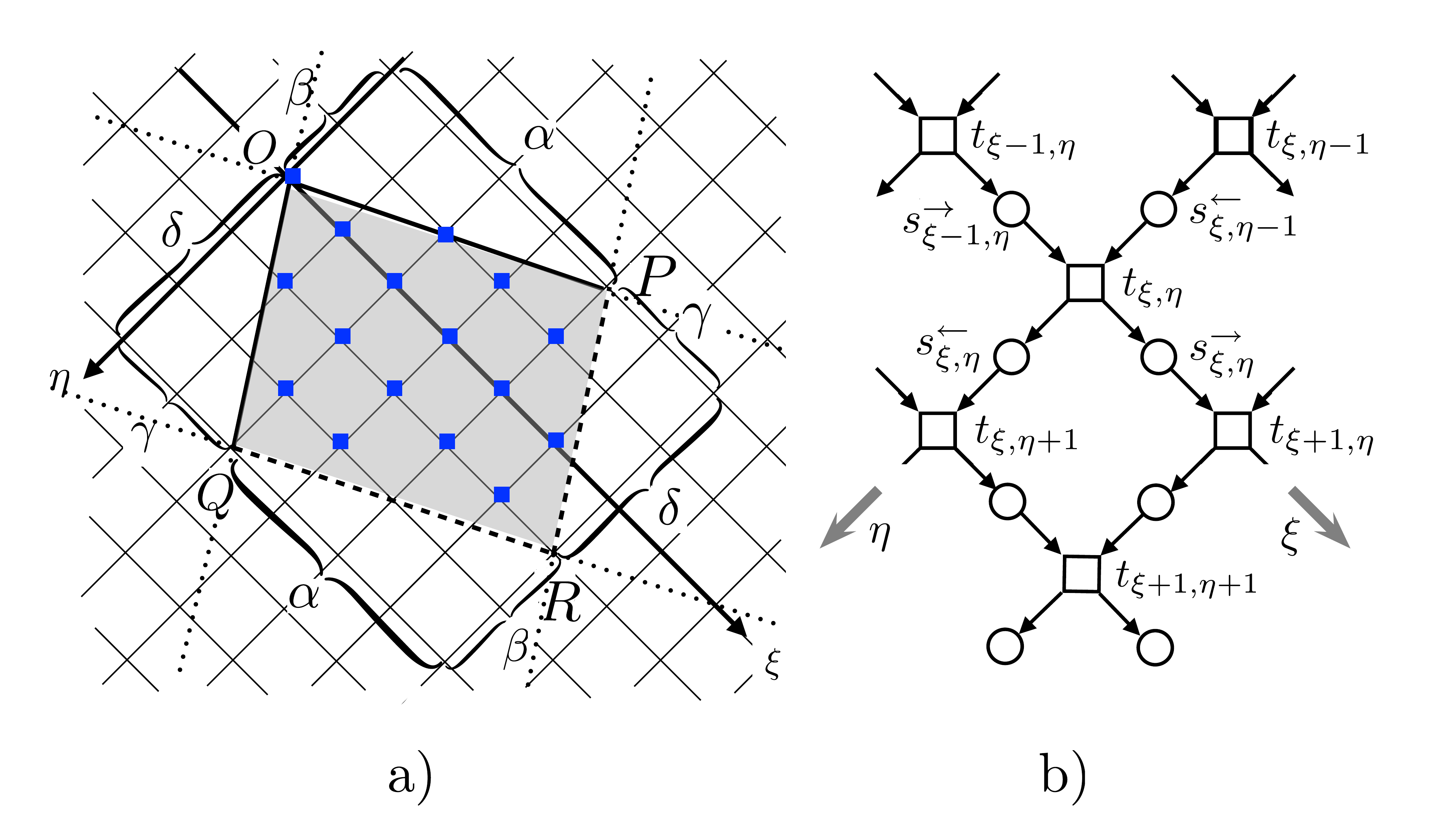}
		\caption{a)  Fundamental parallelogram of $\mathcal{C}(4,2,2,3)$ and  b) Petri space.}
		\label{P-space+FD}
	\end{center}
\end{figure}
	
\vspace*{-2ex}

\begin{definition} [\cite{Valk-2019}] \label{petrispace}
The $Petri \; space$ 
is defined by the net 
$\PS{1} := (S_1, T_1, F_1)$ \ where
\begin{itemize}
       \item [a)] $S_1 = \GSvw[1]\cup \GSrw[1]$ with $ \;\GSvw[1]  =$ 
   $\col{\gsvw{\xi,\eta}}{\xi,\eta \in \Int},$ 
$\;\GSrw[1] = \col{\gsrw{\xi,\eta}}{\xi,\eta \in \Int }$ 
is the set of places,
       \item [b)]  $ T_1 = \col{t_{\xi,\eta}}{\xi,\eta \in \Int }$ is the set of transitions,
       \item [c)] $ F_1 =$ $ \col{(t_{\xi,\eta},\gsvw{\xi,\eta})}{\xi,\eta \in \Int } \cup 
 \col{(\gsvw{\xi,\eta},t_{\xi+1,\eta})}{\xi,\eta \in \Int } \cup \\ $ 
\hspace*{26pt}$\col{(t_{\xi,\eta},\gsrw{\xi,\eta})}{\xi,\eta \in \Int } \cup \col{(\gsrw{\xi,\eta},t_{\xi,\eta+1})}{\xi,\eta \in \Int }$ 
\hspace{1em} is the flow relation.
\end{itemize}     
 $\GSvw[1]$ is the set of \emph{forward places} and $\GSrw[1]$ the set of 
 \emph{backward places}.
 $ \prenbfw{t_{\xi,\eta}}:= \gsvw{\xi-1,\eta}$ is the forward input place of $t_{\xi,\eta}$ and in the same way 
 $ \prenbbw{t_{\xi,\eta}}:= \gsrw{\xi,\eta-1}$,
 $ \postnbfw{t_{\xi,\eta}} := \gsvw{\xi,\eta}$ and
 $\postnbbw{t_{\xi,\eta}}:= \gsrw{\xi,\eta}$ (see Figure~\ref{P-space+FD} b).

\end{definition}

In two steps, by a twofold folding with respect to time and space, Petri defined the cyclic structure of a cycloid, by parameters  $ \alpha, \beta, \gamma$ and $ \delta $ from $\Natp$. As one of these steps is  
 a folding $f$ with respect to space, the transition $P = t_{\alpha,-\beta}$ is merged with the origin $O = t_{0,0}$
 (see Figures \ref{pspace} and~\ref{P-space+FD}~a). 
 The corresponding folding is
 $f(t_{\xi+m\cdot\alpha,\eta-m\cdot\beta}) = f(t_{\xi,\eta})$, fusing all points 
  $t_{\xi+m\cdot\alpha,\eta-m\cdot\beta}$ 
   of the Petri space with $t_{\xi,\eta}$ where $m\in \Int$.  
The resulting still infinite net is called a \emph{time orthoid} (\cite{Petri-NTS}, page 37), as it extends infinitely in temporal future and past.
The second step is a folding with $f(t_{\xi+n\cdot\gamma,\eta+n\cdot\delta}) = t_{\xi,\eta}$ with
$n\in \Int$ reducing the system to a cyclic structure also in time direction. By this folding the transition $Q$ is merged with the origin $O$. In general, the cycloid is obtained by a linear combination of these two cases, by applying $m$-times the first case and $n$-times the second: $f(t_{\xi,\eta})=t_{\xi+m\cdot\alpha + n\cdot\gamma,\eta-m\cdot\beta+n\cdot\delta}$ with $m,n \in \Int$. 
By these two foldings a transition $ t_{\xi,\eta}$ is merged with each transition $t_{\xi+m\cdot\alpha+n\cdot\gamma,\,\eta-m\cdot\beta+n\cdot\delta} $ of the Petri space for all $ \xi, \eta, m, n \in \Int $. 
Hence, the corner $R = t_{\alpha+\gamma,\delta-\beta}$ is also merged with the origin $O$ by choosing $(m,n) = (1,1)$.
To refer to the  cycloid  $ \mathcal{C}( 2,3,3,3) $ discussed in Section \ref{sec-intro} transition $t_{2,2}$ is fused with $t_{4,-1}$ by choosing $(m,n) = (1,0)$ since $f(t_{2,2})=f(t_{\xi,\eta})=
t_{\xi+m\cdot\alpha + n\cdot\gamma,\;\eta-m\cdot\beta+n\cdot\delta} = 
t_{2+1\cdot 2 + 0\cdot 3,\;2-1\cdot 3+0\cdot 3} = t_{4,-1}$.
Note that a cycloid is a net, as will be defined next, while the fundamental parallelogram is its representation in the Petri space.


To keep the example more general, 
in Figure~\ref{P-space+FD} a) the values 
 $ ( \alpha, \beta, \gamma, \delta ) = (4,2,2,3)$ for a not regular cycloid are chosen.
 In this representation of a cycloid, called \emph{fundamental parallelogram}, the squares of the transitions as well as the circles of the places are omitted.
 All transitions with coordinates within the parallelogram belong to the cycloid including those on the lines between $O,Q$ and $O, P$, but excluding those of the points $Q,R,P$ and those on the dotted edges between them.
 All copies of the fundamental parallelogram outside along the dotted lines,  having the same shape,  are merged with it by the quotient operation of the equivalence relation in Definition \ref{cycloid}. 
 In Figure~\ref{pspace}~b) the directly neighbouring parallelograms are labelled with Roman numerals \textbf{I} to \textbf{VIII}, clockwise starting with the right neighbour parallelogram.


\begin{definition} [\cite{Valk-2019}] 
\label{cycloid}
A \emph{cycloid} is a net \ $ \mathcal{C}( \alpha, \beta, \gamma, \delta ) = (S, T, F)$, defined by parameters 
 \ $ \alpha, \beta, \gamma, \delta \in \Natp$, by a quotient \cite{Reisig-Smith-87} of the Petri space $\PS{1} := (S_1, T_1, F_1)$ \ 
with respect to the equivalence relation
$\mo\zykaeq 
\subseteq X_1 \cp X_1 $ defined by
$ x_{\xi,\eta} \zykaeq x_{\xi+m\cdot\alpha+n\cdot\gamma,\,\eta-m\cdot\beta+n\cdot\delta} $
for all $ \xi, \eta, m, n \in \Int $.
In this definition, where $X_1 = S_1 \cup T_1$, type preservation is demanded: 
$ \mo\zykaeq[\GSvw[1]] \subseteq \GSvw[1]$, 
$\mo\zykaeq[\GSrw[1]]  \subseteq \GSrw[1] $ and
$\mo\zykaeq[T_1]  \subseteq T_1     $.
\begin{itemize}
       \item [a)] $ S =   \eqcl[\zykaeq]{\GSvw[1]}   \cup \eqcl[\zykaeq]{\GSrw[1]}  $ \\ is the set of equivalent classes of the forward and backward places of $\PS{1} $,
       \item [b)] $T =   \eqcl[\zykaeq]{T_1} $ is the set of equivalent classes of the transitions of $\PS{1} $,
       \item [c)] The flow relation is defined via the representatives $x'$ and $y'$ of the classes:  \\ 
 $ \eqcl[\zykaeq]{x} \mb{F} \eqcl[\zykaeq]{y} \: \Leftrightarrow
\exists\,{x'\in\eqcl[\zykaeq]{x}}\,\exists\,y' \in \eqcl[\zykaeq]{y}: \,x' F_1 y' $
 \ for all $x, y \in X_1 $. 
 \end{itemize}  
The matrix $\mathbf{A} = \begin{pmatrix} \alpha & \gamma \\ -\beta & \delta \end{pmatrix} $ is called the matrix of the cycloid.
Petri denoted the number $|T|$ of transitions as the area $A$ of the cycloid and proved in \cite{Petri-NTS} its value to $|T| =A =\alpha\cdot\delta+\beta\cdot\gamma$ which equals the determinant $A = det(\mathbf{A})$.
The embedding of a cycloid in the Petri space is called \emph{fundamental parallelogram} 
(see Figure~\ref{P-space+FD} a).
When the distinction between forward places $\GSvw $ and backward places $\GSrw $ is explicitly given we denote it as the \emph{cycloid net} of the cycloid and represent it by
 $\N{} =(\GSvw,\GSrw,T,F)$.
\end{definition}

\begin{lemma} \label{corr}
Definition \ref{cycloid} c) is correct: the definition of $F$ is independent of the choice of representatives.
\end{lemma} 
\begin{shortproof} 
Each transition $t_{\xi,\eta}$ is connected to four  places 
 $ \gsvw{\xi-1,\eta}$,
 $ \gsrw{\xi,\eta-1}$,
 $  \gsvw{\xi,\eta}$ and
 $\gsrw{\xi,\eta}$ (see Figure~\ref{P-space+FD} b). The following proof considers the first case as the remaining are similar. Let be $(\gsvw{\xi-1,\eta},t_{\xi,\eta})\in F_1$ and  $\gsvw{\xi',\eta'}\in\eqcl[\zykaeq]{\gsvw{\xi-1,\eta}}$
with $(\gsvw{\xi',\eta'},t_{\xi'+1,\eta'}  )\in F_1$. We have to prove $t_{\xi'+1,\eta'}\in \eqcl[\zykaeq]{t_{\xi,\eta}}$
or $t_{\xi'+1,\eta'} \equiv t_{\xi,\eta}$.
By definition $\gsvw{\xi',\eta'}\in\eqcl[\zykaeq]{\gsvw{\xi-1,\eta}}$ implies 
$\xi' = \xi-1+m\cdot\alpha+n\cdot\gamma$ and
$\eta' = \eta-m\cdot\beta+n\cdot\delta$
for some $m,n \in \Int$. Therefore 
$t_{\xi'+1,\eta'} = t_{\xi-1+m\cdot\alpha+n\cdot\gamma+1,\eta-m\cdot\beta+n\cdot\delta} =$
$ t_{\xi+m\cdot\alpha+n\cdot\gamma,\eta-m\cdot\beta+n\cdot\delta}  \equiv $
$t_{\xi,\eta}$.
\end{shortproof}  

The independence of the choice of representatives also extends to the relationships between classes of transitions of a cycloid. We will give an example using the classes  in Table \ref{values} for the cycloid 
$ \mathcal{C}( 2,3,3,3 ) $.


In  graphical representations like Figure \ref{2333} or Figure \ref{4seasons} forward  and backward places
are distinguished by the suffix \textbf{f} and \textbf{b}, respectively,  in their  identifiers.
Note that the equivalence relation $\mo\zykaeq $ is also defined on the  places of the Petri space. To compute the classes of the equivalence relation  some of the values of $\xi+m\cdot \alpha + n\cdot \gamma$ and $\eta-m\cdot \beta+n\cdot \delta$ are given in Table \ref{values} for the cycloid  $ \mathcal{C}( 2,3,3,3) $ from Section \ref{sec-intro}. Only one of these values can denote a transition $t_{\xi,\eta}$ within the fundamental parallelogram of Figure \ref{pspace} b). The corresponding coordinates are shown in bold. Note that the first and last line describe the same class
with partly different representatives from their infinite number of elements.

To give an example for the Definition of the flow relation (case c) of Definition \ref{cycloid} we consider transition 
\textbf{t9} and its forward input place \textbf{s8f} of the cycloid  $ \mathcal{C}( 2,3,3,3 ) $ of Figure \ref{2333}.
The coordinates of \textbf{t9}  are $(2,2)$ and by the first line of Table \ref{values} some values of its class are
$\textbf{t9} = \eqcl[\zykaeq]{t_{\textbf{2,2}}} = \{  t_{\textbf{2,2}}, t_{4,-1}, t_{5,5}, t_{0,5}, t_{7,2}, t_{9,-1}, \cdots\}$. In the Petri space the forward input place of $t_{\textbf{2,2}}$  is the forward output place of $t_{1,2}$ and its class is 
$ \eqcl[\zykaeq]{s^\rightarrow_{1,2}} 
= \{  s^\rightarrow_{1,2},s^\rightarrow_{\mathbf{3,-1}},s^\rightarrow_{4,5},s^\rightarrow_{-1,5}, s^\rightarrow_{6,2}, s^\rightarrow_{8,-1} \cdots\}$ by the third line of Table \ref{values}. 
The only element of this class, that lies in the fundamental parallelogram, is $s^\rightarrow_{\mathbf{3,-1}}$ which is $\postnbfw{\mathbf{t8}} = \mathbf{s8f}$. Hence, the relation $ \eqcl[\zykaeq]{x} \mb{F} \eqcl[\zykaeq]{y} $ from Definition \ref{cycloid} c) implies $( \mathbf{s8f},\mathbf{t9}  ) \in F_{\mathcal{C}( 2,3,3,3 )}$ in the cycloid of Figure \ref{2333}.
The three remaining edges at transition $\mathbf{t9} $ can be determined in the same way.

An example for the independence of the choice of representatives 
is now given by  this table.
The transition $t_{3,-1}$ is found in the third row, representing the the class of 
$(\xi,\eta) =(1,2)$. Its successor in $\xi-$direction is $t_{4,-1}$ in the class of $(\xi,\eta) =(2,2)$ of the first row. Now, we observe that all elements of the first row are obtained from the elements of the third row in the same column by adding the value $(1,0)$.

\begin{table}[htbb]
 \begin{center}
 \small
\begin{tabular}{||c||c|c|c|c|c|c|l}
\hline
 $ \mathcal{C}( 2,3,3,3 ) $ &$(m,n) $  &  $(m,n) $  & $(m,n) $ &$(m,n) $  & $(m,n) $ & $(m,n) $ &…\\  
 $$&$=(0,0)$  &  $=(1,0)$  & $=(0,1)$ &$=(-1,0) $ & $=(1,1)$ & $=(2,1)$ &…\\  \hline \hline 
 $(\xi,\eta) =(2,2)$&$ \mathbf{(2,2)} $ & $(4,-1)$& $(5,5)$& (0,5)&$(7,2)$ & $(9,-1)$  & …  \\
  \hline
  $(\xi,\eta) =(2,1)$&$  \mathbf{(2,1)} $ & $(4,-2)$& $(5,4)$&(0.4)& $(7,1)$ & $(9,-2)$&  … \\
  \hline
  $(\xi,\eta) =(1,2)$&$ (1,2) $ & $ \mathbf{(3,-1)}$& $(4,5)$&(-1,5)& $(6,2)$ & $(8,-1)$&  … \\
  \hline
  $(\xi,\eta) =(2,3)$&$ (2,3) $ & $ \mathbf{(4,0)}$& $(5,6)$& (0,6)&$(7,3)$ & $(9,0)$&  … \\
   \hline
  $(\xi,\eta) =(4,-1)$&$ (4,-1) $ & $ (6,-4)$& $(7,2)$& $(\mathbf{2,2)}$&$(9,-1)$ & $(11,-4)$&  … \\
  \hline
   \hline
\end{tabular}
\caption{Values of $(\xi+m\cdot 2 + n\cdot 3,\eta-m\cdot 3+n\cdot 3$) for the cycloid $ \mathcal{C}( 2,3,3,3 ) $.}
\label{values}
 \end{center}
       \end{table}
       


Terms, as introduced for the Petri space, as the forward output place  $ \postnbfw{t_{\xi,\eta}} := \gsvw{\xi,\eta}$ also apply for cycloids. As each place has a unique input and output transition, the notions are extended to transitions, as for the following example. The \emph{forward output transition} of $t_{\xi,\eta}$ denotes the (unique) output transition of  the place $ \postnbfw{t_{\xi,\eta}} = \gsvw{\xi,\eta}$, namely the transition
$  {(\gsvw{\xi,\eta})}^{\ndot} = t_{\xi+1,\eta}$ (see Figure \ref{P-space+FD} b).

%
For proving the equivalence of two points in the Petri space the following procedure\footnote{The algorithm is implemented under \url{https://cycloids.de}.} is useful.


\begin{theorem}[\cite{Valk-2020}] \label{parameter}
\begin{itemize}
       \item [a)]    
Two points $\vec{x}_1, \vec{x}_2\in X_1$ are equivalent $\vec{x}_1 \equiv \vec{x}_2$ if and only if
 for the difference $\vec{v} := \vec{x_2}-\vec{x_1}$
 the parameter vector 
 $\pi(\vec{v}) = \frac{1}{A} \cdot\mathbf{B} \cdot \vec{v}$ has integer values, where $A$ is the area and 
 $\mathbf{B} = \begin{pmatrix} \delta & -\gamma \\ \beta & \alpha \end{pmatrix}$.
 \item [b)] 
 In analogy to Definition \ref{cycloid} we obtain 
 $\vec{x}_1 \equiv \vec{x}_2 \Leftrightarrow$ 
 $\exists \;m, n \in \Int: \vec{x_2}-\vec{x_1} =\mathbf{A}\begin{pmatrix} m \\ n \end{pmatrix}$.
\end{itemize}  
\end{theorem}

\begin{shortproof}
For  $\vec{x}_1 := (\xi_1,\eta_1),
 \vec{x}_2 := (\xi_2,\eta_2), 
\vec{v} := \vec{x}_2 - \vec{x}_1$
from Definition \ref{cycloid} we obtain \ in vector form: 

\begin{align*}
\vec{x}_1 \equiv \vec{x}_2 &\Leftrightarrow 
\exists \, m,n \in \Int: \begin{pmatrix} \xi_2 \\ \eta_2  \end{pmatrix} =  
\begin{pmatrix} \xi_1 +m\cdot\alpha +n\cdot\gamma \\ \eta_1-m\cdot\beta + n\cdot\delta  \end{pmatrix} \\
& \Leftrightarrow
\exists \, m,n \in \Int:
\vec{v}=\begin{pmatrix} \xi_2 -\xi_1\\ \eta_2 - \eta_1 \end{pmatrix} = \begin{pmatrix}  m\cdot\alpha+n\cdot\gamma\\ -m\cdot\beta+n\cdot\delta  \end{pmatrix} = \begin{pmatrix} \alpha & \gamma \\ -\beta & \delta \end{pmatrix} \begin{pmatrix} m \\ n  \end{pmatrix} = \mathbf{A}\begin{pmatrix} m \\ n  \end{pmatrix} \\
& \Leftrightarrow
\begin{pmatrix} m \\ n  \end{pmatrix} = \mathbf{A}^{-1}\vec{v}
\in \Int \times \Int . 
\end{align*}
%
It is well-known that $\mathbf{A}^{-1} = \frac{1}{det(\mathbf{A})}\mathbf{B}$ 
if $det(\mathbf{A}) > 0$ (see any book on linear algebra). 
The condition $det(\mathbf{A}) = A = \alpha\cdot\delta + \beta\cdot\gamma >0$ is satisfied by the definition of a cycloid.
 \end{shortproof}
 

Taking up the example of the equivalent points $\vec{x}_1 = (4,-1)$ and $ \vec{x}_2 = (2,2)$ from  
$ \mathcal{C}( 2,3,3,3 ) $ in Figure \ref{pspace}, we calculate $\vec{v} := \vec{x_2}-\vec{x_1} = (-2,3)$ and
$\pi(\vec{v}) = \frac{1}{A} \cdot\mathbf{B} \cdot \vec{v} = \frac{1}{15}\cdot  \begin{pmatrix} 3 & -3 \\ 3 & 2 \end{pmatrix} \cdot  \begin{pmatrix} -2 \\ 3 \end{pmatrix}
 = 
 \begin{pmatrix} -1 \\ 0 \end{pmatrix}$. The result agrees with the entry of the fifth column and sixth line of Table \ref{values}.

The shear mappings\footnote{We would like to thank Olaf Kummer for pointing out these isomorphisms.}
for cycloids defined below form the basis of this article.
We give here a proof using cycloid algebra, which was not yet known when the article \cite{Valk-2019} had been published.

\begin{theorem}[shear mappings \cite{Valk-2019}]\label{shear}
The following cycloids are net isomorphic to $\mathcal{C}(\alpha,\beta,\gamma,\delta) $ (Definition \ref{def-net}) :\begin{itemize}
     \item [a)] $\mathcal{C}(\alpha,\beta,\gamma -  \alpha,\delta+\beta) $ if  $\gamma >  \alpha$, 
      \item [b)] $\mathcal{C}(\alpha,\beta,\gamma +  \alpha,\delta- \beta) $ if $\delta >  \beta$.
      \item [c)] $\zyk(\alpha-\gamma,\beta+\delta,\gamma,\delta) $ if $\alpha > \gamma$, 
       \item [d)]  $\zyk(\alpha+\gamma,\beta-\delta,\gamma,\delta) $ if $\beta > \delta$,
       \item [e)] $\mathcal{C}(\beta,\alpha,\delta, \gamma)$ (the $symmetric \; cycloid$ of $\mathcal{C}(\alpha,				\beta,\gamma,\delta) $).
\end{itemize}
\end{theorem}

\begin{proof} 
We give a bijection on the Petri space, which is a congruence with respect to equivalence. 
Let be $ \mathcal{C} = \mathcal{C}( \alpha, \beta, \gamma, \delta ) $ with matrix
$\mathbf{A}$ (Definition \ref{cycloid}) and the vector $\overrightarrow{mn} := (m,n) \in \Int^2$.\\
\begin{itemize}
\item[-] a) and b): The bijection is the identity map and we prove that the equivalence relation of 
$ \mathcal{C}_1 = \mathcal{C}_1( \alpha, \beta, \gamma \pm \alpha, \delta \mp \beta) $ 
with 
  matrix 
 $\mathbf{A}_1 = \begin{pmatrix} \alpha &\; \gamma \pm \alpha \\ -\beta & \;\delta \mp \beta \end{pmatrix}$  
remains unchanged:
 \begin{align*}
 \mathbf{A}_1 \cdot \overrightarrow{mn} &= 
 \mathbf{A} \cdot \overrightarrow{mn} + \begin{pmatrix} 0 & \;\; \pm \alpha \\ 0 & \;\; \mp \beta \end{pmatrix}\cdot \overrightarrow{mn} 
 =  \mathbf{A} \cdot \overrightarrow{mn} + \begin{pmatrix} \pm n \cdot \alpha \\ \mp n\cdot\beta \end{pmatrix} 
 \\ &= \mathbf{A} \cdot \overrightarrow{mn} + \mathbf{A} \cdot \begin{pmatrix} \pm n \\ 0 \end{pmatrix} =
 \mathbf{A} \cdot \begin{pmatrix} m \pm n \\ n \end{pmatrix}.
\end{align*}
 Hence, the by Theorem \ref{parameter} b) the equivalence relations of $ \mathcal{C}$ and $ \mathcal{C}_1 $ are the same, since $m$ and $n$ are integers if and only if $m \pm n$ and $n$ are integers.\\
\item [-] c) and d) are similar to a) and b).\\
\item[-] e): We denote $ \mathcal{C}_2 = \mathcal{C}_2( \beta,\alpha,\delta,\gamma ) $ with matrix
  $\mathbf{A}_2 = \begin{pmatrix} \beta &\; \delta\\ -\alpha & \gamma \end{pmatrix}$.
  Using the sets $X$ and $X_2$ of $ \mathcal{C}$ and $ \mathcal{C}_2$, respectively (Definition \ref{def-net}),
  the isomorphism is defined by $\varphi(x_{\xi,\eta}) := x_{\eta+\beta,\xi-\alpha}$. 
  Obviously, $\varphi$ is injective and surjective. 
  In the following we use the indices as coordinates of the points in the Petri space and write 
  $\varphi \begin{pmatrix} \xi\\ \eta \end{pmatrix} =  \begin{pmatrix} \eta+\beta\\ \xi-\alpha \end{pmatrix} $. It remains to prove that $\varphi$ is a congruence, i.e. 
  $$\begin{pmatrix} \xi\\ \eta \end{pmatrix} \equiv \begin{pmatrix} \xi_1\\ \eta_1 \end{pmatrix} \; \Rightarrow \; \varphi\begin{pmatrix} \xi\\ \eta \end{pmatrix} \equiv \varphi\begin{pmatrix} \xi_1\\ \eta_1 \end{pmatrix} $$
 For the precondition of this implication we have  by Theorem \ref{parameter}   the following  Equation (\ref{shear1}):
 \begin{equation}\label{shear1}
 \begin{pmatrix} \xi\\ \eta \end{pmatrix} \equiv \begin{pmatrix} \xi_1\\ \eta_1 \end{pmatrix} \; \Leftrightarrow \;  \begin{pmatrix} \xi-\xi_1 \\ \eta - \eta_1 \end{pmatrix} = 
 \mathbf{A}\begin{pmatrix} m\\ n  \end{pmatrix} = \begin{pmatrix} \alpha \cdot m + \gamma \cdot n \\ -\beta \cdot m + \delta \cdot n  \end{pmatrix} \;\;\; (\text{for some} \;m,n \in \Int).
 \end{equation}
 We use this term to prove the conclusion: \\
\[
\varphi\begin{pmatrix} \xi \\ \eta  \end{pmatrix} \equiv \varphi\begin{pmatrix} \xi_1\\ \eta_1  \end{pmatrix}  
 \; \Leftrightarrow \; 
 \begin{pmatrix} \eta+\beta \\ \xi-\alpha  \end{pmatrix}  - \begin{pmatrix} \eta_1+\beta \\ \xi_1-\alpha 
 \end{pmatrix} =
 \begin{pmatrix} \eta - \eta_1 \\ \xi - \xi_1  \end{pmatrix} =
 \mathbf{A}_2 \begin{pmatrix} m' \\ n' \end{pmatrix}  = 
 \begin{pmatrix} \beta \cdot m' + \delta \cdot n'\\ -\alpha \cdot m' + \gamma \cdot n'  \end{pmatrix} 
 \]
 Taking $m' := -m$ and $n' := n$ by Equation (\ref{shear1}) the last vector belongs to $\Int^2$.
 \end{itemize}
\end{proof}


In plane geometry, a \emph{shear mapping} or \emph{stretching} \cite[p.~149]{Strang}, is a linear map that displaces each point in a fixed direction, by an amount proportional to its signed distance from the line that is parallel to that direction and goes through the origin. 
For a cycloid $ \mathcal{C}( \alpha, \beta, \gamma, \delta ) $ the corners of its fundamental parallelogram have the coordinates: 
\[O = \begin{pmatrix} 0 \\ 0 \end{pmatrix}, 
P = \begin{pmatrix} \alpha \\ -\beta \end{pmatrix},
R = \begin{pmatrix} \alpha+\gamma \\ \delta-\beta \end{pmatrix} \text{ and }
Q = \begin{pmatrix} \gamma \\ \delta \end{pmatrix}.
\]
 Comparing them with the corners $O',P',R',Q'$ of the transformed cycloid
$\mathcal{C}(\alpha,\beta,\gamma + \alpha,\delta- \beta) $ of Theorem \ref{shear} b) we observe $O' = O, P' = P, Q' = \begin{pmatrix} \gamma + \alpha\\ \delta -\beta\end{pmatrix}=R$ and 
$R' = \begin{pmatrix} 2\cdot \alpha +\gamma \\ \delta -2 \cdot\beta\end{pmatrix}$.
The lines $\overline{Q,R}$ and 
$\overline{Q',R'}$ are the same. Therefore the second is a shearing of the first one.
This is shown in Figure\footnote{The figure has been designed using the tool \url{http://cycloids.adventas.de}.}~\ref{OP-shearing} for the cycloids $ \mathcal{C}( 2,3,2,8)$ and $\mathcal{C}( 2,3,4,5) $.
The point $P''$ is for later reference in the proof of Theorem \ref{shearing}.

\begin{figure}[htbp]
	\begin{center}
		\includegraphics [scale = 0.32]{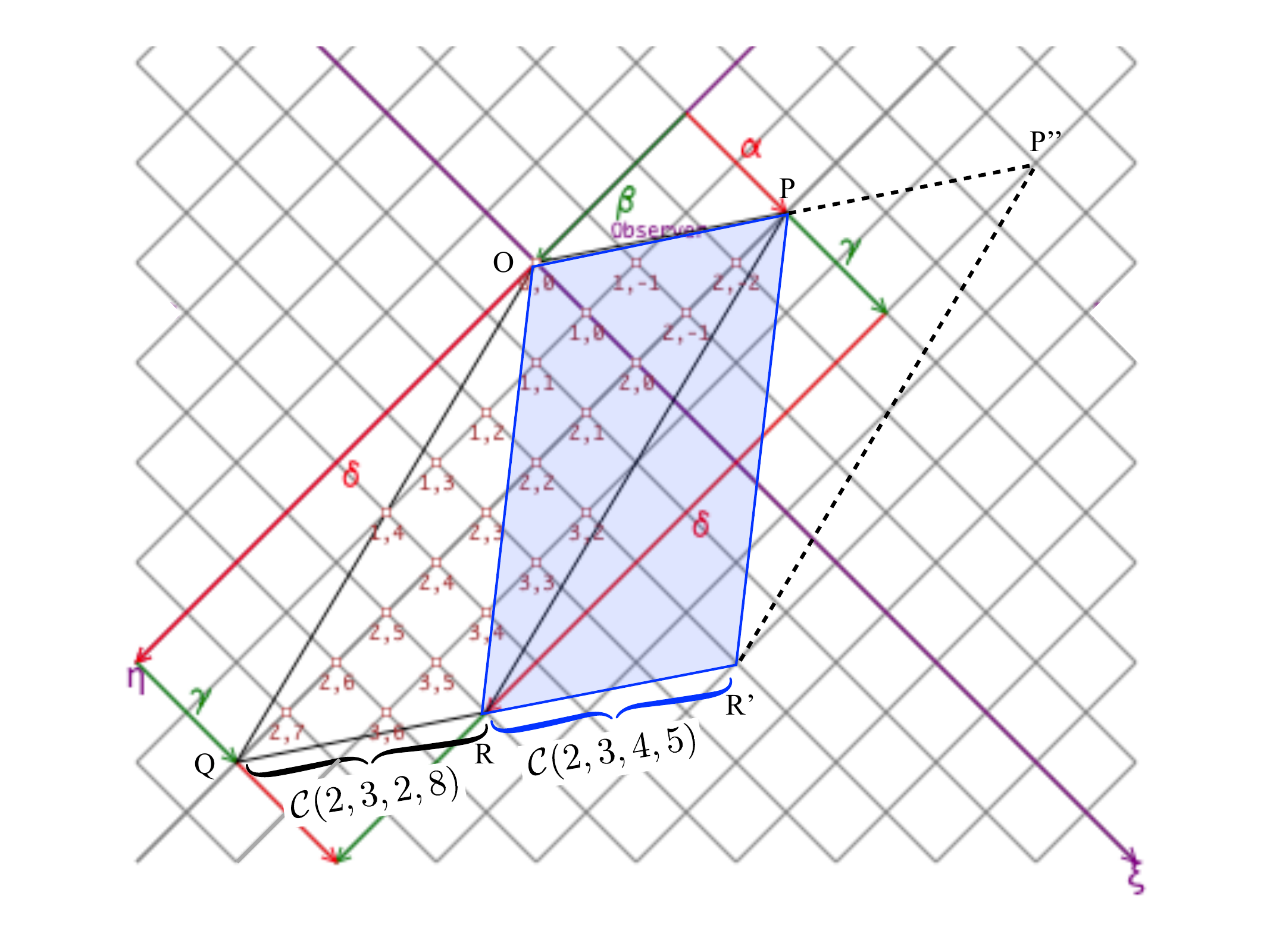}
		\caption{A shear mapping from  $ \mathcal{C}( 2,3,2,8)   $ to $ \mathcal{C}( 2,3,4,5) $.}
		\label{OP-shearing}
	\end{center}
\end{figure}

%

\begin{definition} \label{def-f-b-cycle}
 A \emph{forward-cycle} of a cycloid is an elementary\footnote{An elementary cycle is a cycle where all nodes are different.} cycle containing only forward places of $S_1^\rightarrow$.
 A \emph{backward-cycle} of a cycloid is an elementary cycle containing only
 backward places of $S_1^\leftarrow$ (Definition \ref{petrispace}). 
\end{definition}

\begin{theorem} \label{th-f-b-cycle}
In a cycloid   $ \mathcal{C}( \alpha, \beta, \gamma, \delta ) $ with area $A$ the length of a forward-cycle is 
$p = \frac{A}{gcd(\beta,\delta)} $ and length of a backward-cycle is $p' =\frac{A}{gcd(\alpha,\gamma)}$. 
\end{theorem} 

\begin{proof} By the symmetry of a cycloid all  forward-cycles have the same length.  Therefore it is sufficient to consider the  forward-cycle starting in the origin $(0,0)$. It is given by proceeding on the $\xi-$axis until  for the first time
a point $(\xi,0)$ is obtained, which is equivalent to the origin. By Theorem \ref{parameter} a necessary and sufficient condition for 
$\begin{pmatrix}0 \\ 0  \end{pmatrix}  \equiv \begin{pmatrix} \xi \\  0 \end{pmatrix}$ is the following equation: 
\[
\pi(\begin{pmatrix}\xi \\ 0  \end{pmatrix}  -  \begin{pmatrix} 0\\  0 \end{pmatrix}) = 
\frac{1}{A} \begin{pmatrix} \delta & -\gamma \\ \beta & \alpha \end{pmatrix}  \begin{pmatrix}\xi \\ 0  \end{pmatrix}  =
\frac{1}{A}  \begin{pmatrix}\delta \cdot \xi \\ \beta \cdot \xi   \end{pmatrix} \in \Int^2. 
\]
This is equivalent to 
$\frac{\delta}{A} \cdot \xi \in \Int \; \land \frac{\beta}{A} \cdot \xi \in \Int $. $\xi = A$ is a solution in $\frac{\delta}{A} \cdot \xi \in \Int $, but $\xi = \frac{A}{gcd(A,\delta)} $ is minimal. The same with  $\xi = \frac{A}{gcd(A,\beta)} $ for $\frac{\beta}{A} \cdot \xi \in \Int $ and both together give
$ \xi = \frac{A}{\omega} $ with $\omega = gcd(gcd(\beta,A),gcd(\delta,A))$.  To finish the proof it is sufficient to prove that $\omega$ equals
$ gcd(\beta,\delta) $. To this end we first prove: 
\begin{equation}\label{Gl1}
gcd(\beta,\alpha \cdot \delta + \beta \cdot \gamma) =   gcd(\beta,\alpha \cdot \delta)
 \end{equation}
 In fact, if $q$ divides $\beta$ and $\alpha \cdot \delta + \beta \cdot \gamma$ then $q$ divides $\alpha \cdot \delta$.
Conversely, if $q$ divides $\beta$ and $\alpha \cdot \delta$ then $q$ divides $\alpha \cdot \delta + \beta \cdot \gamma$.
In the same way the following equation holds.
\begin{equation}\label{Gl2}
gcd(\delta,\alpha \cdot \delta + \beta \cdot \gamma) =   gcd(\delta,\beta \cdot \gamma)
 \end{equation}
 The next equation
 \begin{equation}\label{Gl3} 
 gcd(gcd(\beta,\alpha \cdot \delta),gcd(\delta,\beta \cdot \gamma)) = gcd(\beta,\delta)
 \end{equation}
  is proved by the following true logical formula: 
  $q | \beta  \land q | \alpha \cdot \delta \land q | \delta  \land q | \beta \cdot \gamma \; \Leftrightarrow \; q | \beta  \land  q | \delta $.
  Using these equations we obtain: 
  \[
  \omega 
  {\underset{\text{}}{=}} 
  gcd(gcd(\beta,\alpha \cdot \delta + \beta \cdot \gamma),gcd(\delta,\beta,\alpha \cdot \delta + \beta \cdot \gamma))  {\underset{\text{(\ref{Gl1})(\ref{Gl2})}}{=}} 
  gcd(gcd(\beta,\alpha \cdot \delta),gcd(\delta,\beta \cdot \gamma))  {\underset{\text{(\ref{Gl3})}}{=}} gcd(\beta,\delta).
  \]
  The case of backward-cycles is proved in a similar way. Alternatively, we use the isomorphism of Theorem \ref{shear} e) which gives 
   $ \mathcal{C}( \alpha, \beta, \gamma, \delta )$  $ \simeq  \mathcal{C}(  \beta, \alpha, \delta, \gamma ) $ and derive the second part from the first part by the substitution $\alpha \mapsto \beta, \beta \mapsto \alpha, \gamma \mapsto \delta,\delta \mapsto \gamma $.
 \end{proof} 

%


For the cycloids  $ \mathcal{C}( 4,3,3,3 ) $ and $ \mathcal{C}( 4,6,3,3 ) $ from the introduction we obtain $p = 7$  and $p = 10$, respectively.
An important class of cycloids has the property to represent a number of sequential processes of the same length.
Such a cycloid is called regular.

\begin{definition} \label{regular} 
A cycloid $ \mathcal{C} = \mathcal{C}( \alpha, \beta, \gamma, \delta ) $ is \emph{regular} if $\beta$ divides $\delta$.
It consists of a number $\beta$ forward-cycles (called \emph{processes}) of length $p = \frac{A}{\beta} $.
 $ \mathcal{C} $ is called \emph{co-regular} if $\alpha$ divides $\gamma$.
 Then it consists of a number $\alpha$ backward-cycles (called \emph{co-processes}) of length $p = \frac{A}{\alpha} $. $ \mathcal{C} $ is called \emph{canonical regular} if $\beta=\gamma=\delta$.
\end{definition}

The cycloid  $ \mathcal{C}( 4,3,3,3 ) $  from the introduction is canonical regular, whereas and $ \mathcal{C}( 4,6,3,3 ) $ is not regular. 
In applications, a cycloid can appear as a net without the parameters  $ \alpha, \beta, \gamma$ and $ \delta  $ being visible. However, these are required to use algorithms for cycloids. To determine these parameters, properties of the net such as number of transitions, length of the shortest cycle, etc. must be calculated. The following theorem can be used for this purpose.
%
The version of this theorem from  \cite{Valk-2019} does not cover all cases, but the special case 
of canonical regular cycloids ($\gamma = \delta$) is still valid (see case c) of the following theorem). This subcase was important for the applications to regular cycloids in \cite{Valk-2020}. 
If we assume that a net is a cycloid, then the question remains open as to where this knowledge comes from. Since cycloids are finite graphs and their syntax is unambiguous, we conjecture that this property is decidable. However, formalising such a procedure is beyond the scope of this work.


%
 Lemma \ref{normal form} prepares the following proof of a theorem.

\begin{lemma} [\cite{Valk-2019}] \label{normal form}
For any cycloid $\zyk(\alpha,\beta,\gamma,\delta) $ there is a minimal cycle containing the origin $O$ in its fundamental parallelogram representation.
\end{lemma}
 
\begin{figure}[htbp]
 \begin{center}
        \includegraphics [scale = 0.32]{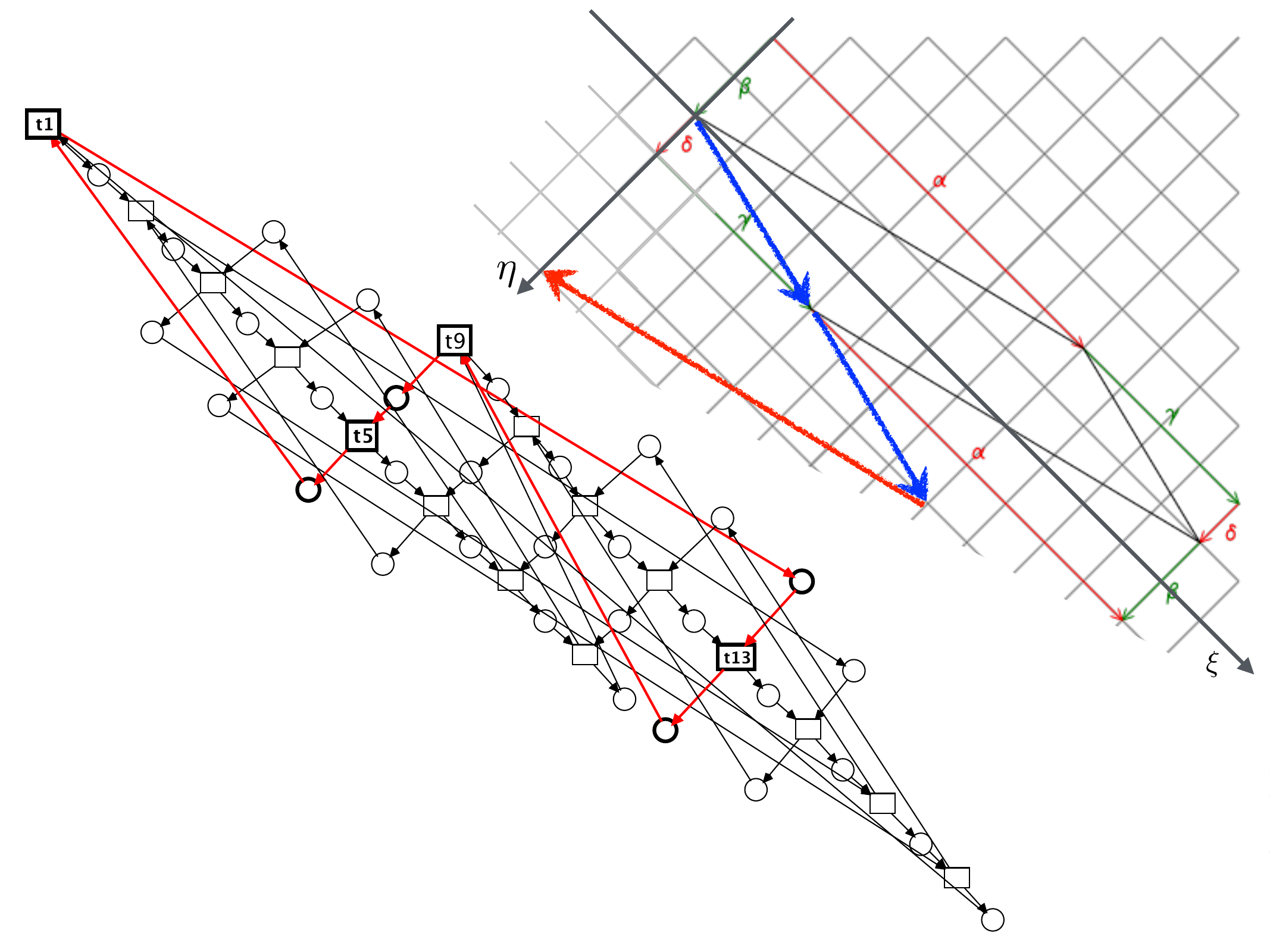}
        \caption{ $ \mathcal{C}( 8,2,4,1 ) $ with minimal cycle of length 4. }
        \label{8-2-4-1}
      \end{center}
\end{figure}


The proof of Theorem \ref{minimal cycles} uses the representation of cycles by paths in the Petri space whose endpoints are equivalent.
Therefore a transition 
$t_{u,v}$ equivalent to the origin is sought in the Petri space. The search can be restricted to points different from the origin and with $u \geq 0$ and $v \geq 0$, since only such transitions are reachable on paths. 
The application of cycloid algebra shows that such transitions are linear combinations of the vectors $(\alpha,-\beta)$ and $(\gamma,\delta)$. All these transitions are equivalent to the corners $O,P,Q$ and $R$ of the fundamental parallelogram. The length $u+v$ of a path from the origin to such a transition has to be minimal. Petri observed that each cycloid has a cycle of length $\gamma+\delta$ via the corner $Q$ and called it the \emph{local basic circuit}  \cite{Petri-HH-lec}. 
As shown by the cycloid
$ \mathcal{C}( 8,2,4,1 ) $ of Figure \ref{8-2-4-1} this  value $\gamma+\delta=5$ is not necessarily minimal. 
The cycle \textbf{t1} \textbf{t13} \textbf{t9} \textbf{t5} of minimal length $cyc = 4$ is obtained by the linear combination  $\begin{pmatrix} u \\ v \end{pmatrix}  = j \cdot \begin{pmatrix} \gamma \\\delta \end{pmatrix} + i \cdot \begin{pmatrix} \alpha \\ - \beta \end{pmatrix} = \begin{pmatrix} 0 \\ 4 \end{pmatrix} $ with $j=2$ and $i=-1$.
This is indicated by two small  and one large arrow        in the fundamental parallelogram
of Figure \ref{8-2-4-1}.
The minimisation over two parameters $i$ and $j$ 
can be reduced to one bounded parameter $j$ in part b) of the Theorem.
 By restricting to particular cases in the remaining cases no minimum operator is needed. 


\begin{theorem}\label{minimal cycles}
The length of a minimal cycle of a cycloid $\zyk = \mathcal{C}(\alpha,\beta,\gamma,\delta)$ is 
$cyc(\alpha,\beta,\gamma,\delta) $ (denoted $cyc$ in the following), where
\begin{itemize}
\item [a)]
                      $ cyc = min \left\{u+v \;|\begin{pmatrix} u \\ v  \end{pmatrix} = \mathbf{A} \cdot
                      \begin{pmatrix} i \\ j  \end{pmatrix}, \; i \in \Int, \;j  \in \Nat , \; u  \geq 0, \;    v  \geq
0\right\} $ 
       \item [b)] $ cyc = min \left\{j\cdot (\gamma+ \delta)+i\cdot(\alpha-\beta) \;|\;j  \in \Nat, \: 
       i = \left\{
	                     \begin{array}{lll}
		                                    \lfloor\frac{j \cdot \delta}{\beta}\rfloor  &  \textrm{if} &\alpha \leq \beta  \\
		                                   -\lfloor\frac{j \cdot\gamma}{\alpha}\rfloor  & \textrm{if} & \alpha > \beta   
	                        \end{array}
              \right\} 
                     \right\} $ \\
                     The value of $j$ is bounded: $j  \leq \frac{A}{\gamma} $ if $\alpha  \leq \beta$ and $j  \leq \frac{A}{\delta} $ otherwise.
       \item [c)]  
                     $cyc = \gamma + \delta + 
                     \; \left\{
	                     \begin{array}{lll}
		            \lfloor\frac{\delta}{\beta}\rfloor (\alpha - \beta) &  \textrm{if} &\alpha \leq \beta  \;\;\textrm{and}  \;\; \gamma  \geq \delta\\
		            -\lfloor\frac{\gamma}{\alpha}\rfloor (\alpha - \beta) & \textrm{if} & \alpha > \beta   \;\;\textrm{and} \;\; \gamma   \leq \delta
	                \end{array}
                     \right\} $ 
         \item [d)] $cyc = \gamma + \frac{\delta}{\beta}\cdot \alpha = \frac{A}{\beta} \;\;\textrm{if}  \;\;  \alpha  \leq \beta 
         \;\;\textrm{and} \;\; \zyk \;\; \textrm{is regular} \;\; (\textrm{i.e.}\;\; \beta|\delta )$
         \item [e)]   $cyc = \delta + \frac{\gamma}{\alpha}\cdot \beta = \frac{A}{\alpha} \;\;\textrm{if}  \;\;  \alpha  > \beta 
         \;\;\textrm{and} \;\; \zyk \;\; \textrm{is co-regular} \;\; (\textrm{i.e.}\;\; \alpha|\gamma )$          
         \end{itemize}           
\end{theorem}


\begin{proof}
\begin{itemize}
\item[a)] 
With respect to paths and cycles in the fundamental parallelogram and by Lemma \ref{normal form} it is sufficient to consider paths  starting in the origin $O$.
Such a cycle of the cycloid corresponds to a path with positive length from $O$ to an equivalent point $\vec{x}$ in the Petri space.
From Theorem \ref{parameter}  we obtain with $\vec{x_2} = \vec{x}$ and $\vec{x_1} = (0,0) $  the necessary and sufficient condition 
$\exists \;i,j \in \Int : \vec{x} =  \mathbf{A} \cdot \begin{pmatrix} i \\ j  \end{pmatrix}$ with $\lnot (i = 0 \land j=0)$. If $\vec{x}= \begin{pmatrix} u \\ v  \end{pmatrix} $  
then $u+v > 0$ is the length of the path from the origin $O$ to the endpoint of $\vec{x}$. $u+v$ should be a minimum to obtain $cyc$. However, some choices of $\vec{x}$ can be excluded. There is no path from $O$ to 
$(u,v)$ if $u < 0$ or $v < 0$. Therefore $j \leq 0$ can be excluded in $\begin{pmatrix} u \\ v  \end{pmatrix} = \mathbf{A} \cdot 
                      \begin{pmatrix} i \\ j  \end{pmatrix} = \begin{pmatrix} \alpha & \gamma \\ -\beta & \delta \end{pmatrix}\cdot\begin{pmatrix} i \\ j \end{pmatrix} =
                      \begin{pmatrix} i\cdot \alpha + j \cdot \gamma \\ -i \cdot \beta+ j \cdot \delta  \end{pmatrix} $. 
This is true by the following proof by contradiction:  assume $j   \leq 0$. \\
Case 1: If $i \geq 0$ then $ v = -i\cdot\beta+j\cdot\delta < 0$ in contradiction to the condition $v \geq 0$.\\
Case 2: If $i < 0$ then $ u = i\cdot\alpha+j\cdot\gamma < 0$ in contradiction to the condition $u \geq 0$.\\
\item[b)]
We first consider the case $\alpha \leq \beta$ and prove $ \lfloor\frac{j \cdot \delta}{\beta}\rfloor $ if $\gamma  \geq \delta$. 
Denote the cut point of the line $\overline{QR}$ with the $\xi$-axis by $A_1$ (the line cannot be in parallel to the $\xi$-axis). Next, in a similar way, for $j  \geq 1 $ the endpoint of the vector $j\cdot (\gamma,\delta)$ is denoted by $Q_j$, including $Q_1 = Q$. (See Figure \ref{cyc-theorem} for the cases $j \in \{ 1,2,3 \}$.)
Furthermore we name the cut point of the line through $Q_j$ and  the endpoint of  $j\cdot (\gamma,\delta)+ (\alpha, -\beta)$ with the $\xi$-axis by $A_j$.
On this line the points $\vec{x} = j \cdot \begin{pmatrix} \gamma \\\delta \end{pmatrix} + i \cdot \begin{pmatrix} \alpha \\ - \beta \end{pmatrix} = \begin{pmatrix} u \\ v \end{pmatrix} $  are situated which define the value of $cyc = u+v$. By the condition  $v  \geq 0$ we obtain $j \cdot \delta - i \cdot \beta  \geq 0$ which is
\begin{equation}\label{GL1}
i  \leq \frac{j \cdot \delta}{\beta} 
 \end{equation}
Next we derive an expression for $i$ in dependance of $j$ by proving
that increasing the value of $i$ does not increase the distance to the origin (while the condition
$\eta \geq 0$ is not violated when going $\beta$ steps in direction $-\eta$).
\begin{figure}[htbp]
	\begin{center}
		\includegraphics [scale = 0.31]{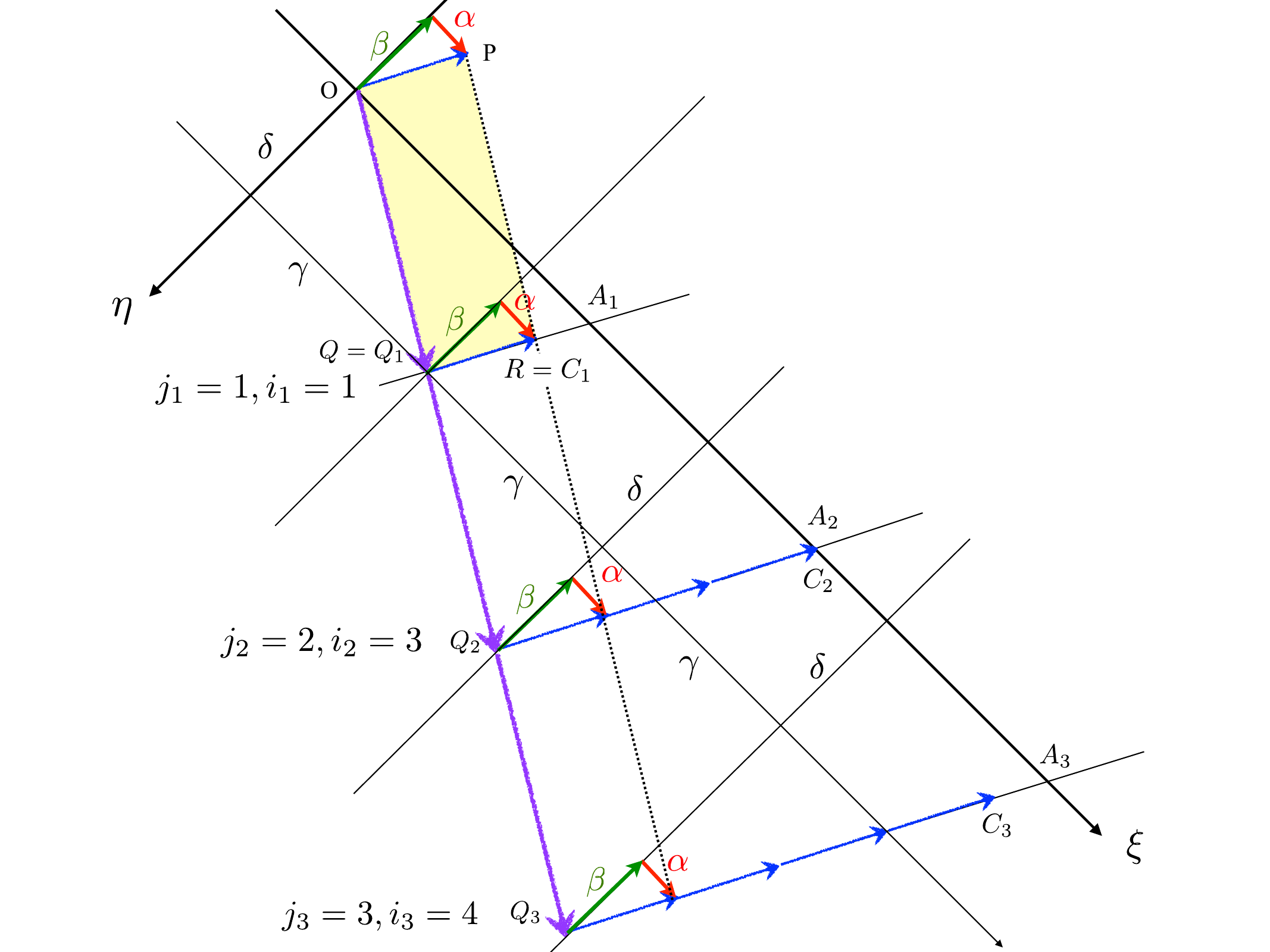}
		\caption{Referenced in the proof of Theorem \ref{minimal cycles}.}
		\label{cyc-theorem}
	\end{center}
\end{figure}
More precisely, 
for any $\xi \geq 0, \eta \geq 0$ we have to prove 
$d(O,\begin{pmatrix} \xi \\ \eta \end{pmatrix}) \geq d(O,\begin{pmatrix} \xi \\ \eta \end{pmatrix} + \begin{pmatrix} \alpha \\ -\beta \end{pmatrix} ) $ under the condition $\eta - \beta \geq 0$. 
This follows from
$\alpha \leq \beta$ by $0 \geq \alpha - \beta 
\; \Rightarrow \; \xi + \eta \geq \xi +\alpha + \eta- \beta 
\; \Rightarrow \; |\xi + \eta | \geq | \xi +\alpha | + |\eta- \beta |
\; \Rightarrow \; d(O,\begin{pmatrix} \xi \\ \eta \end{pmatrix}) \geq d(O, \begin{pmatrix}\xi + \alpha \\ \eta-\beta \end{pmatrix} )$.
By choosing the maximal value of $i$ under the inequality (\ref{GL1})  we obtain $ i = \lfloor \frac{j \cdot \delta}{\beta}\rfloor$.
Therefore the candidates to compute $cyc$ are the endpoints of the vectors 
\begin{equation}\label{vec-x-j-a}
\vec{x}{_j} = j \cdot \begin{pmatrix} \gamma \\\delta \end{pmatrix} + \lfloor \frac{j \cdot \delta}{\beta}\rfloor \cdot \begin{pmatrix} \alpha \\ - \beta \end{pmatrix}  \;\;\;\;\;\; (j \in \Nat). 
 \end{equation} 
 The distance from the origin $O$ to $\vec{x}{_j} $ is $d(O,\vec{x}{_j}) = $
 $j\cdot (\gamma+ \delta)+\lfloor\frac{j \cdot \delta}{\beta}\rfloor \cdot(\alpha-\beta) $. 
For the alternative case $\alpha > \beta$  we look at the symmetric cycloid $\zyk(\beta,\alpha,\delta,\gamma) $ (by interchanging $\alpha$ and $\beta$, as well as $\gamma$ and $\delta$), which is net isomorphic (Theorem \ref{shear} e) and therefore has a minimal cycle of the same length. Equation (\ref{vec-x-j-a})  is replaced by Equation (\ref{vec-x-j-b}):
\begin{equation}\label{vec-x-j-b}
\vec{x}{_j} = j \cdot \begin{pmatrix} \delta \\\gamma \end{pmatrix} + \lfloor \frac{j \cdot \gamma}{\alpha}\rfloor \cdot \begin{pmatrix} \beta \\ - \alpha \end{pmatrix}  \;\;\;\;\;\; (j \in \Nat). 
 \end{equation} 
and we obtain
 $d(O,\vec{x}{_j})   = j \cdot (\gamma + \delta) + \lfloor\frac{j\cdot\gamma}{\alpha}\rfloor \cdot (\beta -\alpha)$ in  case of $\alpha > \beta$. 
%
 To derive the bound we start with the observation that the length of a cycle is bounded by the number $A$ of transitions. In the case $\alpha  \leq \beta$ it follows with respect to the minimal value of $j$:
 $ cyc = j\cdot (\gamma+ \delta)+\lfloor\frac{j \cdot \delta}{\beta}\rfloor \cdot(\alpha-\beta)   \leq A$ which transforms to 
 $  j\cdot \delta - \lfloor\frac{j \cdot \delta}{\beta}\rfloor \cdot\beta + j \cdot \gamma +  
 \lfloor\frac{j \cdot \delta}{\beta}\rfloor \cdot \alpha  = 
 (j\cdot \delta)\; mod \;\beta + j \cdot \gamma +  
 \lfloor\frac{j \cdot \delta}{\beta}\rfloor \cdot \alpha  
 \leq A$. Since $(j\cdot \delta)\; mod \;\beta +\lfloor\frac{j \cdot \delta}{\beta}\rfloor \cdot \alpha \ \geq 0$ we obtain $j \cdot \gamma  \leq A$ and $j  \leq \frac{A}{\gamma} $. The result for the case $\alpha   \geq \beta$ is proved in a similar way.

 %
\item[c)] 
We prove $j=1$ in case b) of the theorem under the the additional condition $\gamma  \geq \delta$.
From $\vec{x}{_j}  \geq (0,0)$ we deduce from Equation (\ref{vec-x-j-a}):
\begin{equation}\label{ineq-vec-x-j}
j \cdot \delta - \lfloor \frac{j \cdot \delta}{\beta}\rfloor \cdot \beta  \geq 0
 \end{equation} 
The endpoints of the vectors $\vec{x}{_j}$ are denoted by $C_j$ in Figure \ref{cyc-theorem}. The path from the origin $O$ to  $\vec{x}{_j}$ has the length
$cyc_j := j \cdot (\gamma+\delta) + \lfloor \frac{j \cdot \delta}{\beta}\rfloor \cdot (\alpha-\beta)$ and we next prove $cyc_j  \geq cyc_1$ for $j>1$ which shows that $j=1$ is the optimal solution for $cyc$. This is done by the inequality
\begin{align*}
cyc_j  - cyc_1 &= 
j \cdot (\gamma+\delta) + \lfloor \frac{j \cdot \delta}{\beta}\rfloor \cdot (\alpha-\beta) - (\gamma+\delta+ \lfloor \frac{ \delta}{\beta}\rfloor \cdot (\alpha-\beta)) \\
& = (j-1)\cdot \gamma + (j\cdot\delta-\lfloor \frac{j \cdot \delta}{\beta}\rfloor \cdot \beta) - \delta + (\lfloor \frac{j \cdot \delta}{\beta}\rfloor \cdot \alpha - \lfloor \frac{\delta}{\beta}\rfloor \cdot \alpha) +  \lfloor \frac{\delta}{\beta}\rfloor \cdot \beta  \\
& \geq (j-1)\cdot \gamma  - \delta  \geq \gamma - \delta  \geq 0. 
 \end{align*}
 The second summand in the second line is not negative due to the inequality (\ref{ineq-vec-x-j}). This holds obviously for the fourth summand 
 $\lfloor \frac{j \cdot \delta}{\beta}\rfloor \cdot \alpha - \lfloor \frac{\delta}{\beta}\rfloor \cdot \alpha  \geq 0$. In the last inequalities $j > 0$ and $\gamma  \geq \delta$ is used.
 It remains to prove that negative values of $i$ are not needed to compute $cyc$ in the cases under review. In the same way as before the inequality
 $d(O,\begin{pmatrix} \xi \\ \eta \end{pmatrix})  \leq d(O, \begin{pmatrix}\xi \\ \eta\end{pmatrix} -\begin{pmatrix}\alpha \\ -\beta \end{pmatrix} )$ is proved. This shows, that adding the vector $-(\alpha,-\beta)$ to $(\gamma,\delta)$ cannot decrease the distance from the origin. The cases for $j > 1$ is are similar.\\
For the alternative case $\alpha > \beta$ and $\delta  \geq \gamma$ we look at the symmetric cycloid $\zyk(\beta,\alpha,\delta,\gamma) $ (by interchanging $\alpha$ and $\beta$, as well as $\gamma$ and $\delta$).\\
\item[d]
If $\beta | \delta$ then Equation (\ref{vec-x-j-a}) becomes
\[
\vec{x}{_j} = j \cdot \begin{pmatrix} \gamma \\\delta \end{pmatrix} + \frac{j \cdot \delta}{\beta} \cdot \begin{pmatrix} \alpha \\ - \beta \end{pmatrix} = 
\begin{pmatrix} j\cdot \gamma+j\cdot\frac{\delta}{\beta} \cdot \alpha\\ j\cdot \delta+j\cdot\frac{\delta}{\beta} \cdot (-\beta) \end{pmatrix} = 
j \cdot\begin{pmatrix} \gamma+\frac{\delta}{\beta} \cdot \alpha\\ 0\end{pmatrix}. 
\]
 Since all the points for different $j$ are on the $\xi$-axis, for $j=1$ we obtain a minimal value of  $cyc = \gamma+\frac{\delta}{\beta} \cdot \alpha$.\\
\item[e)] 
Again, for the alternative case $\alpha > \beta$ and $\alpha  |  \gamma$ we look at the symmetric cycloid. 
\end{itemize}
\end{proof}

 The pattern of Figure \ref{cyc-theorem} is derived from  the cycloid  $ \mathcal{C}( 1,2,5,3 ) $. The point $C_3$ is computed by the following formula, as derived in the preceding proof: 
\[
\vec{x}{_3} = 3 \cdot \begin{pmatrix} \gamma \\\delta \end{pmatrix} + \lfloor \frac{3 \cdot \delta}{\beta}\rfloor \cdot \begin{pmatrix} \alpha \\ - \beta \end{pmatrix}  =
  3 \cdot \begin{pmatrix} 5 \\ 3 \end{pmatrix} + \lfloor \frac{3 \cdot 3}{2}\rfloor \cdot \begin{pmatrix} 1 \\ - 2 \end{pmatrix}  = \begin{pmatrix} 19 \\ 1  \end{pmatrix},
 \] 
 leading to 
  $cyc_3 = 20$. The values for $C_2$ and $C_1 = R$ are $\begin{pmatrix} 13 \\ 0  \end{pmatrix} $ and $\begin{pmatrix} 6 \\ 1  \end{pmatrix} $, respectively.

 Cycloids of minimal length are characteristic for the structure of cycloids. They can be identified from the cycloid net without knowing the fundamental parallelogram. On the other hand, cycloids where the minimal length $cyc$ of cycloids is given by 
 the Equation (\ref{basic}) cover a huge class of cycloids with favourable properties.
This formula specialises the general formula of Theorem \ref{minimal cycles} b) to the case $j=1$ and therefore avoids the use of the minimum operator. Due to the relation to the local basic circuit (see text before Theorem \ref{minimal cycles}) we call the class of \emph{lbc-cycloids}.

By a Pyton-Programm for all of the 
 $10^8$ cycloids with $1  \leq \alpha,\beta,\gamma,\delta  \leq 100$  the values of $i$  and $j$ in the formula of 
 Theorem \ref{minimal cycles} a) have been computed. As a result, 99 \% of these  belonged to 
 the class of \emph{lbc-cycloids}.
 To give a rather strange example of such cases, also all cycloids  
$ \mathcal{C}( 1,11,3, \delta ) $ with $1  \leq \delta  \leq 10^6 \land \delta \neq 6$ belong to this class.

  \begin{equation} \label{basic} 
cyc = \gamma + \delta + 
                     \; \left\{
	                     \begin{array}{lll}
		            \lfloor\frac{\delta}{\beta}\rfloor (\alpha - \beta) &  \textrm{if} &\alpha \leq \beta  \;\;\ \\
		            -\lfloor\frac{\gamma}{\alpha}\rfloor (\alpha - \beta) & \textrm{if} & \alpha > \beta   \;\;
	                \end{array}
                     \right\} 
\end{equation}

In Theorem \ref{minimal cycles} c) and d)  particular subclasses of lbc-cycloids are given.
  In \cite{Valk-2020} different types of circular traffic queues have been studied. It has been proven, that they are represented by cycloids  $ \mathcal{C}( g,c,1,1 ) $, $ \mathcal{C}( g,c,c,c ) $ and
  $ \mathcal{C}( g,c,\frac{g \cdot c}{\Delta} ,1\frac{g \cdot c}{\Delta} ) $, where $c$ is the number of cars, $g$ is the number of gaps in the queue and $\Delta = gcd(g,c)$. 
 By  Theorem \ref{minimal cycles} c) they are all lbc-cycloids. In this paper we also prove that the Equation (\ref{basic}) is invariant to some reductions of cycloids.

  \begin{definition} \label{lbc-cycloid} 
  A cycloid  $ \mathcal{C}( \alpha, \beta, \gamma, \delta ) $ is called \emph{lbc-cycloid} if the minimal length $cyc$ of its cycles is given by Equation (\ref{basic}).\end{definition}

\begin{definition} [\cite{Valk-2019}] \label{standard-M0}
For a cycloid $\mathcal{C}(\alpha,\beta,\gamma,\delta)$ 
we define a cycloid system
$\mathcal{C}(\alpha,\beta,\gamma,\delta,M_0)$ 
or $\mathcal{C}(\N{},M_0)$ by adding the standard initial marking: 
\begin{align*}
M_0 = & \col{ \gsvw{\xi,\eta} \in \GSvw[1] }{\, \beta\xi + \alpha\eta \,\leq\, 0 \ \land \ \beta(\xi+1) + \alpha\eta \,>\, 0 }/_\zykaeq \ \,\cup \\
& \col{\gsrw{\xi,\eta} \in \GSrw[1] }{\, \beta\xi + \alpha\eta \,\leq\, 0 \ \land \ \beta\xi + \alpha(\eta+1) \,>\, 0 }/_\zykaeq.
\end{align*} 
\end{definition}

\begin{lemma} [\cite{Valk-2019}]\label{le-standard-M0}
Given a cycloid system $\mathcal{C}(\alpha,\beta,\gamma,\delta,M_0)$ 
with standard initial marking $M_0$ then $| M_0 \cap S^\rightarrow |= \beta$ and 
$|M_0 \cap S^\leftarrow |= \alpha$.
\end{lemma} 

%
The following Synthesis Theorem for lbc-cycloids (Definition \ref{lbc-cycloid}) allows for a cycloid system, given as a net without the parameters
 $ \alpha, \beta, \gamma, \delta $, to compute these parameters. It does
not necessarily give a unique result, but for $\alpha \neq \beta$ the resulting cycloids are isomorphic. 
In the theorem $\tau_0 := |\{t |\;\;|{}^{\ndot} t\cap M_0| \geq 1 \;\}|$ is the number of initially marked transitions and $\tau_a := |\{t |\;\;|{}^{\ndot} t\cap M_0| = 2 \;\}|$ is the number of initially active transitions.
They are used to determine $\alpha$ and $\beta$. 
Cycloids with identical systems parameters $\tau_0, \tau_a$, area $A$ and minimal cycle length $cyc$
 are called $\sigma$-equivalent.
 While most of  the theorem is from \cite{Valk-2019},
the given relations of the different solutions are derived in the proof for this theorem in \cite{Valk-2019}.
\begin{theorem} [lbc Synthesis Theorem \cite{Valk-2019}] \label{synthesis}
Cycloid systems with identical system parameters $\tau_0, \tau_a,$ $ A$ and $cyc$ are called $\sigma$-$equivalent$.
Given a lbc-cycloid system $\mathcal{C}(\alpha,\beta,\gamma,\delta, M_0)$ (Definition \ref{lbc-cycloid}) in its net representation $(S,T,F,M_0)$ where the parameters $\tau_0, \tau_a, A$ and $cyc$ are known (but the parameters $\alpha,\beta,\gamma,\delta$ are not). Then a $\sigma$-$equivalent$ cycloid $\mathcal{C}(\alpha',\beta',\gamma',\delta')$ can be computed by $\alpha' = \tau_0$, $\beta' = \tau_a$ and
$\gamma'$ and $\delta'$ 
by some positive integer solution of the following formulas using these settings of $\alpha'$ and $\beta'$:\\
	a) case $\alpha' > \beta'$: $\gamma' \; mod\; \alpha' = \frac{\alpha'\cdot cyc - A}{\alpha'-\beta'}$ and 
	\\
	b) case $\alpha' < \beta'$: $\delta' \; mod\; \beta' = \frac{\beta'\cdot cyc - A}{\beta' - \alpha'}$ and 
\\
By  $A = \alpha \cdot \delta + \beta \cdot \gamma$ the missing parameters $\delta$ in a) and $\gamma$ in b) are obtained.
 These equations may result in different cycloid parameters. The solutions however are related by the 
 isomorphic transformations in Theorem \ref{shear}.
 If the distinction between 
 $S^\rightarrow$ and $S^\leftarrow $
is known Lemma \ref{le-standard-M0} can be used in place of $\tau_0$ and $ \tau_a$ .
\end{theorem}

To give an example, consider the values $\alpha = 4, \beta = 3, cyc = 12, A = 46$ (which are obtained from the  cycloid  $ \mathcal{C}( 4,3,6,7 ) $ ) then we compute
$\gamma' \; mod\; \alpha' = \frac{\alpha'\cdot cyc - A}{\alpha'-\beta'} = \frac{4\cdot 12 - 46}{4-3} = 2$ and 
$\gamma' = 2$ and $\delta' = \frac{A-\beta \cdot \gamma}{\alpha} =  \frac{46-3 \cdot 2}{4} = 10 $. The resulting cycloid  $ \mathcal{C}( \alpha', \beta', \gamma', \delta'  ) =   \mathcal{C}(4,3,2,10) $ is transformed by Theorem \ref{shear} b) to $ \mathcal{C}( 4,3,6,7 ) $.

 When working with cycloids it is sometimes important to find for a transition $t$ outside the fundamental parallelogram the unique equivalent element $\rho(t)$ inside. 
In general,
by enumerating all elements of the fundamental parallelogram (using Theorem 7 in \cite{Valk-2018}) and applying the equivalence test from Theorem \ref{parameter} a runtime is obtained which already fails for small cycloids.
The following theorem allows for a better algorithm\footnote{The algorithm is implemented under \url{https://cycloids.de}.}, which is linear with respect to the cycloid parameters when the area $A$ is given.
$\rho $ is a mapping defined on the whole Petri space with
range in the fundamental parallelogram of  $ \mathcal{C}$. With the interim result of Lemma \ref{xi-max} it is used in Theorem \ref{bd-irreducible} to prove properties of $\beta\delta$-irreducible cycloids.

\begin{theorem} \label{xy-to-FP}
For any element $\vec{u} = (u,v)$ of the Petri space the (unique) equivalent element within the fundamental parallelogram of a cycloid    $ \mathcal{C}=\mathcal{C}( \alpha, \beta, \gamma, \delta ) $ with matrix $\mathbf{A}$ is 
\begin{equation}\label{pi}
\rho(\vec{u}) = \vec{u} - \mathbf{A}\begin{pmatrix} m \\ n \end{pmatrix}
\end{equation}
where 
$m = \lfloor \frac{1}{A}(u\cdot\delta - v\cdot \gamma)\rfloor$ and $n = \lfloor \frac{1}{A}(v\cdot\alpha + u\cdot\beta)\rfloor$.\\
We write $\rho_\mathcal{C}(\vec{u})$ when the reference to the cycloid  $ \mathcal{C}$ is not obvious.
\end{theorem}

\begin{proof}
As \emph{m-sector} we denote the band between the lines $\overline{OQ} $ and $\overline{PR} $ including the points of $\overline{OQ} $, but not those of $\overline{PR} $ (see Figure \ref{u-v-to-FP}). The idea of the proof is to reach the \emph{m-sector} on a line parallel to $\overline{QR} $  starting in $u$  until a point $c$  is reached\footnote{$u, a, b, c, \cdots$ denote the points in the Petri space, where the vectors
$\vec{u},\vec{a},\vec{b},\vec{c},\cdots$ are pointing to them when originated in the origin $(0,0)$, respectively.}.
The point $c$ is the intersection of this line and the line through $x$ parallel to 
$\overline{OQ} $.
  Then $m,n$ are the integer multiple of the vectors $\begin{pmatrix}-\alpha  \\ \beta   \end{pmatrix}$,
$\begin{pmatrix}\gamma  \\ \delta   \end{pmatrix}$, resulting in  the line segments $u$ to $c$ and $c$ to $x$, respectively.

More formally we derive: $\vec{u} \equiv \vec{x} \Leftrightarrow \vec{u} - \vec{x} = \mathbf{A}\begin{pmatrix} m \\ n  \end{pmatrix}$ (by Theorem \ref{parameter}). Hence we obtain:  
\[
\vec{x} = \vec{u} - \mathbf{A}\begin{pmatrix} m \\ n  \end{pmatrix}  = \vec{u} - \begin{pmatrix} \alpha & \gamma \\ -\beta & \delta  \end{pmatrix} \begin{pmatrix} m \\ n  \end{pmatrix} = \vec{u} - \begin{pmatrix}m \cdot \alpha + n \cdot \gamma \\ -m \cdot \beta + n \cdot\delta   \end{pmatrix}  = \vec{u} +m \begin{pmatrix}-\alpha  \\ \beta   \end{pmatrix} - n \begin{pmatrix}\gamma  \\ \delta   \end{pmatrix} \]
or
\begin{equation}\label{Vektorgleichung 1}
\vec{x}+ n \begin{pmatrix}\gamma  \\ \delta   \end{pmatrix} = \vec{u} +m \begin{pmatrix}-\alpha  \\ \beta   \end{pmatrix}. 
\end{equation}
This vector equation results in two linear equations for four unknown $n,m$ and $\vec{x} = (x,y)$. 
With $\lambda \in \Real $ and $\vec{u} +\lambda \begin{pmatrix}-\alpha  \\ \beta   \end{pmatrix}  $ the right hand side of equation (\ref{Vektorgleichung 1}) defines the line through $u$ parallel to the line $\overline{QR} $.  
It is intersecting the line
$\overline{OQ} $ in point $d$ and the line $\overline{PR} $ in point $b$.
In the same way equation (\ref{Vektorgleichung 1})  defines with $\mu \in \Real $ and 
$\vec{x}+ \mu \begin{pmatrix}\gamma  \\ \delta   \end{pmatrix}  $ the line through $\vec{x}$ parallel to the line containing $O$ und $Q$.
Both lines intersect in $d$, which is defined by the equation
\begin{equation}\label{Vektorgleichung 2}
 \begin{pmatrix} 0\\ 0  \end{pmatrix} + \mu \begin{pmatrix} \gamma\\ \delta  \end{pmatrix} = \begin{pmatrix} u \\ v  \end{pmatrix} + \lambda\begin{pmatrix} -\alpha \\ \beta  \end{pmatrix} 
\end{equation}
This vector equation provides two linear equations with the solutions
$\lambda = \frac{1}{A}(u\delta - v \gamma)$ und 
$\mu =\frac{1}{A}(v\alpha + u\beta) $. 
%
%
\begin{figure}[t]
 \begin{center}
        \includegraphics [scale = 0.40]{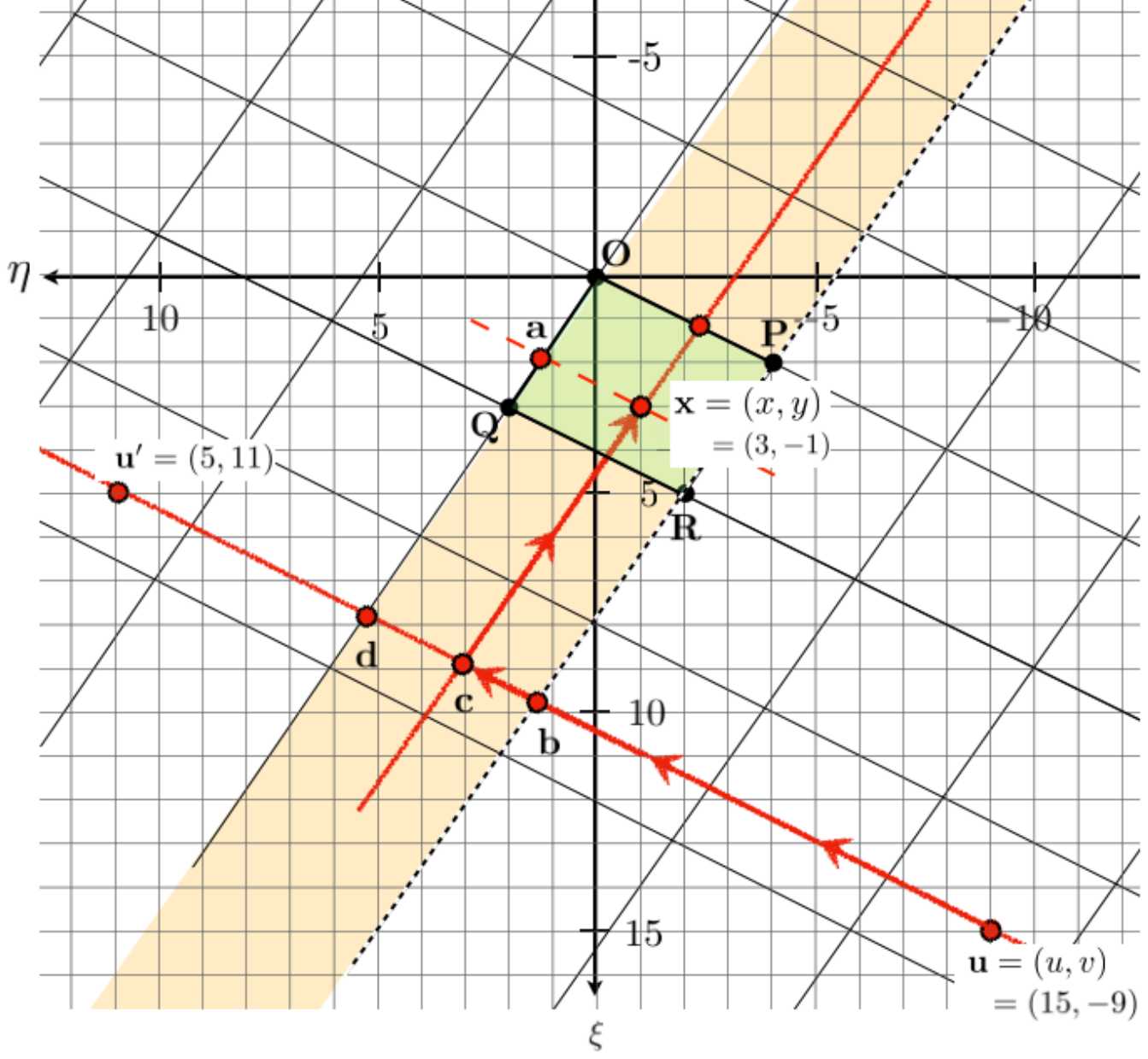}
        \caption{Cycloid  $ \mathcal{C}( 2,4,3,2 )$ in the Petri space to illustrate the proof of Theorem \ref{xy-to-FP}.}
        \label{u-v-to-FP}
      \end{center}
\end{figure}
%
If $\lambda$ is not  negative, then $u$ is on the half-line starting with (and including) the point $d$ in the direction of $c$  and $b$,
hence 
$\mathbf{d} = \vec{u} +\lambda \begin{pmatrix}-\alpha  \\ \beta   \end{pmatrix}$  and 
$\mathbf{b} = \mathbf{d} - \begin{pmatrix}-\alpha  \\ \beta   \end{pmatrix}= 
\vec{u} +(\lambda -1) \begin{pmatrix}-\alpha  \\ \beta   \end{pmatrix}$. Since $c$ is between $d$ and $b$ on the same line  we obtain for 
$\lambda_1$ satisfying $\mathbf{c} = \vec{u} +\lambda_1 \begin{pmatrix}-\alpha  \\ \beta   \end{pmatrix}$ the inequality 
$\lambda  \geq \lambda_1 > \lambda -1$ and since $\lambda_1 \in \Int$ the result 
$m = \lfloor  \lambda \rfloor $. 

If $\lambda$ is  negative, then   $u$, now denoted as $u'$, is on the complementary half-line until (but not containing)  $d$ and not containing  $b$.
The distance between $u'$ and $c$ now is  \textbf{greater} (or equal) than the distance between $u'$ and $d$. 
Instead of the inequality 
$\lambda  \geq \lambda_1 > \lambda -1$ in the case before, we now have $\lambda   \leq \lambda_1 < \lambda -1$ and the number $\lambda_1 \in \Int$ is again obtained by $m = \lambda_1 =  \lfloor  \lambda \rfloor $. 
The formula for  $n = \lfloor  \mu \rfloor$ is derived analogously: it remains to determine the vector 
$\overrightarrow{c\,x} := \vec{x} - \vec{c}$. To this end we introduce $\vec{a} := \vec{d} + \overrightarrow{c\,x}$.
As in the case where we concluded $m = \lambda_1 =  \lfloor \lambda \rfloor$ from $\lambda$
we compare 
$\begin{pmatrix} 0\\ 0  \end{pmatrix} = \vec{d} - \mu \begin{pmatrix} \gamma\\ \delta  \end{pmatrix} $ 
with
$\vec{a}= \vec{d} - \mu_1 \begin{pmatrix} \gamma\\ \delta  \end{pmatrix} $ 
to obtain $n = \mu_1 =  \lfloor  \mu \rfloor$ and 
$ \overrightarrow{c\,x} = \vec{a} - \vec{d} =  \vec{d} - \mu_1 \begin{pmatrix} \gamma\\ \delta  \end{pmatrix} - \vec{d} =
 - \mu_1 \begin{pmatrix} \gamma\\ \delta  \end{pmatrix}$. Combining this result with the vector $\vec{c}$ we obtain for the wanted point
\[
\vec{x} = \vec{c} +  \overrightarrow{c\,x}  = 
\vec{u} + \lfloor  \lambda \rfloor \begin{pmatrix}-\alpha  \\ \beta   \end{pmatrix} +\lfloor  \mu \rfloor \begin{pmatrix}-\gamma \\ -\delta   \end{pmatrix} =
\vec{u} + m \begin{pmatrix}-\alpha  \\ \beta   \end{pmatrix} +n \begin{pmatrix}-\gamma \\ -\delta   \end{pmatrix} =
\vec{u} - \mathbf{A}\begin{pmatrix} m \\ n  \end{pmatrix}
\]
the formula just before  equation (\ref{Vektorgleichung 1}).
\end{proof}

An example for the preceding proof is shown in Figure \ref{u-v-to-FP} for the cycloid $ \mathcal{C}( 2,4,3,2 ) $
for $\vec{u} = (u,v) = (15,-9)$ we obtain
$(\lambda,\mu) = (\frac{57}{16}, \frac{21}{8})$, $(m,n) = (3,2), \textbf{c} =(9,3), (x,y) = (3,-1)$.
For the point $(u',v') = (5,11)$ on the same line $\overline{(d,b)} $ we obtain
$(\lambda,\mu) = (- \frac{23}{16}, \frac{21}{8}), (m,n) = (-2,2), \textbf{c} = (9,3), (x,y) = (3,-1) $. 
Next we apply the theorem for  $ \mathcal{C}( 2,3,3,3 ) $ (Figure \ref{2333}) to the element $(9,-1)$,
belonging to the class in the first line of Table \ref{values} (the same as in the last line). 
 In this case we obtain: $(u,v) = (9,-1)$, 
 $m = \lfloor \frac{1}{A}(u\delta - v \gamma)\rfloor= \lfloor \frac{1}{15}(9\cdot3 - (-1)\cdot3)\rfloor = \lfloor \frac{30}{15}\rfloor=2$ 
  and $n = \lfloor \frac{1}{A}(v\cdot\alpha + u\cdot\beta)\rfloor = \lfloor \frac{1}{15}(-1\cdot2 + 9\cdot3)\rfloor= \lfloor \frac{25}{15}\rfloor=1$ and \\
 \[\rho\left( \begin{pmatrix} 9\\ -1 \end{pmatrix} \right) =  \begin{pmatrix} 9\\ -1 \end{pmatrix} - \begin{pmatrix} 2 & 3 \\ -3 & 3 \end{pmatrix}\cdot\begin{pmatrix} 2 \\ 1 \end{pmatrix} =
  \begin{pmatrix} 9\\ -1 \end{pmatrix} - \begin{pmatrix} 7 \\ -3 \end{pmatrix} =  \begin{pmatrix} 2 \\ 2 \end{pmatrix}.
  \]


 \section{Reduction of Cycloids}\label{sec-reduction}
 
  Reductions are defined  on the set  $\Sigma$ of all cycloids  using the terminology from
 \cite{Baa-Nip-1998}.
 For each parameter $\lambda \in  \{   \alpha, \beta, \gamma, \delta \} $ a reduction step $\xrightarrow[\text{\tiny{}}]{\lambda}$ represents a particular scaling down of this parameter.
 Following Theorem \ref{shear} we introduce four such reductions, also called  \emph{reduction rules $R_\lambda$}, for cycloids  keeping them isomorphic. 
 
  \begin{definition} \label{reduction rules} 
For cycloids $ \mathcal{C}_1( \alpha_1, \beta_1, \gamma_1, \delta_1 ) $ and $ \mathcal{C}_2( \alpha_2, \beta_2, \gamma_2, \delta_2 ) $ the following conditional reduction steps  (also called \emph{reduction rules $R_\lambda$}) are defined:
\begin{itemize}
\item [$R_{\alpha}$ :] $\mathcal{C}_1 \xrightarrow[\text{\tiny{}}]{\alpha} \mathcal{C}_2 :\Leftrightarrow$ 
       $ \;\; \alpha_2 = \alpha_1-\gamma_1, \beta_2 = \beta_1+\delta_1, \gamma_2 = \gamma_1$ and $\delta_2 = \delta_1 $
        if $\alpha_1 > \gamma_1$.
\item [$R_{\beta}$ :]  $\mathcal{C}_1 \xrightarrow[\text{\tiny{}}]{\beta} \mathcal{C}_2 :\Leftrightarrow$ 
         $\;\; \alpha_2 = \alpha_1+ \gamma_1, \beta_2 = \beta_1-\delta_1, \gamma_2 = \gamma_1$ and 
 	$\delta_2 = \delta_1$
	 if $\beta_1 > \delta_1$.
 \item [$R_{\gamma}$ :] $\mathcal{C}_1 \xrightarrow[\text{\tiny{}}]{\gamma} \mathcal{C}_2 :\Leftrightarrow$ 
         $ \;\; \alpha_2 = \alpha_1, \beta_2 = \beta_1, \gamma_2 = \gamma_1-\alpha_1$ and $\delta_2 = \delta_1 + 		\beta_1$
	if $\gamma_1 > \alpha_1$.
 \item [$R_{\delta}$ :]  $\mathcal{C}_1 \xrightarrow[\text{\tiny{}}]{\delta} \mathcal{C}_2 :\Leftrightarrow$ 
          $\;\; \alpha_2 = \alpha_1, \beta_2 = \beta_1, \gamma_2 = \gamma_1+\alpha_1$ and  $\delta_2 = \delta_1-\beta_1$ 
	if $\delta_1 > \beta_1$.
 \end{itemize}    
  \end{definition}
 

As several of such rules can be applied, we extend reductions more generally to sets $\Lambda
$ of rules.
 
 \begin{definition} \label{reduction} 
 Let be $\Sigma :=   \{  \mathcal{C}( \alpha, \beta, \gamma, \delta ) | \alpha, \beta, \gamma, \delta \in \Natp \}$
the set of all cycloids and $\Lambda \subseteq \{   \alpha, \beta, \gamma, \delta \}  $ a nonempty subset of parameters. 
A  \emph{$\Lambda$-reduction} is a binary relation $\Lambda \subseteq \Sigma \times \Sigma $ written 
$\mathcal{C} \xrightarrow[\text{\tiny{}}]{\Lambda} \mathcal{C}' $ for $(\mathcal{C},\mathcal{C}') \in \Lambda$
\begin{itemize}
\item [a)] 
$C \xrightarrow{\Lambda^*} C'$
       holds if 
       there is a sequence, called $\Lambda$-\emph{reduction chain},
       $\mathcal{C}_1   \xrightarrow[\text{\tiny{}}]{\lambda_1}  \mathcal{C}_2   \xrightarrow[\text{\tiny{}}]{\lambda_2}  	   	\cdots   \xrightarrow[\text{\tiny{}}]{\lambda_n} \mathcal{C}_{n+1},  $ $(n  \geq 1)$ of reduction steps with $\mathcal{C} 	= \mathcal{C}_1$, $\mathcal{C}_{n+1} = \mathcal{C}'$ and  $\lambda_i \in \Lambda$ for all $1  \leq i  \leq n$.
$n=1$ includes the \emph{reduction step}
       $\mathcal{C} \xrightarrow[\text{\tiny{}}]{\lambda} \mathcal{C}' $.
       In the cases $\Lambda = \{ \lambda \}$ and $\Lambda = \{ \lambda_1,\lambda_2 \}$ we write shorter  
       $\mathcal{C} \xrightarrow[\text{\tiny{}}]{\lambda} \mathcal{C}' $ and
       $\mathcal{C} \xrightarrow[\text{\tiny{}}]{\lambda_1\lambda_2} \mathcal{C}' $, respectively.
       In these two cases we also denote a $\Lambda$-\emph{reduction chain} by $\lambda$-\emph{reduction chain} and $\lambda_1\lambda_2$-\emph{reduction chain}, respectively.
\item [b)] $ \mathcal{C}$ is $\Lambda$-\emph{reducible} if there is  $ \mathcal{C}_1 \in \Sigma$ such that 
        $\mathcal{C} \xrightarrow[\text{\tiny{}}]{\Lambda} \mathcal{C}_1 $.
\item [c)] $ \mathcal{C}$ is $\Lambda$-\emph{irreducible}\footnote{In \cite{Baa-Nip-1998} also the term \emph{in normal form} is used.}
   	if it is not  $\Lambda$-reducible.
 \item [d)] $ \mathcal{C}'$ is a  $\Lambda$-\emph{reduction} of $ \mathcal{C}$ if 
 	 $\mathcal{C} \xrightarrow[\text{\tiny{}}]{\Lambda^*} \mathcal{C}' $ and $ \mathcal{C}'$ is $\Lambda$-\emph{irreducible}.
 \item [e)] Cycloids $ \mathcal{C}_1$ and $ \mathcal{C}_2$  are called $\alpha\gamma$-\emph{reduction 			equivalent}, denoted $\mathcal{C}_1\; \simeq_{\alpha\gamma} \;\mathcal{C}_2$, if they have the same $		\alpha\gamma$-reduction. 
	They are called $\beta\delta$-\emph{reduction equivalent}, 		denoted 
	$\mathcal{C}_1\; \simeq_{\beta\delta} \;\mathcal{C}_2$, if they have the same $\beta\delta$-reduction. 
\item[f)]
 	If $ \mathcal{C} $ is $\gamma$-irreducible and $\gamma < \alpha$ (resp. $\gamma = \alpha$) then $ 			\mathcal{C} $ is called weakly $\gamma$-irreducible (resp. strongly $\gamma$-irreducible).
 	If $ \mathcal{C} $ is $\delta$-irreducible and $\delta < \beta$ (resp. $\delta = \beta$) then $ \mathcal{C} $ is 		called weakly $\delta$-irreducible (resp. strongly $\delta$-irreducible).
\end{itemize}     
\end{definition}

The chain
$ \mathcal{C}( 1,13,1,16 ) \xrightarrow{\delta}
\mathcal{C}( 1,13,2,3 ) \xrightarrow{\beta}
\mathcal{C}( 3,10,2,3 ) \xrightarrow{\beta} 
\mathcal{C}( 5,7,2,3)\xrightarrow{\beta} 
\mathcal{C}( 7,4,2,3) \xrightarrow{\beta}
\mathcal{C}( 9,1,2,3) \xrightarrow{\delta}
\mathcal{C}( 9,1,11,2 ) \xrightarrow{\delta} 
\mathcal{C}( 9,1,20,1 )$
is a  $\beta\delta$-reduction  chain ending in a  $\beta\delta$-irreducible cycloid which is also strongly $\delta$-irreducible.
If we reverse the chain, we  obtain a $\alpha\gamma$-reduction  ending in the $\alpha\gamma$-irreducible cycloid $ \mathcal{C}( 1,13,1,16 )$.
 Also a motive for the naming of strongly/weakly-irreducible can be given by this example.
 $ \mathcal{C}( 1,13,2,3 ) $ is weakly $\delta$-irreducible. It can be further reduced, however,   in the $\beta\delta$-reduction until the strongly $\delta$-irreducible cycloid  $ \mathcal{C}( 9,1,20,1) $ is reached.

  The formula for the length of a minimal cycle in Theorem \ref{minimal cycles} c) is invariant with respect to these reductions. This allows to apply the formula in some cases, even if the conditions $\gamma  \leq \delta$ or $\gamma   \geq \delta$ are not satisfied.

  \begin{theorem} \label{invariant}
Let be $\lambda \in \{ \gamma,\delta \}$ and  $ \mathcal{ C}=\mathcal{C}( \alpha, \beta, \gamma, \delta )  \xrightarrow[\text{\tiny{}}]{\lambda}   \mathcal{C}_1 = \mathcal{C}( \alpha, \beta, \gamma_1, \delta_1 ) $
and
\[
cyc = \gamma + \delta + 
                     \; \left\{
	                     \begin{array}{lll}
		            \lfloor\frac{\delta}{\beta}\rfloor (\alpha - \beta) &  \textrm{if} &\alpha \leq \beta    
		            \\
		            -\lfloor\frac{\gamma}{\alpha}\rfloor (\alpha - \beta) & \textrm{if} & \alpha > \beta   
		            	                \end{array}
                     \right\} .  
  \]                  
If $cyc_1$ is the same formula, but applied to $ \mathcal{ C}_1$, then for $\alpha \leq \beta,\lambda = \delta$ 
or $\alpha>\beta, \lambda = \gamma$ we obtain
  $cyc_1 = cyc$.                  
\end{theorem} 

\begin{proof} 
For  $\lambda = \gamma$ and $\lambda = \delta$ we obtain 
$\gamma_1 = \gamma - \alpha, \delta_1 = \delta + \beta$ and $\gamma_1 = \gamma + \alpha, \delta_1 = \delta - \beta$, respectively. 
If $\alpha  \leq \beta$ and using $\lfloor x+k \rfloor = \lfloor  x\rfloor + k \;
(x \in \Real, \; k \in \Int$)
and applying
the formula  to $ \mathcal{ C}_1$ we obtain
\begin{align*}
cyc_1 &= \gamma_1 + \delta_1 + \; \lfloor\frac{\delta_1}{\beta}\rfloor (\alpha - \beta)  
 = \gamma  \mp \alpha + \delta \pm \beta + \; \lfloor\frac{\delta\pm \beta}{\beta}\rfloor (\alpha - \beta)  \\
& = \gamma   + \delta  \mp (\alpha - \beta) + \; (\lfloor\frac{\delta}{\beta}\rfloor \pm 1)(\alpha - \beta)  
 = \gamma   + \delta   + \; \lfloor\frac{\delta}{\beta}\rfloor (\alpha - \beta)  = cyc.
\end{align*}
In the same way for $\alpha > \beta$ we derive 
\begin{align*}
cyc_1 &= \gamma_1 + \delta_1 - \; \lfloor\frac{\gamma_1}{\alpha}\rfloor (\alpha - \beta)  
 = \gamma  \mp \alpha + \delta \pm \beta - \; \lfloor\frac{\gamma\mp \alpha}{\alpha}\rfloor (\alpha - \beta)  \\
 &= \gamma   + \delta  \mp (\alpha - \beta) - \; (\lfloor\frac{\gamma}{\alpha}\rfloor \mp 1)(\alpha - \beta)  
 = \gamma   + \delta   - \; \lfloor\frac{\gamma}{\alpha}\rfloor (\alpha - \beta)  = cyc.
\end{align*}
\end{proof}

To give an example consider the cycloid   $ \mathcal{C}( 4,2,17,1 ) $ where Theorem \ref{minimal cycles} c) can not be applied since $\alpha > \beta$ and $\gamma > \delta
$.
By repeated $\gamma$-reduction we obtain $ \mathcal{C}( 4,2,1,9 ) $ with
$\gamma < \delta$. Therefore we can compute $cyc = 17+1 - \lfloor \frac{17}{4}  \rfloor (4-2) = 18 - 8 = 10$. This is the same value as for the irreducible cycloid:
$cyc = 1+9 - \lfloor \frac{1}{4}  \rfloor (4-2) = 10 - 0 = 10$. 

In the following theorem the cases
$\lfloor \frac{\delta}{\beta}  \rfloor = 0$  and 
$\lfloor \frac{\gamma}{\alpha}  \rfloor = 0$ are used  to directly compute $\gamma$ and $\delta$ from 
  the properties $cyc$ and $A$ of the cycloid net.
  If the cycloid is $\gamma$- or $\delta$-irreducible Theorem \ref{synthesis} the takes a simpler form, since 
  the floor functions in Theorem \ref{minimal cycles} c) obtain zero as their values. This is the case if the reduction is weak (Definition \ref{reduction} f). If the reduction is strong we have $\alpha=\gamma$ or $\beta=\delta$.
  In the case of a strongly $\delta$-irreducible cycloid with $\beta=\delta$  the equation $A = \alpha\cdot\beta+\beta\cdot\gamma$ implies simply $\gamma = \frac{A}{\beta}-\alpha  $. The case of a  strongly $\gamma$-irreducible cycloid is similar.
That is why we treat only the  case of weakly irreducible cycloids  in the following theorem, which takes a similar form as Theorem \ref{synthesis}.

\begin{theorem} \label{Th-normal-a-neq-b}
Let be $ \mathcal{C} =\mathcal{C}( \alpha, \beta, \gamma, \delta) $ a lbc-cycloid with known values
$\alpha \neq \beta$, area $A$, minimal cycle length $cyc$. Furthermore we assume that the cycloid is weakly $\delta$-irreducible if $\alpha  \leq \beta$ and weakly $\gamma$-irreducible if $\alpha  >\beta$.
Then $\gamma =\frac{1}{\alpha - \beta}\cdot(\alpha \cdot cyc - A )$ and
$\delta= \frac{1}{\alpha - \beta}\cdot(A - \beta \cdot cyc )$.
 \end{theorem}

\begin{proof} 
Since we have $\lfloor \frac{\gamma}{\alpha} \rfloor = 0$ or $\lfloor \frac{\delta}{\beta} \rfloor = 0$ 
by Theorem \ref{minimal cycles} we obtain $cyc = \gamma+\delta$. 
 With the formula for $A$ we have the equation 
 $ \begin{pmatrix} cyc \\ A \end{pmatrix}= 
 \begin{pmatrix} 1 & 1 \\ \beta & \alpha \end{pmatrix}^{} \begin{pmatrix} \gamma \\ \delta \end{pmatrix}
 $ 
to compute the solution 
\[
 \begin{pmatrix} \gamma \\ \delta \end{pmatrix} = 
 \begin{pmatrix} 1 & 1 \\ \beta & \alpha \end{pmatrix}^{-1} \begin{pmatrix} cyc \\ A \end{pmatrix} = 
 \frac{1}{\alpha-\beta}\begin{pmatrix} \alpha & -1 \\ -\beta & 1 \end{pmatrix} \begin{pmatrix} cyc \\ A \end{pmatrix}=
 \frac{1}{\alpha-\beta}\begin{pmatrix} \alpha\cdot cyc-A \\ -\beta \cdot cyc + A \end{pmatrix}.
 \]
\end{proof}


\section{Cycloid Isomorphisms and Reduction Equivalence}\label{sec-morphism}

In this section the synthesis problem is solved for
cycloid nets without initial marking. In addition, also
the first two parameters $\alpha$ and $\beta$ are subjects to reductions.
A net isomorphism between two cycloids (Definition \ref{def-net}) does not necessarily preserve forward or backward places.
To preserve these properties we define the notion of a 
cycloid isomorphism.

\begin{definition} \label{def-cyc-iso} 
Given two cycloids 
$ \mathcal{C}_i =\mathcal{C}_i( \alpha_i, \beta_i, \gamma_i, \delta_i ), i\in\{ 1,2 \}$
and their cycloid nets  $\N{i} =
(\GSvw_i, \GSrw_i, T_i, F_i)$ (Definition \ref{cycloid}), 
 a mapping $f: X_1 \to X_2$ 
 ($X_i = S_i \cup T_i,
 S_i = S^{\rightarrow}_i \cup S^{\leftarrow}_i ) $ is a \emph{cycloid morphism} if
 \begin{align*}
f(F_1 \cap (S^{\rightarrow}_1 \cp T_1) ) &\subseteq (F_2 \cap (S^{\rightarrow}_2 \cp T_2)) \cup id, \\
f(F_1 \cap (S^{\leftarrow}_1 \cp T_1) ) &\subseteq (F_2 \cap (S^{\leftarrow}_2 \cp T_2)) \cup id, \\
f(F_1 \cap (T_1 \cp S^{\rightarrow}_1) ) &\subseteq (F_2 \cap (T_2 \cp S^{\rightarrow}_2)) \cup id, \text{ and}\\ 
f(F_1 \cap (T_1 \cp S^{\rightarrow}_1) ) &\subseteq (F_2 \cap (T_2 \cp S^{\rightarrow}_2)) \cup id. 
\end{align*}
 $f$ is a \emph{cycloid isomorphism} if 
 $f$ is a bijection and the inverse $f^{-1}$ is a also a cycloid morphism.
 Then $ \mathcal{C}_1$ and $ \mathcal{C}_2$ are called \emph{cycloid isomorphic},
 denoted $\mathcal{C}_1 \simeq_{\text{cyc}} \mathcal{C}_2$.
	If $ \mathcal{C}_1$ and $ \mathcal{C}_2$ are cycloid systems with initial markings $M_0^1$ and $M_0^2$, respectively, 
	then the definition of a cycloid isomorphism is extended by
		$f(S^{\rightarrow}_1 \cap M^1_0 ) = S^{\rightarrow}_2 \cap M^2_0 $ and 
		$f(S^{\leftarrow}_1 \cap M^1_0 ) = S^{\leftarrow}_2 \cap M^2_0 $.
\end{definition}

The next Lemma \ref{pi(post)} prepares the proof that the shear mappings (Theorem \ref{shear}) are cycloid isomorphisms.


\begin{lemma} \label{pi(post)}     
For each transition $t$ of the Petri space
$\rho(\postnbfw{t}) = \postnbfw{\rho(t)}$ and $\rho(\postnbbw{t}) = \postnbbw{\rho(t)}$ hold.
\end{lemma}
\begin{proof} 
By the equivalence relation $\equiv $ (Definition \ref{cycloid}) the Petri space is divided into tiles of the fundamental parallelogram's shape. A cycloid is defined by folding all these tiles onto the fundamental parallelogram with origin $t_{0,0}$\footnote{In Figure \ref{P-space+FD} a) neighbouring tiles of the fundamental parallelogram are represented by broken lines.}.
By Equation (\ref{pi}) from Theorem \ref{xy-to-FP} we obtain 
\begin{equation}\label{10+}
 \rho(\vec{u}) = 
\vec{u} - \begin{pmatrix} \alpha & \gamma \\ -\beta & \delta \end{pmatrix}\begin{pmatrix} m \\ n \end{pmatrix} = 
\vec{u}-m \cdot\begin{pmatrix} \alpha \\ -\beta  \end{pmatrix} - n \cdot \begin{pmatrix} \gamma\\ \delta \end{pmatrix}. 
\end{equation}

First, we calculate the values of m and n when a transition $t_{u,v}$ of the Petri space with coordinates $\vec{u} =(u,v)$ lies in the neighbouring tiles of the fundamental parallelogram. Starting with the neighbour on the right, these are designated clockwise with the Roman numerals \textbf{I}  to \textbf{VIII} (see Figure \ref{pspace} b). By Theorem \ref{xy-to-FP} for the corner $P = \vec{u}=(u,v) = (\alpha,-\beta)$, which is located outside the fundamental parallelogram,  we obtain
\[
m = \lfloor \frac{1}{A}(u\cdot\delta - v\cdot \gamma)\rfloor =\lfloor \frac{1}{A}(\alpha\cdot\delta +\beta\cdot \gamma)\rfloor = \lfloor \frac{A}{A} \rfloor =1 
\]
and 
\[
n = \lfloor \frac{1}{A}(v\cdot\alpha + u\cdot\beta)\rfloor
= \lfloor \frac{1}{A}(-\beta\cdot\alpha + \alpha\cdot\beta)\rfloor = 0,
\] 
hence $ \rho(\vec{u}) = (0,0)$.
By the symmetry of the fundamental parallelograms $(m,n) = (1,0)$ are the parameters for all transitions in neighbour tile \textbf{I}. Table \ref{tile} shows in its second line with respect to the neighbour   \textbf{I} the values $\Delta ( \alpha, \beta, \gamma, \delta ) = (-\alpha,\beta)$ and $(m,n) = (1,0)$. The last column gives an example for the cycloid  $ \mathcal{C}( 2,3,3,3 ) $ from Figure \ref{pspace} b) for the transition $(4,-4)$  by $\rho(4,-4) = (4,-4) - \Delta ( 2,3,3,3 ) = (4,-4) -  ( 2,-3 ) = (2,-1)$. With respect to the neighbouring tile \textbf{II} we repeat the calculation with respect to the corner  $R =(\alpha+\gamma,\delta,-\beta)$  instead of $P$ and obtain $(m,n) = (1,1)$. Similar computations for all neighbours are given in Table \ref{tile}, whereby the results for \textbf{VI}, \textbf{VII} and \textbf{VIII}  are not used in the following, but are given for the sake of completeness.
Note that for the second and third column the entries of the lines \textbf{V} to \textbf{VIII} are the negative   of the entries of \textbf{I} to \textbf{IV}.

Instead of 
the  equation $\rho(\postnbfw{t}) = \postnbfw{\rho(t)}$ of the Lemma we prove 
$\rho(t_{u+1,v}) = \postnbfw{\rho(t_{u,v}) }$, since $\postnbfw{t}=t_{u+1,v}$ is in the Petri space and $\postnbfw{\rho(t)}$  lies in the cycloid net. We distinguish the two cases, where $t_{u,v}$ and $t_{u+1,v}$
are located in the same tile or not.

\textbf{Case 1:} 
The transitions $t_{u,v}$ and $t_{u+1,v} $ are located both within the same tile.
Then they relate to the same values of $m=0$ and $n=0$ in Equation (\ref{10+}). By representing transitions like
$t_{a,b}$ by their coordinates $(a,b)$ we obtain from Equation (\ref{10+}): 
\[
\rho\begin{pmatrix} u+1 \\ v \end{pmatrix} = 
\begin{pmatrix}u +1 \\ v \end{pmatrix} - \mathbf{A}\begin{pmatrix} 0 \\ 0 \end{pmatrix} = 
\begin{pmatrix} u \\ v \end{pmatrix} + \begin{pmatrix} 1 \\ 0 \end{pmatrix} 
= 
\rho\begin{pmatrix} u\\ v \end{pmatrix}  + \begin{pmatrix} 1 \\ 0 \end{pmatrix}
= \postnbfw{\rho\begin{pmatrix} u\\ v \end{pmatrix} }.
\]
The analogous derivation  holds with respect to the backward output transition $\rho\begin{pmatrix} u \\ v +1\end{pmatrix}$.\\

\textbf{Case 2:}  $t_{u+1,v} $ is located  within the tile to the right of that of $t_{u,v}$.
Now we assume that:
\[\rho\begin{pmatrix} u \\ v \end{pmatrix} = \begin{pmatrix}u  \\ v \end{pmatrix} - \mathbf{A}\begin{pmatrix} m \\ n \end{pmatrix}.
\]
Since by Table \ref{tile} the  change of parameters for the tile to the right hand side is $(m,n) = (1,0)$ we obtain
\begin{align*}
 \rho\begin{pmatrix} u+1 \\ v \end{pmatrix} &= 
\begin{pmatrix}u +1 \\ v \end{pmatrix} - \mathbf{A}\begin{pmatrix} m+1 \\ n \end{pmatrix} = 
\begin{pmatrix} u \\ v \end{pmatrix} + \begin{pmatrix} 1 \\ 0 \end{pmatrix} -(m+1) \cdot \begin{pmatrix} \alpha \\ -\beta  \end{pmatrix} - n \begin{pmatrix} \gamma\\ \delta \end{pmatrix} \\
&= \begin{pmatrix} u \\ v \end{pmatrix}  -m \cdot \begin{pmatrix} \alpha \\ -\beta  \end{pmatrix} - n \begin{pmatrix} \gamma\\ \delta \end{pmatrix}  + \begin{pmatrix} 1 \\ 0 \end{pmatrix} -  \begin{pmatrix} \alpha \\ -\beta \end{pmatrix} = \\
%
&=  \begin{pmatrix}u  \\ v \end{pmatrix} - \mathbf{A}\begin{pmatrix} m \\ n \end{pmatrix} + \begin{pmatrix} 1 \\ 0 \end{pmatrix} -  \begin{pmatrix} \alpha \\ -\beta \end{pmatrix}=
\rho\begin{pmatrix} u\\ v \end{pmatrix}  + \begin{pmatrix} 1 \\ 0 \end{pmatrix} -  \begin{pmatrix} \alpha \\ -\beta \end{pmatrix}= 
 \postnbfw{\rho\begin{pmatrix} u\\ v \end{pmatrix} }.
 \end{align*} 
 The penultimate equation applies due to the assumption. The last equation describes the step from the fundamental parallelogram to the right tile and back into the fundamental parallelogram.
The analogous derivation  holds with respect to the backward output transition $\rho\begin{pmatrix} u \\ v +1\end{pmatrix}
$.
This completes the proof of
$\rho(t_{u+1,v}) = \postnbfw{\rho(t_{u,v}) } $.
The analogous derivation  holds with respect to the backward output transition $\rho\begin{pmatrix} u \\ v +1\end{pmatrix}$.\\
The cases with respect to tiles  \textbf{II} to \textbf{V} are similar, whereas those with respect to \textbf{VI} to \textbf{VIII} do not apply, since these  tiles cannot be reached from the fundamental parallelogram.
\end{proof}


\begin{table}[htbb]
 \begin{center}
\begin{tabular}{|c||c|c|c|l}
\hline
 tile&$\Delta ( \alpha, \beta, \gamma, \delta ) $  &   $(m,n)$  & $(2,-1)+\Delta(2,3,3,3)$  \\  \hline \hline 
 \textbf{I}&$(\alpha,-\beta) $ & $(1,0)$ &$(4,-4)$ \\
   \hline
  \textbf{II}&$(\alpha+\gamma,-\beta+\delta) $ & $(1,1)$ &$(7,-1)$ \\
   \hline
  \textbf{III}&$(\gamma,\delta) $ & $(0,1)$ &$(5,2)$ \\
   \hline
 \textbf{IV}&$(-\alpha+\gamma,\beta+\delta) $ & $(-1,1)$ &$(3,5)$ \\
   \hline
 \textbf{V}&$(-\alpha,\beta) $ & $(-1,0)$ &$(0,2)$ \\
   \hline
 \textbf{VI}&$(-\alpha-\gamma,\beta-\delta) $ & $(-1,-1)$ &$(-3,-1)$ \\
   \hline
 \textbf{VII}&$(-\gamma,-\delta) $ & $(0,-1)$ &$(-1,-4)$ \\
   \hline
 \textbf{VIII}&$(\alpha-\gamma,-\beta-\delta) $ & $(1,-1)$ &$(1,-7)$ \\
   \hline
   \hline
\end{tabular}
 \end{center}
 \caption{Values of $(\xi+m\cdot 2 + n\cdot 3,\eta-m\cdot 3+n\cdot 3$)}
\label{tile}
       \end{table}
       

\begin{figure}[t]
	\begin{center}
		\includegraphics [scale = 0.4]{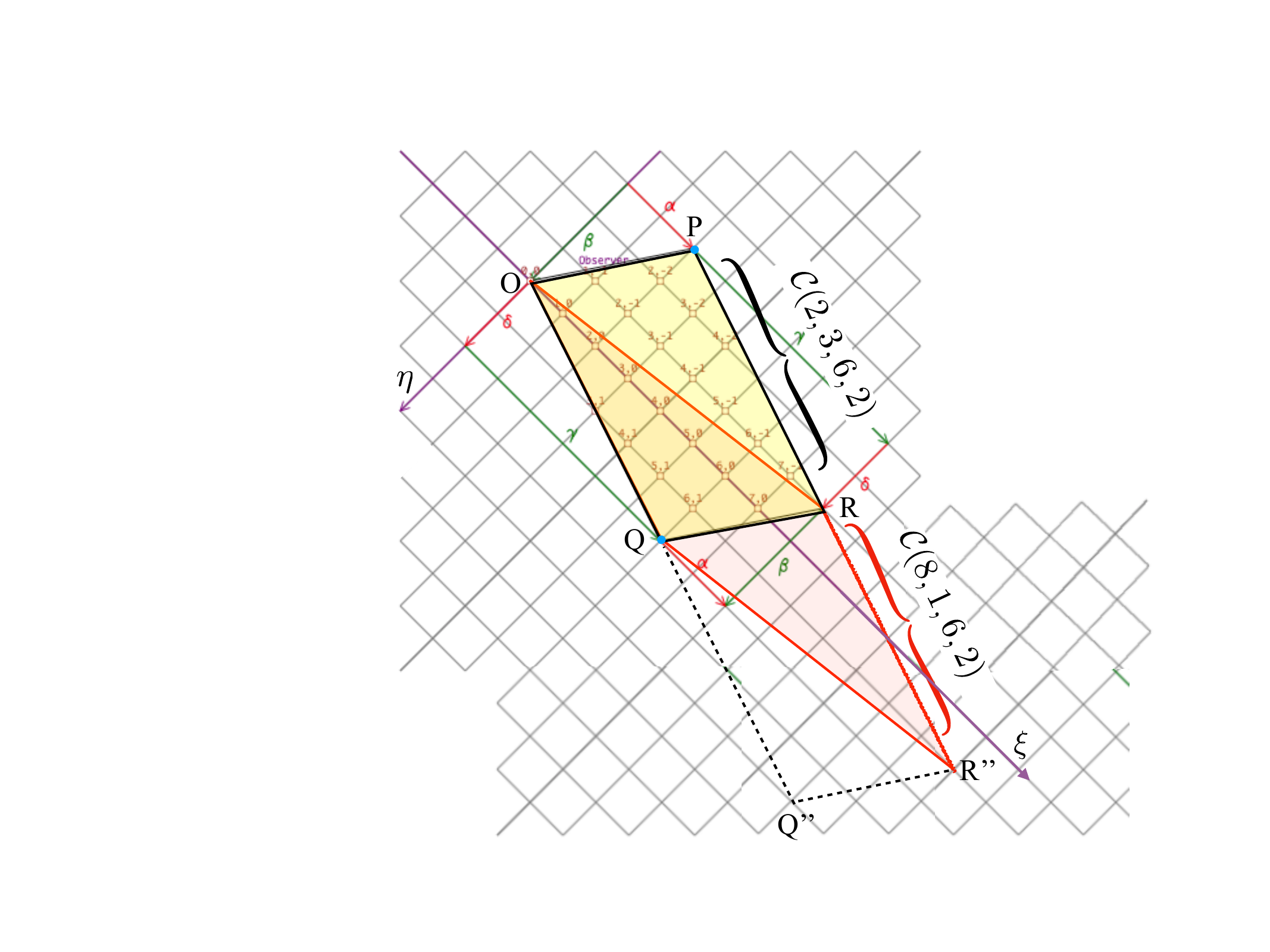}
		\caption{Shearing from $ \mathcal{C}( 2,3,6,2) $ to $ \mathcal{C}( 8,1,6,2) $. }
		\label{OQ-shearing}
	\end{center}
\end{figure}

\begin{theorem} \label{shearing}
If  $ \mathcal{C}'$ is the result of applying one of the four reduction rules from Definition \ref{reduction rules} to  $ \mathcal{C}$ then $ \mathcal{C}'$ and $ \mathcal{C}$ are cycloid isomorphic (Definition \ref{def-cyc-iso}).
\end{theorem}

\begin{proof} 
	In Theorem \ref{shear} the property of isomorphism has been proven.
	To prove the property of cycloid isomorphism for the cases $R_\alpha$,  $R_\beta$,  $R_\gamma$ and  $R_\delta$  in Definition \ref{reduction rules} we define cycloid isomorphisms $\psi_\alpha$, $\psi_\beta$, $\psi_\gamma$ and $\psi_\delta$, respectively.
	We start with case $R_\delta$  and introduce on the Petri space the linear isomorphism $\varphi_\delta(\xi,\eta) = (\xi+\alpha,\eta-\beta) $. It maps the corners 
	$O =(0,0)$, $Q =(\gamma,\delta)$, $P =(\alpha,-\beta)$ and $R =(\alpha+\gamma,\delta,-\beta)$ 
	of the fundamental parallelogram of $\mathcal{C}(\alpha,\beta,\gamma,\delta) $ to
	$O'' =(\alpha,-\beta)$, $Q''=(\gamma+\alpha,\delta-\beta)$, $P'' =(2\cdot\alpha,-2\cdot\beta)$ and $R''=(2\cdot\alpha+\gamma,\delta-2\cdot\beta)$, respectively. 
For the cycloid  $ \mathcal{C}( 2,3,2,8) $	in Figure \ref{OP-shearing} the images of the vertices are   $O''= P$, $Q''= R$, $P''$ and $R''=R'$ and broken lines for additional edges. 
In general, $\varphi_\delta(0,0)=(\alpha,-\beta)$ implies that $\varphi_\delta(\xi,\eta)$ is equivalent to $(\xi,\eta)$ in this case, and consequently for all elements $(\xi,\eta)$  of  $ \mathcal{C} $'s fundamental parallelogram: $(\xi,\eta) \equiv \varphi_\delta(\xi,\eta)$.
However, in the definition of a fundamental parallelogram the origin $O$ has to be the origin $t_{0,0}$ of the Petri space. 
Therefore we have to apply the map $\rho_{ \mathcal{C}'}$, which was proven to preserve forward and backward places and transitions in Lemma \ref{pi(post)} and is therefore a cycloid isomorphism. Finally for this case, we obtain the cycloid isomorphism 
$\psi_\delta(\xi,\eta) = \rho_{ \mathcal{C}'}(\varphi_\delta(\xi,\eta))$ or $\psi_\delta= \varphi_\delta\circ\rho_{ \mathcal{C}'}$.

For case $R_\gamma$ of the theorem the proof is similar by using	$\varphi_\gamma(\xi,\eta) = (\xi-\alpha,\eta+\beta) $ and 
$\psi_\gamma= \varphi_\gamma\circ\rho_{ \mathcal{C}'}$. The equivalence preserving vector is $(-\alpha,\beta)$ in place of $(\alpha,-\beta)$.

Also case $R_\alpha$ is proven in a similar way, now by using $\varphi_\alpha(\xi,\eta) = (\xi-\gamma,\eta-\delta) $
and 
$\psi_\alpha= \varphi_\alpha\circ\rho_{ \mathcal{C}'}$. The equivalence preserving vector is $(\gamma,\delta)$.

Case $R_\beta$ is dual to case $R_\alpha$, such as case $R_\gamma$ was dual to case $R_\delta$: $\varphi_\beta(\xi,\eta) = (\xi+\gamma,\eta+\delta) $
and $\psi_\beta= \varphi_\beta\circ\rho_{ \mathcal{C}'}$.
\end{proof} 

The image of the fundamental parallelogram of  $ \mathcal{C}( 2,3,6,2 ) $ with respect to $\varphi_\beta$ is given in Figure~\ref{OQ-shearing} by the corners $Q$, $Q''$, $R$ and $R''$ and broken lines for additional edges.
While the invariant edge of the fundamental parallelogram is the edge between $O$ and $Q$ it may be called
$O$-$Q$-shearing to distinguish it from the shearing of Figure~\ref{OP-shearing}, which is a $O$-$P$-shearing by the use of this terminology.
To give an example for a transformation which does not preserve isomorphism, consider the cycloids $ \mathcal{C}( \alpha, \beta, \gamma, \delta ) $ and $ \mathcal{C}( \gamma, \delta,\alpha, \beta) $.
Obviously they have the same area, but are not isomorphic in general.
For instance the cycloids $ \mathcal{C}( 3,1,1,1 ) $ and $ \mathcal{C}( 1,1,3,1) $ have the same area $A = 4$, but different minimal cycle length $2$ and $4$, respectively.

\begin{figure}[t]
 \begin{center}
        \includegraphics [scale = 0.45]{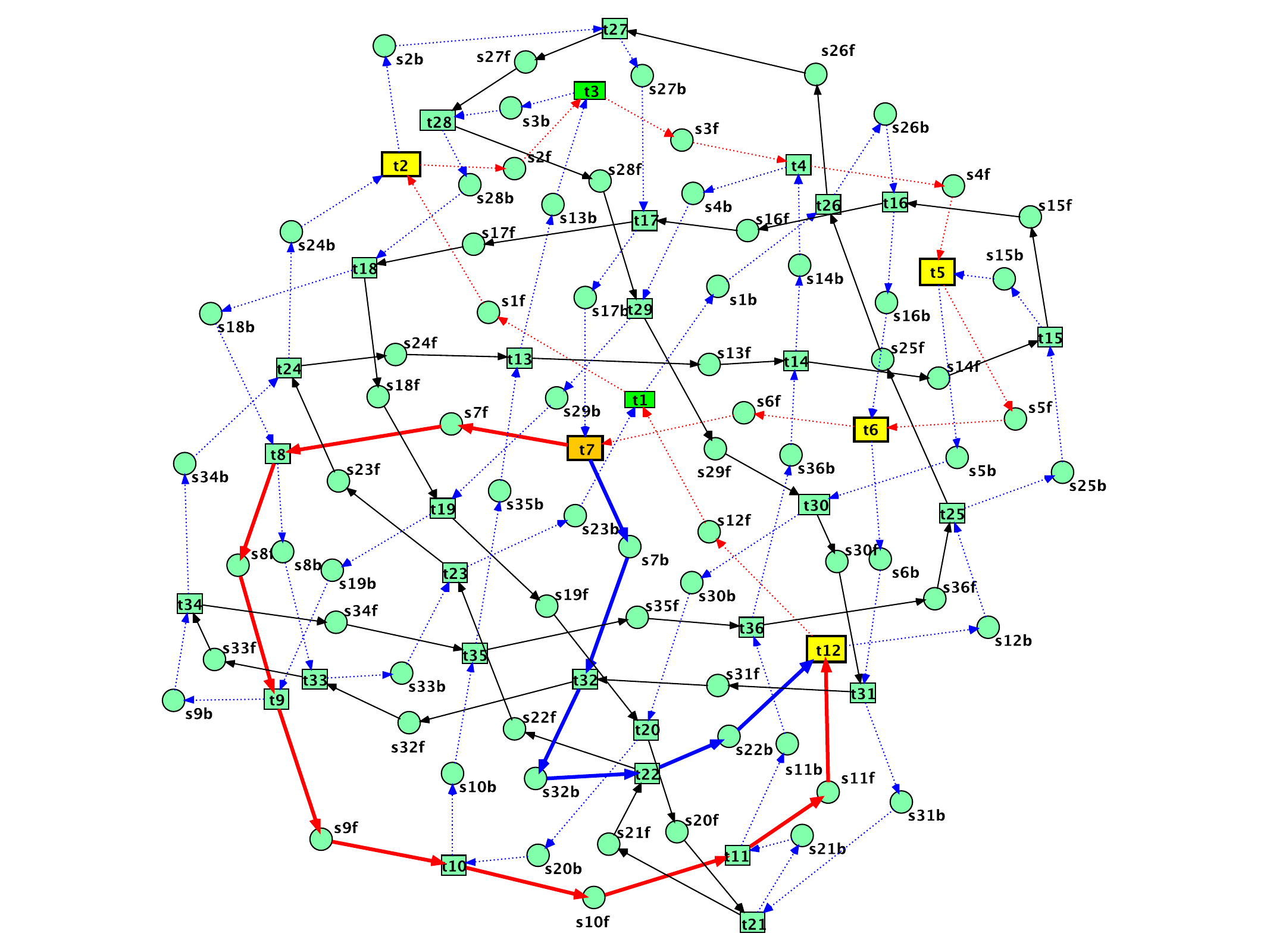}
        \caption{A cycloid net of the cycloid  $ \mathcal{C}( 5,3,2,6 ) $.  }
        \label{c-5-3-2-6-red2}
      \end{center}
\end{figure}

To prepare an algorithm for computing the parameters $ \alpha, \beta, \gamma, \delta $ of a cycloid net as in Figure~\ref{c-5-3-2-6-red2}, we give a formula for the interval of the $\xi$-axis belonging to the fundamental parallelogram. 

\begin{lemma} \label{xi-max}
For a cycloid $ \mathcal{C}( \alpha, \beta, \gamma, \delta )$ the interval of the $\xi$-axis 
within the fundamental parallelogram extends from the origin $t_{0,0}$ up to $t_{\xi_{max} ,0}$ where 
$\xi_{max} = \lceil \frac{A}{max(\beta,\delta)} \rceil - 1$.
\end{lemma} 

\begin{proof} 
The condition for a transition $t_{\xi,0}$ to lie on the $\xi$-axis within the fundamental parallelogram is by Theorem \ref{xy-to-FP}:
\[
\begin{pmatrix} \xi \\ 0 \end{pmatrix} = \begin{pmatrix} \xi \\ 0 \end{pmatrix} - \mathbf{A}\begin{pmatrix} m \\ n \end{pmatrix} \text{ or } \mathbf{A}\begin{pmatrix} m \\ n \end{pmatrix} = \begin{pmatrix} 0 \\ 0 \end{pmatrix} 
\]
with 
\begin{align*}
m &= \lfloor \frac{1}{A}(u\delta - v \gamma)\rfloor =\lfloor \frac{1}{A}(\xi\cdot\delta - 0\cdot \gamma)\rfloor =
\lfloor \frac{\xi\cdot\delta}{A} \rfloor \qquad \text{ and} \\ 
n &= \lfloor \frac{1}{A}(v\alpha + u\beta)\rfloor =\lfloor \frac{1}{A}(0\cdot\alpha + \xi\cdot\beta)\rfloor =
\lfloor \frac{\xi\cdot\beta}{A} \rfloor.
\end{align*}
This gives 
\[
\mathbf{A}\begin{pmatrix} m \\ n \end{pmatrix} = 
\begin{pmatrix} \alpha\cdot \lfloor \frac{\xi\cdot\delta}{A}\rfloor + \gamma \cdot \lfloor \frac{\xi\cdot\beta}{A} \rfloor \\ -\beta\cdot \lfloor \frac{\xi\cdot\delta}{A}\rfloor +\delta \cdot \lfloor \frac{\xi\cdot\beta}{A} \rfloor \end{pmatrix}= 
\begin{pmatrix} 0 \\ 0 \end{pmatrix}.
\]
Since the cycloid parameters are positive, 
from the first row we obtain 
$ \lfloor \frac{\xi\cdot\delta}{A}\rfloor = 0$ and $\lfloor \frac{\xi\cdot\beta}{A} \rfloor =0$ which satisfies also the second row.
The overall condition is therefore $\xi < \frac{A}{\delta} \;\land \;\xi < \frac{A}{\beta} $ or
$\xi < \frac{A}{max(\beta,\delta)}$. 
The largest integer satisfying this condition is $\xi_{max} = \lceil \frac{A}{max(\beta,\delta)} -1 \rceil
= \lceil \frac{A}{max(\beta,\delta)} \rceil - 1$.
\end{proof} 

A more geometric way to obtain this result starts with the observation that $\xi_{max} $ is the largest integer value on the $\xi$-axis before the intersection of the $\xi$-axis with the lines containing $Q,R$ or $P,R$ of the fundamental parallelogram. 
The line containing $Q$ and $R$ is given by the equation $\eta = -\frac{\beta}{\alpha}(\xi-\gamma) + \delta $ (see \cite{Valk-2019}).
Setting $\eta=0$ gives $\xi = \frac{A}{\beta} $ .
The line containing $P$ and $R$ is given by the equation $\eta = \frac{\delta}{\gamma}(\xi-\alpha) - \beta $.
Again, setting $\eta=0$ gives $\xi = \frac{A}{\delta} $. 
Therefore we obtain the overall condition
$\xi < \frac{A}{\delta} \;\land \;\xi < \frac{A}{\beta} $ and proceed as in the proof before. 
For the cycloid $ \mathcal{C}( 4,2,2,3) $ of Figure~\ref{P-space+FD}~a) we obtain $A=16$ and 
 $\xi_{max} = \lceil \frac{16}{max(2,3)} \rceil -1 = \lceil \frac{16}{3}  \rceil -1 = 5$. 
As can be seen in the figure, the $\xi$-axis overlaps with the fundamental parallelogram in the transitions from $t_{0,0}$ to $t_{5,0}$.
The values of $\xi_{max}$ for the cycloids of Figure~\ref{OQ-shearing} are $7$ and $10$.

\begin{lemma} \label{t00-back}
For a cycloid $ \mathcal{C}( \alpha, \beta, \gamma, \delta ) $ the backward output transition of $t_{0,0} $ is
$\postnbbw{t_{0,0}} = t_{\alpha,1-\beta}$ and the backward input transition of  $t_{0,0} $ is
$ \prenbbw{t_{0,0}} = t_{\gamma,\delta-1}$.
\end{lemma}

\begin{proof} 
a) For any cycloid the backward output transition $t_{0,1} = \postnbbw{t_{0,0}} $  of $t_{0,0} $ is not contained in the fundamental parallelogram.
Again, we calculate the equivalent $\vec{x}$ of $t_{0,1} $ within the fundamental parallelogram using Theorem \ref{xy-to-FP}: $\vec{x} = \vec{u} - \mathbf{A}\begin{pmatrix} m \\ n \end{pmatrix}$ where 
$\vec{u} = (u,v) = (0,1)$ and 
\begin{align*}
m &= \lfloor \frac{1}{A}(u\delta - v \gamma)\rfloor =
\lfloor \frac{1}{A}(0\cdot\delta - 1\cdot \gamma)\rfloor =
\lfloor \frac{-\gamma}{A}\rfloor = -1 \text{ and} \\ 
n &= \lfloor \frac{1}{A}(v\alpha + u\beta)\rfloor =
\lfloor \frac{1}{A}(1\cdot\alpha + 0\cdot\beta)\rfloor=
\lfloor \frac{\alpha}{A}\rfloor = 0.
\end{align*}
Hence we obtain 
\[
\vec{x} =\begin{pmatrix} 0 \\ 1 \end{pmatrix} - \begin{pmatrix} \alpha & \gamma \\ -\beta & \delta \end{pmatrix}\begin{pmatrix} -1 \\ 0 \end{pmatrix} = 
\begin{pmatrix} 0 \\ 1 \end{pmatrix} - \begin{pmatrix} -\alpha \\ \beta \end{pmatrix} =
\begin{pmatrix} \alpha \\ 1-\beta \end{pmatrix}.
\]
b) In the same way we calculate the equivalent $\vec{x}$ of $t_{0,-1} $ within the fundamental parallelogram using Theorem \ref{xy-to-FP}: $\vec{x} = \vec{u} - \mathbf{A}\begin{pmatrix} m \\ n \end{pmatrix}$ where 
$\vec{u} = (u,v) = (0,-1)$ and 
\begin{align*}
m &= \lfloor \frac{1}{A}(u\delta - v \gamma)\rfloor =
\lfloor \frac{1}{A}(0\cdot\delta + 1\cdot \gamma)\rfloor =
\lfloor \frac{\gamma}{A}\rfloor = 0 \text{ and} \\ 
n &= \lfloor \frac{1}{A}(v\alpha + u\beta)\rfloor =
\lfloor \frac{1}{A}(-1\cdot\alpha + 0\cdot\beta)\rfloor=
\lfloor \frac{-\alpha}{A}\rfloor = -1.
\end{align*}
Hence we obtain 
\[
\vec{x} =\begin{pmatrix} 0 \\ -1 \end{pmatrix} - \begin{pmatrix} \alpha & \gamma \\ -\beta & \delta \end{pmatrix}\begin{pmatrix} 0 \\ -1 \end{pmatrix} = 
\begin{pmatrix} 0 \\ -1 \end{pmatrix} - \begin{pmatrix} -\gamma \\ -\delta \end{pmatrix} =
\begin{pmatrix} \gamma \\ \delta-1 \end{pmatrix}.
\]
\end{proof}

This result is also obtained in a more geometric way as follows.
The position of the output transition of $\postnbbw{t_{0,0}} $ in the fundamental parallelogram is one step from $P$ in direction of the $\eta$-axis: 
$P + (0,1) = (\alpha,-\beta) + (0,1) = (\alpha,1-\beta)$ and, similarly for the second case: $Q + (0,-1) = (\gamma,\delta) + (0,-1) = (\gamma,\delta-1)$.

Reductions of Definition \ref{reduction} are obviously unique if only one rule is allowed, i.e. $\Lambda
 = \{ \lambda \}$. This is not true if two rules can be applied: i.e. $\Lambda = \{ \lambda_1,\lambda_2 \}$. 
To give an example, the sequence 
$  \mathcal{C}( 1,1,2,1 )  \xrightarrow[\text{}]{\gamma}   \mathcal{C}( 1,1,1,2 ) \xrightarrow[\text{}]{\delta}   \mathcal{C}( 1,1,2,1 )  \xrightarrow[\text{}]{\gamma}\cdots $ starts a non terminating sequence of 
$\gamma\delta$-reduction steps. The following theorem shows a different situation in the case of 
$\beta\delta$-reduction steps.

\begin{theorem} \label{bd-irreducible}
Each cycloid    $\mathcal{C}'=\mathcal{C}( \alpha', \beta', \gamma', \delta' ) $ has  an unique  $\beta\delta$-reduction $ \mathcal{C}=\mathcal{C}( \alpha, \beta, \gamma, \delta ) $ with the following properties:
\begin{itemize}
       \item [a)]  $\beta = \delta$ with $\beta = gcd(\beta',\delta')$.
       \item [b)]  $ \mathcal{C}$ is regular, composed of $\beta$ processes of length 
       		$p = \frac{A}{\beta} $.
       \item [c)] The cycloid corner R of $ \mathcal{C}$ is on the $\xi$-axis. Therefore the transitions of the forward 		cycle  through the origin are all lying on the $\xi-$axis.  
       \end{itemize}     
\end{theorem} 

\begin{proof} 
In a sequence of  $\beta\delta$-reduction steps at most one of the rules $R_\beta$ and $R_\delta$ is applicable since their conditions are contradictory. In each step the sum $\beta+\delta$  decreases strictly. Therefore each such sequence ends after a finite number of steps with a unique result.
\begin{itemize}
\item[a)] The pseudo-code in Table \ref{code} to compute  $ \mathcal{C}$ from $ \mathcal{C}'$ implements the repeated applications of rules $R_\beta$ and $R_\delta$ (Definition \ref{reduction rules}) by Dijkstra's guarded commands formalism (see \cite[p.~343]{Gries81}). Essentially this is the classical algorithm of Euclid.
\begin{table}[thbp] \small
    \begin{center}
          \begin{tabular}{rllccc}
   \textbf{Input}&  $ \alpha', \beta', \gamma', \delta'$    \;\;\;\;\;  \textbf{positive integers}  &          &   \\  

           &  $ \alpha,\beta,\gamma,\delta := \alpha',\beta',\gamma',\delta'$        &          &   \\  
            \textnormal{\textbf{do}}    &  $\{ gcd(\beta,\delta) = gcd(\beta',\delta') \}$        &           &   \\  
              &  $\beta > \delta  \; \rightarrow  \;\beta:= \beta - \delta \; ; \; \alpha := \alpha + \gamma $        & (Rule $R_\beta$)         &   \\ 
           $\square$  & $\beta < \delta  \; \rightarrow  \;\delta:= \delta - \beta  \; ; \; \gamma:= \gamma + \alpha$       & (Rule $R_\delta$)         &   \\  
             \textnormal{\textbf{od}}    &  $\{ \beta = \delta = gcd(\beta',\delta' )\}$       &          &  \\ 
              \textbf{Output}  & $ \alpha, \beta, \gamma, \delta$       &          &  \\ 
         \end{tabular}
   \end{center}
   \vspace*{-3ex}
       \caption{ Pseudo-code to compute for a given  $ \mathcal{C}( \alpha', \beta', \gamma', \delta' ) $ a $\beta\delta$-irreducible cycloid  $ \mathcal{C}( \alpha, \beta, \gamma, \delta ) $.}
    \label{code}
\end{table}
The loop invariant  $\{gcd(\beta,\delta) = gcd(\beta',\delta')\} $ holds initially since $\beta' = \beta$ and $ \delta' = \delta $. By the well known properties $gcd(\beta,\delta) = gcd(\beta-\delta,\delta)$ if $\beta>\delta$ and $gcd(\beta,\delta) = gcd(\delta-\beta,\beta)$ if $\beta<\delta$ the loop invariant property is true. The loop terminates by the strictly decreasing loop variant $\beta+\delta$ with $gcd(\beta,\delta) = gcd(\beta',\delta') = \beta = \delta$.
\item[b)] $ \mathcal{C}$ is regular (Definition \ref{regular}) since $\beta = \delta$
 implies  $\beta | \delta$ and by Theorem \ref{th-f-b-cycle} the process length is
       $p = \frac{A}{gcd(\beta,\delta)} =\frac{A}{\beta} $.
 \item[c)]  Since 
       $R = \begin{pmatrix} \alpha+\gamma \\ \delta - \beta  \end{pmatrix}$ in general, we obtain 
       $R =\begin{pmatrix} \alpha+\gamma \\ 0  \end{pmatrix}$ for $ \mathcal{C}$. Hence, R is on the $\xi$-axis. 		The transitions of the forward cycle  through the origin are all lying on the $\xi-$axis if and only if  
       $\xi_{max} $ (Lemma \ref{xi-max}) is one less than the $\xi$-coordinate of $R$. Therefore we prove
       $\xi_{max} =  \alpha+\gamma-1 $ as follows:   
       $\xi_{max} = 
       \lceil \frac{A}{max(\beta,\delta)} \rceil - 1 = 
       \lceil \frac{\alpha\cdot\beta+\beta\cdot\gamma}{max(\beta,\beta)} \rceil - 1 = 
       \lceil \frac{\beta\cdot(\alpha+\gamma)}{\beta} \rceil - 1 = 
        \alpha+\gamma-1 $.
        \end{itemize}
\end{proof}

Successive applications of the same rule in classical versions of Euclid's algorithm 
can be executed in a single operation using the modulo function \cite{sipser}.
Since $a \, mod \,b < \frac{a}{2}$ for any $0 < b < a$, a time complexity of 
$T(n) = \mathcal{O}(log_2 \, n)$ is obtained for computing $gcd(a,b)$ where $n = max(a,b)$. The form of Table \ref{code} is preferred here since it is using the rules $R_\beta$ and $R_\delta$ .

\begin{figure}[htbp]
 \begin{center}
        \includegraphics [scale = 0.32]{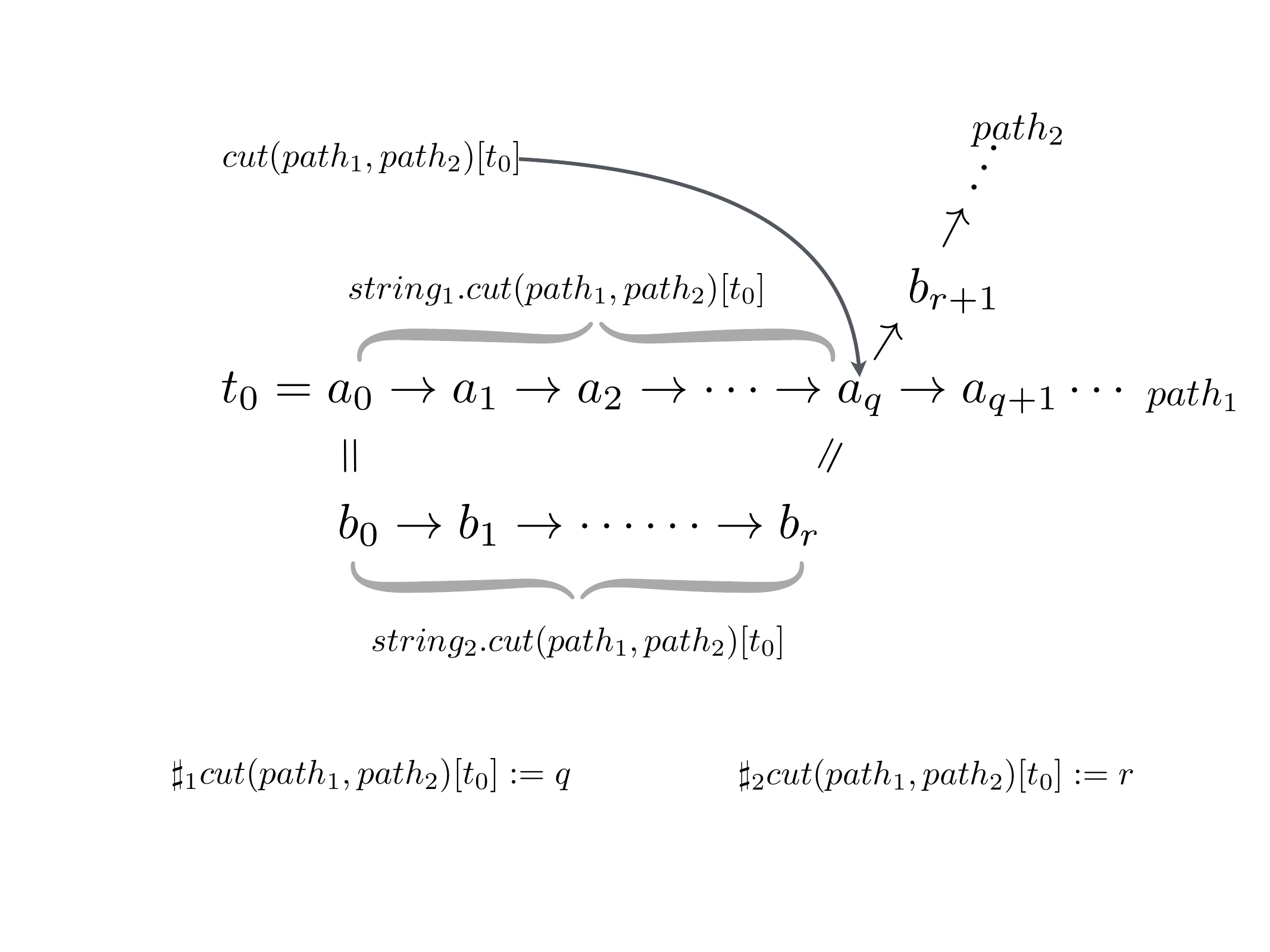}
\vspace*{-5ex}
        \caption{Definition of $cut(path_1,path_2)[t_0]$ for $path_1 = a_0,a_1,\cdots$
        and $path_2 = b_0,b_1,\cdots$}
        \label{fig-cut}
      \end{center}
\end{figure}


In applications, it may happen that a net $\N{} =(\GSvw, \GSrw, T, F)$ is known to be the net of a cycloid, but its parameters  $ \alpha, \beta, \gamma, \delta  $ are not known. This applies, for example, to systems of cooperating processes. In order to use the theory of cycloids, however, these parameters must be determined. To achieve this, graphical properties of  $\N{}$ could be utilised, such as the number of transitions, the length or structure of processes or the intersection of different paths. In the following, such structures are defined that are based on the relative length of transition sequences, here called paths. As a remarkable result, these investigations lead to a decision algorithm for cycloid isomorphism with given parameters.

In the following definition, the terms of a forward-cycle and backward-cycle from Definition \ref{def-f-b-cycle} of length $p =\frac{A}{gcd(\beta,\delta)}$  and 
$p' =\frac{A}{gcd(\alpha,\gamma)}$, respectively (Theorem \ref{th-f-b-cycle}) are taken up, but considered as infinite sequences of transitions, called paths.


\begin{definition} \label{def-path} 
Let $t_0$ be a transition of a cycloid net $\N{} =(\GSvw, \GSrw, T, F)$ (Definition \ref{cycloid} ) and $\tau = \{
t_i | i \in \Nat \}$ an  infinite sequence of transitions
 with $t_0$ as first element.
 \begin{itemize}
      \item [a)] If $t_{n+1} = {\postnbfw{(t_{n}}\;)}^{\ndot}$ for all $n\in \Nat$ then $\tau$ is denoted as    $fw[t_0]$ and called a \emph{forward path},
 	\item [b)] If $t_{n+1} = {}^{\ndot}{(\prenbfw{t_{n}})}$ for all $n\in \Nat$ then $\tau$ is denoted as $		\overleftarrow{fw}[t_0]$ and called a \emph{forward path in opposite direction},
 	\item [c)] If $t_{n+1} = (\postnbbw{t_{n}}\;)^{\ndot}$ for all $n\in \Nat$ then $\tau$ is denoted as 		$bw[t_0]$ and called a \emph{backward path},
 	\item [d)] If $t_{n+1} = {}^{\ndot}{(\prenbbw{t_{n}})}$ for all $n\in \Nat$ then $\tau$ is denoted as $		\overleftarrow{bw}[t_0]$ and called a  \emph{backward path in opposite 			direction}.
	\item [e)] For a path $pa = t_0, t_1,\cdots$ from $\{ fw[t_0],\overleftarrow{fw}[t_0],bw[t_0], \overleftarrow{bw}[t_0]\}$  we denote the element $t_{n}$ by $pa(n) := t_n$. 
	(To give the example from  the application of this definition in Corollary \ref{alpha-gamma-cor}~a)~1): 
	$pa =fw[t_0]$ and $n =(\alpha-\gamma)$ and 
	$fw[t_0](\alpha-\gamma)$ is  the $n$-th element of the sequence $fw[t_0]$ with $n = \alpha-\gamma$.)
 \end{itemize}
\end{definition} 


Figure \ref{fig-cut} shows symbolically how two such paths intersect in a transition, which is then called a \emph{cut}.

\begin{definition} \label{def-cut} 
Let be $t_0$ a transition of a cycloid net    $\N{} =(\GSvw, \GSrw, T, F)$ (Definition \ref{cycloid}) and  $path_1$ and $path_2$ two different paths as defined in Definition \ref{def-path}, both starting in $t_0$. Then we define $cut(path_1,path_2)[t_0]$ to be the first transition $a_q$ in 
 $path_1$, which is also a member $b_r$ of  $path_2$ 
 (see Figure \ref{fig-cut}).
Then we define
 $string_1.cut(path_1,path_2)[t_0]:= a_0,a_1,a_2,\cdots,a_q$
with length 
$\sharp_1 cut(path_1,path_2)[t_0]:= q$.
In a similar way, $string_2.cut(path_1,path_2)[t_0]:= b_0,b_1,b_2,\cdots,b_r$
with length 
$\sharp_2 cut(path_1,path_2)[t_0]:= r$.
The transition $cut(path_1,path_2)[t_0]$  exists and the values of $q$  and $r$ are finite since the sequences in Definition \ref{def-path} are cyclic by returning to $t_0$.

\end{definition}


The following theorem shows how to deduce from a cycloid net the parameters of a $\beta\delta$-irreducible cycloid, generating the cycloid net.
 This is done under assumption that the parameters $ \alpha, \beta, \gamma, \delta $ are not given, but 
 by the definition of a cycloid net (Definition \ref{cycloid})  the distinction between forward and backward places is given.
 
\begin{theorem} \label{th-beta-delta-red}
Let be  
$\N{} =(\GSvw, \GSrw, T, F)$ (Definition \ref{cycloid}) a cycloid net of $ \mathcal{C}_1$ 
with area $|T|=A $ and $t_0 \in T$ a transition.
Furthermore define 
\begin{itemize}
	\item[a)] $\alpha = \sharp_1 cut(fw ,bw )[t_0] $,
	\item[b)] $\beta=\delta =  \sharp_2 cut(fw ,bw )[t_0]$ and 
	\item[c)]  $\gamma = \frac{A}{\beta}-\alpha $.
\end{itemize}
Then  $ \mathcal{C}( \alpha, \beta, \gamma, \delta ) $ is $\beta\delta$-irreducible
and  cycloid-isomorphic to $ \mathcal{C}_1$.
Alternatively, $\gamma$ is the length of the forward cycle from $cut(fw ,bw ) [t_0]$ back to $t_{0}$, without counting  $cut(fw ,bw ) [t_0]$.
\end{theorem} 
\begin{proof} 
Given a cycloid $ \mathcal{C}_1$  we consider its 
$\beta\delta$-reduction  $  \mathcal{C} =\mathcal{C}( \alpha, \beta, \gamma, \delta ) $, which is cycloid equivalent by Theorem \ref{shearing}. To work with coordinates we also consider the fundamental parallelogram of $ \mathcal{C}$ within the Petri space.
By the symmetry of a cycloid, without any loss of generality, we can choose the origin as base transition: $t_{0} = t_{0,0}$ (Definition \ref{def-cut}). By Lemma \ref{t00-back} the backward output transition 
$\postnbbw{t_{0,0}} $ of $t_{0,0}$ is  $t_{\alpha,1-\beta}$
(for instance  $t_{10,-2}$ in the cycloid 
$ \mathcal{C}( 10,3,2,2) $ of Figure~\ref{10-322+12-122}).
If $\beta = 1$ transition $t_{\alpha,1-\beta}$ is lying on the $\xi-$axis. If $\beta > 1$ from transition $t_{\alpha,1-\beta}$ following the backward cycle we reach 
the $\xi-$axis in the transition $t_{\alpha,0}$ after passing $\beta - 1$ transitions and 
$\sharp_2 cut(fw,bw)[t_0] = \beta$ in both cases. By Theorem \ref{bd-irreducible} c) the forward cycle through  $t_{0,0}$ is right on the $\xi-$axis within the fundamental parallelogram. Therefore transition $t_{\alpha,0}$
is the first point where the backward cycle is meeting the forward cycle and 
$\sharp_1 cut(fw,bw)[t_0]=\alpha $.
By Theorem \ref{bd-irreducible} a) we have $\delta= \beta$ and $\gamma$ is computed using $A = \alpha \cdot \delta + \beta \cdot \gamma$. Alternatively, we can use the formula  (Theorem \ref{th-f-b-cycle}) for the length of a forward cycle
$p = \frac{A}{gcd(\beta,\delta)} = \frac{\alpha\cdot\beta+\beta\cdot\delta}{\beta} = \alpha+\delta $. Therefore the forward cycle is made up of parts from $t_{0}$ to $cut(fw,bw)[t_0] $ of length $\alpha$ and from $cut(fw,bw) [t_0]$ to $t_{0}$ of length $\gamma$.
\end{proof} 

\begin{figure}[htbp]
 \begin{center}
        \includegraphics [scale = 0.32]{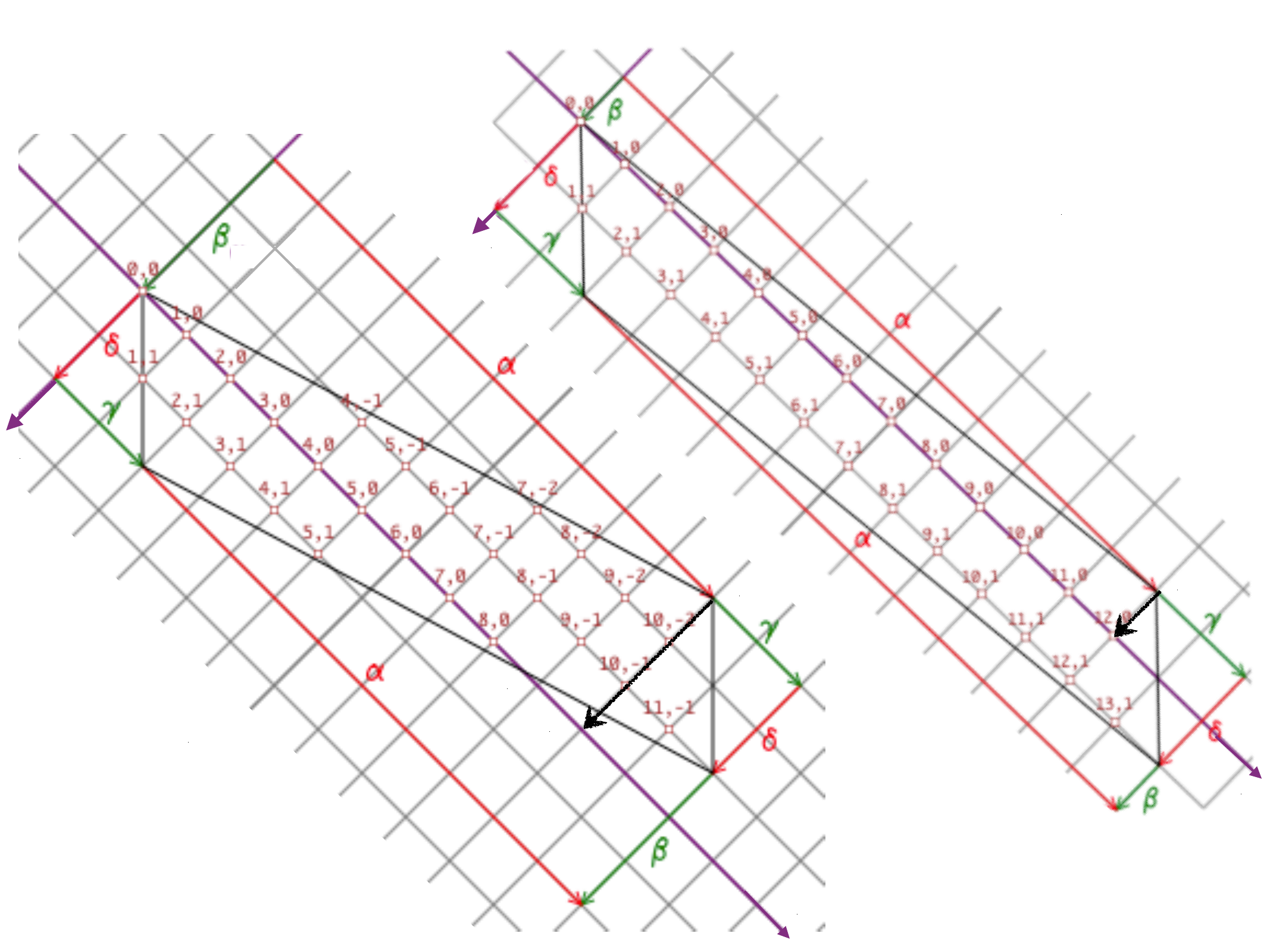}
        \caption{Fundamental parallelograms of  $ \mathcal{C}( 10,3,2,2 ) $ and   $ \mathcal{C}( 12,1,2,2 ) $. }
        \label{10-322+12-122}
      \end{center}
\end{figure}



\vspace*{-2ex}
As example   for the theorem consider the randomly selected transition $t_{0} =\textnormal{\textbf{t7}}$ in the cycloid net of Figure~\ref{c-5-3-2-6-red2}. 
It is generated by the RENEW-tool \cite{Moldt+23b} as cycloid $ \mathcal{C}( 5,3,2,6 ) $ with $A = 36$ transitions, but afterwards the structure of the underlying fundamental parallelogram has been erased, so that the parameters $5,3,2,6 $ are not longer visible in the graph of the net. By the result of Theorem \ref{th-beta-delta-red}  the parameters of its $\beta\delta$-reduction are constructed as follows.
%
%
The tool generating the cycloid net distinguishes forward and backward places by the letters $\textnormal{\textbf{f}}$ and $\textnormal{\textbf{b}}$ at the end of  their names, respectively.
The backward cycle through $t_{0} =\textnormal{\textbf{t7}}$ meets the forward cycle in $cut(fw,bw)[\textnormal{\textbf{t7}}]=
\textnormal{\textbf{t12}}$ with $string_2.cut(fw,bw)[\textnormal{\textbf{t7}}] = \textnormal{\textbf{t32}}$ $\textnormal{\textbf{t22}}$ $\textnormal{\textbf{t12}}$, resulting in 
$\sharp_2 cut(fw[\textnormal{\textbf{t7}}] ,bw[\textnormal{\textbf{t7}}] )  = \beta = \delta = 3$,
while the same for the forward cycle is  $string_1.cut(fw,bw)[\textnormal{\textbf{t7}}] =\textnormal{\textbf{t8}}$ $\textnormal{\textbf{t9}}$ $\textnormal{\textbf{t10}},\textnormal{\textbf{t11}},\textnormal{\textbf{t12}}$, resulting in 
$\sharp_1 cut(fw[\textnormal{\textbf{t7}}] ,bw[\textnormal{\textbf{t7}}] )= \alpha = 5$. 
These paths are shown as thick lines in Figure \ref{c-5-3-2-6-red2}.

Since $\gamma = \frac{36}{3}-5=7 $ 
the $\beta\delta$-irreducible cycloid is  $ \mathcal{C}( 5,3,7,3 ) $. An alternative method to determine $\gamma$ is by counting the length of the transition sequence  \textbf{t1} \textbf{t2} \textbf{t3} \textbf{t4} \textbf{t5} \textbf{t6} \textbf{t7} in the
forward cycle from $\textnormal{\textbf{t12}}$ to $\textnormal{\textbf{t7}}$ of length $7$.

The cycloid $ \mathcal{C}( 12,1,2,2) $ of Figure~\ref{10-322+12-122} gives an example where the backward cycle from point $P$ meets the $\xi$-axis \emph{within} the fundamental parallelogram.
This is very important for the validity of the proof and is holding also for the $\beta\delta$-irreducible cycloid $ \mathcal{C}( 12,1,14,1)$, which was too badly shaped for this picture.
A counterexample is the cycloid
$ \mathcal{C}( 10,3,2,2) $ in the same figure which also $\beta\delta$-reduces to $\mathcal{C}( 12,1,14,1)$.

Since the   graphical structures  of two  isomorphic cycloids are equal, there is a methodically simple algorithm for deciding this isomorphism. It is based on two
$\beta\delta-$reductions of the cycloids which have the behaviour of Euclid's algorithm. Using the modification as described after Table \ref{code}, we obtain an algorithm with a time complexity of $T(n) = \mathcal{O}(log_2 \, n)$, where $n= max\{\beta_1,\delta_1,\beta_2,\delta_2\}$ for the cycloids $ \mathcal{C}_1$ and $ \mathcal{C}_2$ in the following Corollary \ref{th2-cyc-iso}.

\begin{corollary} \label{th2-cyc-iso}
Two  cycloids $ \mathcal{C}_i =\mathcal{C}_i( \alpha_i, \beta_i, \gamma_i, \delta_i ), \; i \in \{ 1,2 \}$ are cycloid isomorphic (Definition \ref{def-cyc-iso}) if and only if they are 
$\beta\delta$-reduction equivalent (Definition \ref{reduction}):
$\mathcal{C}_1 \simeq_{\text{cyc}} \mathcal{C}_2 \Leftrightarrow 
\mathcal{C}_1\; \simeq_{\beta\delta} \;\mathcal{C}_2$. 
\end{corollary} 
\begin{proof} If $\mathcal{C}_1$ and $\mathcal{C}_2$ are cycloid isomorphic
 their cycloid nets $\N{1}$ and $\N{2}$ are 
cycloid isomorphic. The isomorphism maps subsequences of forward and backward cycles to sequences of the same type and length. Also the image of their intersections are the intersections of their images.
Therefore they have the same $\beta\delta$-reduction by 
Theorem \ref{th-beta-delta-red}. The reverse statement of the corollary follows from  Theorem \ref{shearing} since reduction steps as well their reverses preserve cycloid isomorphisms.
\end{proof}  

To give an example, we compare the cycloids  $ \mathcal{C}_1=\mathcal{C}_1( 2,3,1,4 ) $ (the same as in Figure \ref{2-3-1-4}) with $ \mathcal{C}_2=\mathcal{C}_2( 2,5,1,3 ) $. They have the same area $A = 11$ 
(Definition \ref{cycloid}) and process length $p = 11$ (Theorem \ref{th-f-b-cycle}), but are not cycloid isomorphic since their $\beta\delta$-reductions are different, 
namely $ \mathcal{C}^{*}_1( 8,1,3,1 ) $ and $ \mathcal{C}^{*}_2( 7,1,4,1 ) $, respectively.

While Theorem \ref{th-beta-delta-red} represents the final result of a $\beta\delta$-reduction as a graphical property, this is also possible for the entire reduction. Each cycloid of such a reduction can be assigned to a transition of the cycloid net, whereby the path distances describe the respective reduction step.

The following lemma prepares the proof of Theorem \ref{th-all-param}. It states that $(\alpha-\gamma,0)$ is equivalent (Definition \ref{cycloid})  to  $(0,\beta+\delta)$ for a cycloid  $ \mathcal{C}( \alpha, \beta, \gamma, \delta ) $ with $\alpha > \gamma$ and a similar result if $\gamma >\alpha$.

\begin{lemma} \label{alpha-gamma}
For a cycloid  $ \mathcal{C}( \alpha, \beta, \gamma, \delta ) $ we have
\begin{itemize}
       \item [a)] $\begin{pmatrix} \alpha-\gamma \\ 0  \end{pmatrix} \equiv \begin{pmatrix} 0\\ \beta + \delta  \end{pmatrix} $ \;\; if $\alpha > \gamma$ and this does not hold for a smaller value of $\beta+\delta$,
        \item [b)] $\begin{pmatrix} \gamma - \alpha\\ 0  \end{pmatrix} \equiv \begin{pmatrix} 0\\ -(\beta+\delta) 
        \end{pmatrix} $  \;\; if $\gamma > \alpha$ 
        and this does not hold for a smaller value of $\beta+\delta$.
        \end{itemize}     
        
\end{lemma} 
\begin{proof}
By Theorem \ref{parameter} we derive $\begin{pmatrix} \alpha-\gamma \\ 0  \end{pmatrix} \equiv \begin{pmatrix} 0\\ \beta + \delta -x \end{pmatrix} $  
for $0  \leq x <\beta+\delta$ and conclude $x=0$.
For  $\vec{v} =  \begin{pmatrix} \alpha-\gamma \\ 0  \end{pmatrix} -\begin{pmatrix} 0\\ \beta + \delta -x \end{pmatrix}   
 =  \begin{pmatrix} \alpha-\gamma \\ -\beta-\delta+x \end{pmatrix}$  we have to prove that
 $\pi(\vec{v} ) =\frac{1}{A} \cdot \begin{pmatrix} \delta & -\gamma \\ \beta & \alpha \end{pmatrix}
\cdot \vec{v}$ has integer values:
\begin{align*}
\pi(\vec{v} ) &=\frac{1}{A} \cdot \begin{pmatrix} \delta & -\gamma \\ \beta & \alpha \end{pmatrix}\cdot \begin{pmatrix} \alpha-\gamma \\ -\beta-\delta+x \end{pmatrix} = 
\frac{1}{A} \cdot
\begin{pmatrix} \delta\cdot\alpha-\delta\cdot\gamma +\gamma\cdot\beta+\gamma\cdot\delta-\gamma\cdot x\\ 
\beta\cdot\alpha-\beta–\gamma-\alpha\cdot\beta-\alpha\cdot\delta+\alpha\cdot x \end{pmatrix} \\
&= \frac{1}{A}\cdot \begin{pmatrix} A-\gamma\cdot x\\ 
 -A+\alpha\cdot x\end{pmatrix} =
 \begin{pmatrix} 1-\frac{ \gamma\cdot x}{A} \\ 
 -1+\frac{ \alpha\cdot x}{A}\end{pmatrix}.
 \end{align*}   
If $\gamma  \leq \alpha$ by the condition $x< \beta+\delta$ we conclude 
$\gamma\cdot x < \gamma\cdot\beta + \gamma\cdot\delta  \leq $
$ \gamma\cdot\beta + \alpha\cdot\delta  = A$ and $\frac{\gamma\cdot x}{A} <1 $.
If $\alpha  \leq \gamma$ by the condition $x< \beta+\delta$ we conclude 
$\alpha\cdot x < \alpha\cdot\beta + \alpha\cdot\delta  \leq $
$\gamma\cdot\beta + \alpha\cdot\delta  = A$ and $\frac{\alpha\cdot x}{A} <1 $.
Hence, for $\pi(\vec{v}) $ to have integer values $x=0$ is necessary and a) is proved. 
Case b) is similar with $\pi(\vec{v}) = \begin{pmatrix} -1+\frac{ \gamma\cdot x}{A} \\ 
 1-\frac{ \alpha\cdot x}{A}\end{pmatrix} $ instead.
\end{proof}

The relations of Lemma \ref{alpha-gamma} hold for the Petri space. By the cycloid folding they transform to path properties in a cycloid.

\begin{corollary} \label{alpha-gamma-cor}
Let be $t_0$ a transition of a cycloid  $ \mathcal{C}( \alpha, \beta, \gamma, \delta ) $ and    $\N{} =(\GSvw, \GSrw, T, F)$ (Definition \ref{cycloid} ) its cycloid net.
\begin{itemize}
       \item [a)] If $\alpha>\gamma$ then
       		\begin{itemize}
       			\item [1)] $fw[t_0](\alpha-\gamma) = bw[t_0](\beta+\delta)$ (see Definition \ref{def-path} e) and
       			\item [2)] $\sharp_1cut(fw,bw)[t_0] = \alpha-\gamma$ and 
					$\sharp_2 cut(fw,bw)[t_0] = \beta+\delta$.
		\end{itemize}     
        \item [b)]  If $\gamma>\alpha$ then
       	       \begin{itemize}
       			\item [1)] $fw[t_0](\gamma-\alpha) = \overleftarrow{bw}[t_0](\beta+\delta)$ and
       			\item [2)] $\sharp_1cut(fw,\overleftarrow{bw})[t_0] = \gamma-\alpha$ and $\sharp_2 cut(fw,\overleftarrow{bw})[t_0] = \beta+						\delta$ .
		\end{itemize}     
\end{itemize}     
\end{corollary}

\begin{proof} 
By Lemma \ref{alpha-gamma} a) going $\alpha-\gamma$ steps in $\xi$-direction in the Petri space a transition is reached which is equivalent to a transition going $\beta+\delta$ steps in $\eta$-direction. By the folding into the fundamental parallelogram equivalent transition coincide. Since $\beta+\gamma$ is minimal in Lemma \ref{alpha-gamma} the cut point $cut(fw,bw)[t_0]$ exists and $string_1$ and $string_2$ (Definition \ref{def-cut}) have the length $\alpha-\gamma$ and $\beta+\delta$, respectively.
Case b) is similar, but on the backward path one must walk  in the opposite direction.
\end{proof}

By the next result, Theorem \ref{th-beta-delta-red}  is continued  to relate in a reversed direction \emph{all} values $\alpha_r, \beta_r,\gamma_r,\delta_r$ of cycloids $\mathcal{C}_r =  \mathcal{C}_r( \alpha_r, \beta_r, \gamma_r, \delta_r )\;\; (1  \leq r  \leq u)$ in a 
$\beta\delta$-reduction chain  to intersections of forward and backward cycles in all their cycloid nets. They are given as labels $lab(t_r)$ of a sequence $t_1, t_2, \cdots,t_r,\cdots,t_u$, as a continuation of $t_0$ from Theorem \ref{th-beta-delta-red} .

\begin{definition} \label{def-all-param}
Let be  $ \mathcal{C}_1 =   \mathcal{C}_1( \alpha_1, \beta_1, \gamma_1, \delta_1 ) $ a 
$\beta\delta$-irreducible
cycloid, $\N{} =(\GSvw, \GSrw, T, F)$ its  cycloid net (Definition \ref{cycloid}) and $t_1 \in T$ a transition. 
Furthermore let $t_1, t_2, \cdots,t_u \in T^*$ a  sequence of transitions, labeled by cycloids
$\mathcal{C}_r =  \mathcal{C}_r( \alpha_r, \beta_r, \gamma_r, \delta_r )\;\; (1  \leq r  \leq u)$, as follows:

\begin{itemize}
     \item [a)]  $lab(t_1) =  \mathcal{C}_r( \alpha_1, \beta_1, \gamma_1, \delta_1)$,
	\item [b)]  if $\alpha_r > \gamma_r$  then $t_{r+1} := cut(fw,bw)[t_r]$ and  $lab(t_{r+1})= 							\mathcal{C}_{r+1}( \alpha_{r+1},  \beta_{r+1}, \gamma_{r+1}, \delta _{r+1}) $ 
			with
			$\alpha_{r+1}=\sharp_1cut(fw,bw)[t_r]$ and 
			$\beta_{r+1}=\sharp_2cut(fw,bw)[t_r]$ and
			$\gamma_{r+1}	=\gamma_{r}$ and 
			$\delta _{r+1}=\delta _{r}$,
	\item [c)] if $\gamma_r > \alpha_r$  then $t_{r+1} := cut(fw,\overleftarrow{bw})[t_r]$ and  $lab(t_{r+1})= 							\mathcal{C}_{r+1}( \alpha_{r+1},  \beta_{r+1}, \gamma_{r+1}, \delta _{r+1}) $ 				with
			$\alpha_{r+1}=\alpha_r$  and
			$\beta_{r+1}=\beta_r$ and
			$\gamma_{r+1} = \sharp_1cut(fw,\overleftarrow{bw})[t_r]$ and
			$\delta_{r+1}=\sharp_2cut(fw,\overleftarrow{bw})[t_r]$,					
	\item [d)] if $\gamma_r = \alpha_r$  then $r := u$  (end of the sequence).
\end{itemize} 
\end{definition}
    
\begin{theorem} \label{th-all-param}
Let be $t_1, t_2, \cdots,t_u \in T^*$ and $lab(t_1)=\mathcal{C}_1 ,lab( t_2)=\mathcal{C}_2 , \cdots,lab(t_u)=\mathcal{C}_u $ the sequences as defined in Definition \ref{def-all-param}.
\begin{itemize} 
       \item [a)]  
       		$\mathcal{C}_1   \xrightarrow[\text{\tiny{}}]{\lambda_1}  \mathcal{C}_2   \xrightarrow[\text{\tiny{}}]			{\lambda_2}  \cdots   \xrightarrow[\text{\tiny{}}]{\lambda_n} \mathcal{C}_{u}$ 
      		$(\lambda_i \in \{   \alpha_i, \gamma_i\}  )$ is a  $ \alpha\gamma $-reduction chain and  
		$ \mathcal{C}_u$ is $\alpha\gamma$-irreducible,
        \item [b)] $\mathcal{C}_u   \xrightarrow[\text{\tiny{}}]{\lambda_1}  \mathcal{C}_{u-1}   							\xrightarrow[\text{\tiny{}}]{\lambda_2}  \cdots   \xrightarrow[\text{\tiny{}}]{\lambda_n} 						\mathcal{C}_{1}$   $(\lambda_i \in \{   \beta_i, \delta_i\}  )$ is a  
        		$\beta\delta$-\emph{reduction chain}.
\end{itemize}     
\end{theorem} 

\begin{shortproof} 
\begin{itemize}
\item[a)] 
If  $\alpha_r > \gamma_r $ in $\mathcal{C}_r =  \mathcal{C}_r( \alpha_r, \beta_r, \gamma_r, \delta_r )\;\; (1  \leq r  < u)$ then $\lambda_r = \alpha_r$ and by Definition \ref{def-all-param} b) and Corollary \ref{alpha-gamma-cor} a2)
$\alpha_{r+1}=\sharp_1cut(fw,bw)[t_r] = \alpha_r-\gamma_r$ as well 
$\beta_{r+1}=\sharp_2cut(fw,bw)[t_r]= \beta+\delta$. Since also 
$\gamma_{r+1}	=\gamma_{r}$ and $\delta _{r+1}=\delta _{r}$ we obtain
$\mathcal{C}_r  \xrightarrow[\text{\tiny{}}]{\alpha_r} \mathcal{C}_{r+1}$. 
If $\gamma_r > \alpha_r $ by a similar conclusion $\mathcal{C}_r  \xrightarrow[\text{\tiny{}}]{\gamma_r} \mathcal{C}_{r+1}$ is obtained and the sequence of Theorem \ref{th-all-param} is a $ \alpha\gamma $-reduction chain. $ \mathcal{C}_u$ is $\alpha\gamma$-irreducible since  $\gamma_u = \alpha_u$ by step d) of Definition \ref{def-all-param}.
\item[b)] Since  $\mathcal{C}_r  \xrightarrow[\text{\tiny{}}]{\alpha_r} \mathcal{C}_{r+1}\Leftrightarrow
\mathcal{C}_{r+1}  \xrightarrow[\text{\tiny{}}]{\beta_{r+1}} \mathcal{C}_{r}$ and
$\mathcal{C}_r  \xrightarrow[\text{\tiny{}}]{\gamma_r} \mathcal{C}_{r+1}\Leftrightarrow
\mathcal{C}_{r+1}  \xrightarrow[\text{\tiny{}}]{\delta_{r+1}} \mathcal{C}_{r}$
the $ \alpha\gamma $-reduction chain in a) read from back to front is a $ \beta\delta $-reduction chain.
\end{itemize}
\end{shortproof}

Table \ref{application-of-theorem} illustrates the above theorem using example  from Figure \ref{c-5-3-2-6-red2}. As shown in the example following Theorem \ref{th-beta-delta-red},  $ \mathcal{C}_1( 5,3,7,3) $ is the corresponding $\beta\delta$-irreducible cycloid (row 2 in Table \ref{application-of-theorem}).  Since $\gamma_1 = 7 > \alpha_1 = 5$ a $\gamma$-reduction step leads to 
$ \mathcal{C}_2( 5,3,2,6) $ (row 3 in Table 4). 
The formation of C2 is derived from the structure of the graph as follows.
If  the label $lab(t_1)=lab(\textbf{t12})$  is $\mathcal{C}_1 $ 
then 
\begin{itemize}
          \item $cut(fw,\overleftarrow{bw})[\textbf{t12}] = \textbf{t2}$, 
	 \item $string_1.cut(fw,\overleftarrow{bw})[\textbf{t12}] = $ \textbf{t1} \textbf{t2}, \;\;\;\;\;\;\;\;\;  \;  \;\;\;  \;\;\; \;\;\;  \;\;\;     (dotted line in Figure 			\ref{c-5-3-2-6-red2})
 	   \item $\sharp_1.cut(fw,\overleftarrow{bw})[\textbf{t12}] = 2$, 
	    \item $string_2.cut(fw,\overleftarrow{bw})[\textbf{t12}] = $  \textbf{t22} \textbf{t32} \textbf{t7} \textbf{t17} 				\textbf{t27} \textbf{t2}, \;\;\;   (first bold then dotted line in Figure \ref{c-5-3-2-6-red2})
 	   \item $\sharp_2.cut(fw,\overleftarrow{bw})[\textbf{t12}] = 6$.
\end{itemize}    
These values $2$ and $6$ provide the values $\gamma_2$ and $\delta_2$ of $lab(t_2) = lab(\textbf{t2}) = 
 \mathcal{C}_2( 5,3,\textbf{2},\textbf{6} )$ 	in the next row. 
 The entries in the following lines are corresponding, whereby for $\lambda = \alpha$  the fifth column is used instead of the sixth. The final cycloid  $ \mathcal{C}_5( 1,15,1,21) $ is $\alpha\delta$-irreducible.

\begin{table}[htbb]
 \begin{center}
\begin{tabular}{|c||c|c|c|c|c|}
\hline
   					 &  
                               &  
   $t_r$ in                                    & 
   \textbf{tx}                                               & 
   $cut(fw,bw)[\textbf{tx}] =  \textbf{ty} $                      &
    $cut(fw,\overleftarrow{bw})[\textbf{tx} ]=  \textbf{ty}$\\
  $  \mathcal{C}_{r} = lab(t_r)$ 						 &
   	$\lambda$ 						&
   Def. 			&
   in 					&
   $string_1(\sim)$ and $\sharp_1(\sim)$			&
   $string_1(\sim)$ and $\sharp_1(\sim)$\\  
						 &
   						&
  \ref{def-all-param}			&
   $\N{}$ 					&
   $string_2(\sim)$ and  $\sharp_2(\sim)$			&
   $string_2(\sim)$ and  $\sharp_2(\sim)$\\  \hline \hline 
    $ \mathcal{C}_1( 5,3,7,3 )$ 			 			&
  $ \gamma	$						&
  $t_1$							&
  \textbf{t12}						&
  							&
 $cut(fw,\overleftarrow{bw})[\textbf{t12}] = \textbf{t2}$	  \\  
     $\gamma_1 > \alpha_1$			 			&
   							&
  							&
   							&
  							&
  \textbf{t1} \textbf{t2} and  2							\\  
  			 			&
   							&
  							&
   							&
  							&
  \textbf{t22} \textbf{t32} \textbf{t7} \textbf{t17} \textbf{t27} \textbf{t2} and  6							\\  \hline 
      $ \mathcal{C}_2( 5,3,\textbf{2},\textbf{6} )$ 			 			&
  $ \alpha	$						&
  $t_2$							&
  \textbf{t2}						&
  $cut(fw,bw)[\textbf{t2}] = \textbf{t5}$&	  \\  
     $\alpha_2 > \gamma_2$			 			&
   							&
  							&
   							&
  \textbf{t3} \textbf{t4} \textbf{t5} and  3 &							\\  
  			 			&
   							&
  							&
   							&
  \textbf{t27} \textbf{t17} $\cdots$\textbf{t15} \textbf{t5}  and  9	&						\\  \hline 
      $ \mathcal{C}_3( \textbf{3},\textbf{9},2,6)$ 			 			&
  $ \alpha	$						&
  $t_3$							&
  \textbf{t5}						&
  $cut(fw,bw)[\textbf{t5}] = \textbf{t6}$&	  \\  
     $\alpha_3 > \gamma_3$			 			&
   							&
  							&
   							&
  \textbf{t6} and 1  &							\\  
  			 			&
   							&
  							&
   							&
  \textbf{t30} \textbf{t20} $\cdots$\textbf{t16} \textbf{t6}  and  15	&						\\  \hline 
  $ \mathcal{C}_4( \textbf{1},\textbf{15},2,6 )$ 			 			&
  $ \gamma	$						&
  $t_4$							&
  \textbf{t6}						&
  							&
 $cut(fw,\overleftarrow{bw})[\textbf{t6}] = \textbf{t7}$	  \\  
     $\gamma_4 > \alpha_4$			 			&
   							&
  							&
   							&
  							&
  \textbf{t7}  and  1							\\  
  			 			&
   							&
  							&
   							&
  							&
  \textbf{t16} \textbf{t26} $\cdots$ \textbf{t32} \textbf{t7}  and  21						\\  \hline 
   $ \mathcal{C}_5( 1,15,\textbf{1},\textbf{21} )$ 			 			&
   							&
 $ t_5$							&
    \textbf{t7}							&
  							&
  							\\  \hline 

\hline
\end{tabular}
 \end{center}
 \caption{Application of Theorem \ref{th-all-param} to the cycloid net $\N{}$  of Figure \ref{c-5-3-2-6-red2}}
\label{application-of-theorem}
       \end{table}

\begin{figure}[htbp]
	\begin{center}
		\includegraphics [scale = 0.32]{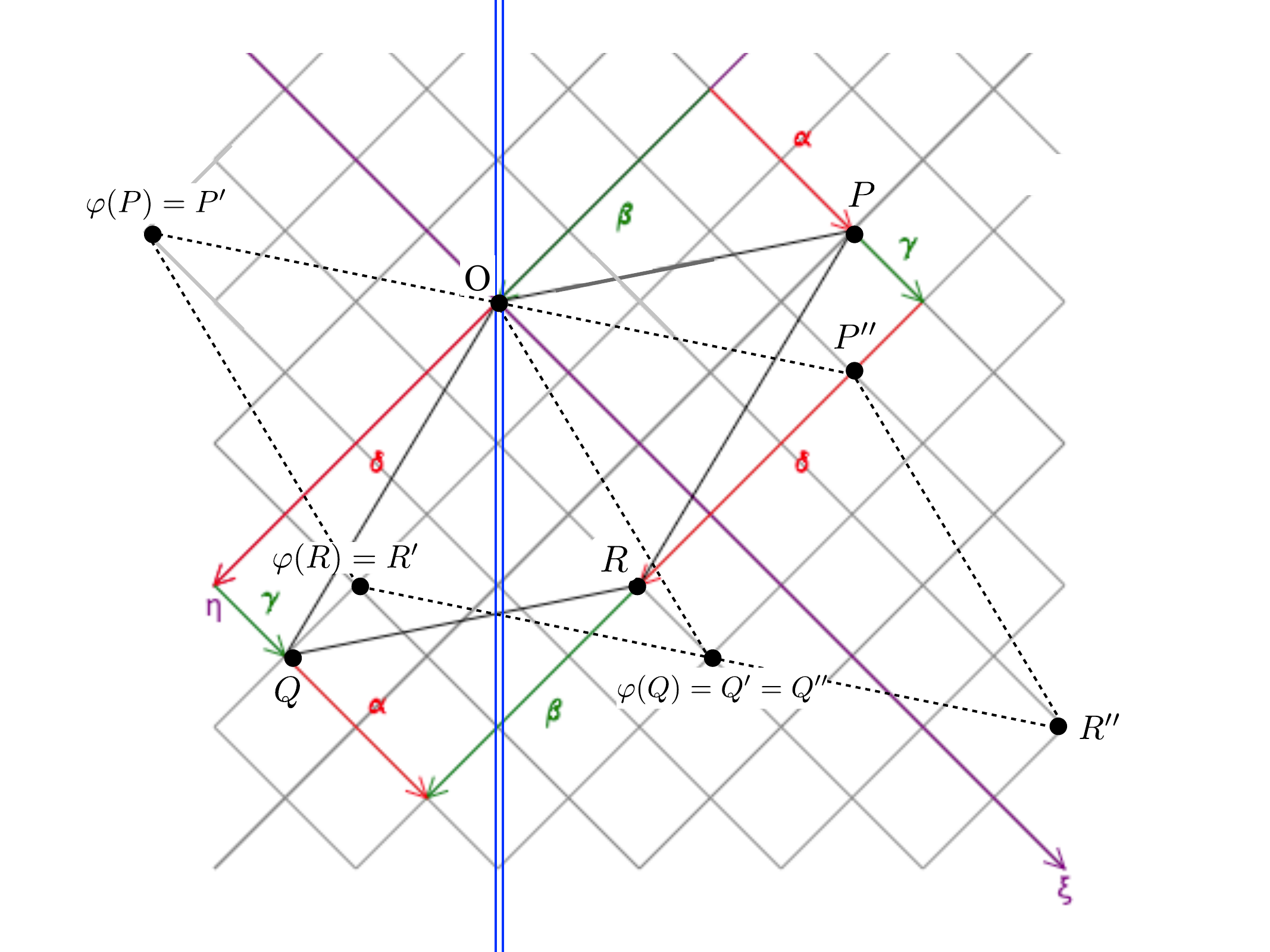}
		\caption{Fundamental parallelogram of  $ \mathcal{C}( 2,3,1,4 ) $ and its symmetric image $ \mathcal{C''}( 3,2,4,1) $}
		\label{2-3-1-4}
	\end{center}
\end{figure}


The above proof of Theorem \ref{th-all-param} b) shows that   results obtained in this article for $\beta\delta$-reductions can be reproduced for $\alpha\gamma$-reductions in a similar way.	When replacing the rules $R_\beta, R_\delta$ by $R_\alpha, R_\gamma$ in the algorithm from Theorem \ref{bd-irreducible}  a $\alpha\gamma$-irreducible cycloid  $ \mathcal{C}( \alpha, \beta, \gamma, \delta ) $ with 
$\alpha = \gamma$ is obtained. 
From the example of a $\beta\delta$-reduction, as given after Definition \ref{reduction}, a $\alpha\gamma$-reduction can be obtained by interchanging the first an the second parameter, as well as the third and the forth one. The inscriptions on the arcs have to be changed from $\delta$ to $\gamma$ and from $\beta$ to $\alpha$.
Although both reduction methods are symmetric, this article focusses on $\beta\delta$-reductions, as it is closer to the concept of regular cycloids (Definition \ref{regular}), which are closely related to circular traffic systems \cite{Valk-2020} and are defined by $\beta$ dividing $\delta$. For some $k\in \Natp$ a regular cycloid  can be written as $ \mathcal{C}_1( \alpha, \beta, \gamma, k\cdot\beta ) $  which can be reduced by $\delta$-reductions to the cycloid isomorphic cycloid
$ \mathcal{C}_2( \alpha, \beta, \gamma+ (k-1)\cdot\alpha, \beta ) $. The behaviour of latter can be compared with the canonical regular cycloid  $ \mathcal{C}_3( \alpha, \beta, \beta, \beta ) $. This is beyond the scope of this article, but it explains why we preferred  $\beta\delta$-reductions over  $\alpha\gamma$-reductions.

Formally, both reduction methods are related by the transformation of Theorem \ref{shear} e) producing the symmetric cycloid  $ \mathcal{C}_2 = \mathcal{C}( \beta,\alpha,\delta, \gamma ) $ of a given cycloid  
$\mathcal{C}_1 = \mathcal{C}( \alpha, \beta, \gamma, \delta ) $.  $ \mathcal{C}_2 $ is net isomorphic 
(Theorem \ref{shear}) e) to $\mathcal{C}_1 $, but not cycloid isomorphic (Definition \ref{def-cyc-iso}), as can be seen in the counterexample below. 
We begin with a consideration that is familiar from the geometry of Euclidean spaces, namely the mapping  $\varphi(\xi,\eta) = (\eta,\xi)$. The resulting diagram, however, is not the fundamental parallelogram of a cycloid.
The image  of $P$ which is $\varphi(\alpha,-\beta) = (-\beta,\alpha)$ has negative $\xi$-coordinates which is impossible for fundamental parallelograms (see the following example).
  Therefore we  apply the mapping $\rho_{\mathcal{C}_2}$ with respect to $\mathcal{C}_2$ from Theorem \ref{xy-to-FP} and obtain the final form of  the fundamental parallelogram of $\mathcal{C''}$.

  Figure \ref{2-3-1-4} shows the fundamental parallelogram of  $ \mathcal{C} =\mathcal{C}( 2,3,1,4) $ with vertices $O,Q,R$ and $P$ as an example of the generation of the isomorphism to the symmetric cycloid  $ \mathcal{C}''=\mathcal{C}( 3,2,4,1) $ (Theorem \ref{shear}).
  The doubled line  through $O$ is the line of reflection of $\varphi(\xi,\eta) = (\eta,\xi)$. We obtain a parallelogram with vertices $\varphi(O) = O, \varphi(Q) = Q', \varphi(R) = R'$ and 
  $\varphi(P) = P'$, which is not a fundamental parallelogram of a cycloid. The application 
  $\rho_{\mathcal{C}''}(\xi,\eta)= (\eta+3,\xi-2)$ transforms this parallelogram into the fundamental parallelogram of $\mathcal{C}''=\mathcal{C}( 3,2,4,1) $ with vertices $O,Q'',R''$ and $P''$. By Corollary \ref{th2-cyc-iso} $\mathcal{C}$ and $\mathcal{C''}$ are not cycloid isomorphic as their $\beta\delta$-irreducible equivalents  $ \mathcal{C}(8,1,3,1) $ and 
    $ \mathcal{C}(7,1,4,1) $ are different.
  
\section{Conclusion}\label{sec-conclusion}
	
The theory of cycloids is extended by the technique of reduction in the style of term rewriting systems. Reduction steps have a geometrical interpretation as shear mappings. New results are derived which are the base of numerous new algorithms for computing cycloid properties. A linear algebra method, called \emph{cycloid algebra},
allows for easier computation of properties like the minimal length of cycles and equivalence of transitions in the Petri space. The parameters $ \alpha, \beta, \gamma$ and $ \delta  $ of the irreducible cycloid and thus of the entire reduction chain can be calculated solely from the properties of paths in the net representation of a cycloid.
This is interpreted as cycloid synthesis and leads to an efficient method for a decision procedure for cycloid isomorphism. 

\section{Acknowledgements}\label{sec-acknowledgements}

We would like to thank the anonymous reviewers for their careful review, which has led to many improvements in terms of comprehensibility and formal clarity.

\bibliographystyle{fundam}
\bibliography{citations-rv}
\end{document}